\documentclass[a4paper,12pt]{book}

\usepackage[utf8]{inputenc}
\usepackage[british]{babel}
\usepackage{amsthm, amsmath, amssymb}
\usepackage{tabularx}
\usepackage{graphicx}
\usepackage{enumerate}
\usepackage{textcomp}
\usepackage{tikz}
\usetikzlibrary{shapes,arrows}
\usepackage{varioref} 
\usepackage[font=small,format=hang,labelfont={sf}]{caption}
\usepackage{xspace}

\allowdisplaybreaks[4]

\tikzstyle{block} = [draw, rectangle,  minimum height=2em, minimum width=2em]
\tikzstyle{pall} = [draw,  circle, fill=black, scale=0.3]
\tikzstyle{output} = [node distance=2cm]
\tikzstyle{sum} = [draw, circle, minimum height=2em]

\newtheorem{thm}{Theorem}[section]
\newtheorem{cor}[thm]{Corollary}
\newtheorem{lem}[thm]{Lemma}
\newtheorem{prop}[thm]{Proposition}

\theoremstyle{definition}
\newtheorem{defi}[thm]{Definition}
\newtheorem{rem}[thm]{Remark}
\newtheorem{ex}[thm]{Example}
\newtheorem{alg}[thm]{Algorithm}

\newcommand{\pint}[1]{\left\lfloor#1\right\rfloor}
\newcommand{\Pint}[1]{\left\lceil#1\right\rceil}
\newcommand{\pt}[1]{\left(#1\right)}
\newcommand{\N}{\mathbb{N}}
\newcommand{\F}{\mathbb{F}_q}
\newcommand{\C}{\mathcal{C}}
\newcommand{\cw}{\boldsymbol{c}}
\newcommand{\e}{\boldsymbol{e}}
\newcommand{\rw}{\boldsymbol{r}}
\newcommand{\vct}[1]{\boldsymbol{#1}}
\newcommand{\K}{\mathbb{K}}
\newcommand{\RS}{RS(n,d,\alpha)}
\newcommand{\kn}{\mathbb{K}^n}
\newcommand{\tl}[1]{\widetilde{#1}}
\newcommand{\dist}{\textnormal{dist}}
\newcommand{\wt}{\textnormal{wt}}
\newcommand{\myd}{\textsc{d}}
\newcommand{\rk}{\textnormal{rk}}

\newcommand{\Mod}[1]{\;\left(\textnormal{mod}\,#1\right)}
\newcommand{\Deg}[1]{\textnormal{deg}\left(#1\right)}
\newcommand{\fPGZ}{\textit{f\,PGZ}\xspace}
\newcommand{\PGZ}{\textit{PGZ}\xspace}
\newcommand{\BM}{\textit{BM}\xspace}
\newcommand{\pBM}{\textit{pBM}\xspace}
\newcommand{\pPGZ}{\textit{pPGZ}\xspace}
\newcommand{\BP}{\textit{BP}\xspace}
\newcommand{\refbm}[1]{\BM.\ref{#1}}
\newcommand{\refpgz}[1]{\PGZ.\ref{#1}}
\newcommand{\reffpgz}[1]{\fPGZ.\ref{#1}}
\newcommand{\myhat}[1]{\mathring{#1}}

\newenvironment{mymatrix}{\left(\!\!\begin{array}{lllllll}}{\end{array}\!\!\right)}

\begin{document}

\frontmatter

\begin{titlepage}
\begin{center}
\includegraphics[width=3cm]{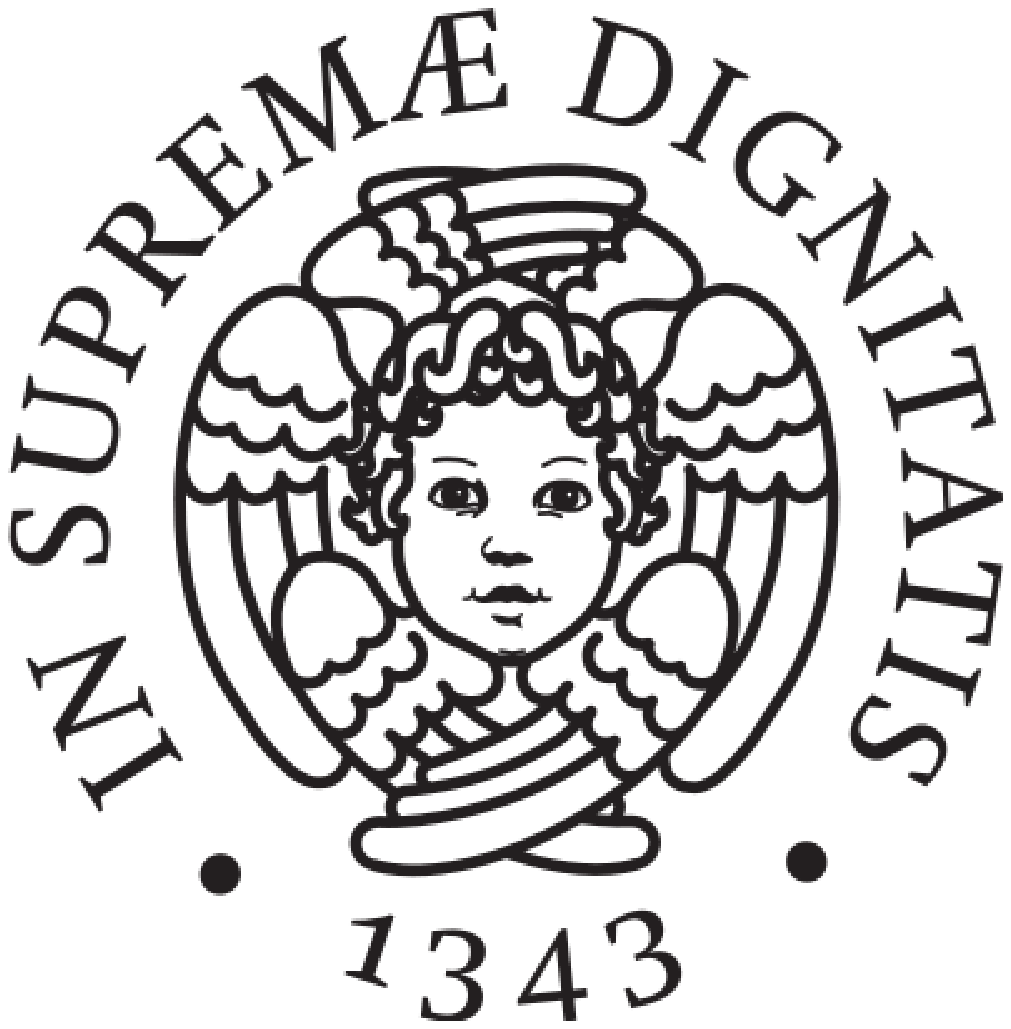}\\[0,1cm]
\textsc{\Large Università degli Studi di Pisa}\\
\vspace{0.2cm}
\textsc{\large Facoltà di Scienze Matematiche, Fisiche e Naturali}\\ 
\textsc{\large{Corso di Laurea Magistrale in Matematica}}\\
\vspace{1cm}
{\huge\textbf{Improved Decoding Algorithms for Reed-Solomon Codes}}\\[0.2cm]
\textsc{\large Tesi di Laurea Magistrale}\\
{\large 3 dicembre 2012}\\[2cm]
{\large \sc{Candidato:}}\\[0.1cm]  
{\Large \bf{Irene Giacomelli}}\\[1,5cm]
{\large \sc{Relatori:}}\\[0.1cm]
{\large \textbf{Prof.ssa Patrizia Gianni\hspace*{0,8cm} Prof. Barry Trager}}\\[1,5cm]
{\large \sc{Controrelatore:}}\\[0.1cm]  
{\large \textbf{Prof. Massimo Caboara}}\\[1,8cm]
\textsc{\Large Anno Accademico 2011/2012}
\end{center}
\end{titlepage}

\thispagestyle{empty}
%

\tableofcontents

\chapter{Preface}
In coding theory, Reed-Solomon codes are one of the most well-known and widely used classes of error-correcting codes. These codes have excellent mathematical features and important applications, from deep-space communication to consumer electronics, such as CDs and DVDs. The extensive practical use of Reed-Solomon codes is made possible by the existence of efficient decoding procedures  and, at the same time, makes it necessary to devise ever-faster decoding algorithms. Indeed, while the structure of Reed-Solomon codes is well understood, the problem of designing optimal decoding algorithms still remains an active area of research.

In this thesis we study and compare two major strategies known for decoding Reed-Solomon codes: the Peterson-Gorenstein-Zierler (\PGZ) and the Berlekamp-Massey (\BM) decoder. Our aim is to improve existing decoding algorithms, by decreasing their computational time complexity, and propose faster new ones, based on a parallel implementation in integrated circuits.

The \PGZ decoder, introduced in 1961, was the first practical decoder for Reed-Solomon codes. It is based on simple tools of linear algebra and finds the number, the positions and the values of errors, which occurred during transmission, by computing determinants or solving linear systems. Despite its simplicity, the \PGZ decoder has often been considered quite inefficient because of its  $O(t^4)$ computational time complexity, where $t$ is the error correction capability of the code. This opinion is not justified, indeed there exists an improved version of the \PGZ decoder with quadratic computational time complexity, which we will call the fast Peterson-Gorenstein-Zierler (\fPGZ) decoding algorithm. This improvement, which was presented by M. Schimidt and G.P. Fettweis in 1996, is obtained by exploiting the Hankel structure of the matrices involved in the decoding and fast inversion techniques for Hankel matrices. 

Because of this, the \PGZ decoder can finally be considered an effective alternative to the  \BM decoder, which was introduced in 1969 as a first example of a quadratic time decoding algorithm for Reed-Solomon codes. In addition we show that the modified version of the \PGZ decoder  is not only an alternative to the \BM decoder, but in a sense it is a particular case of the latter. In fact, we prove that the intermediate outcomes obtained in the implementation of \fPGZ are a subset of those of the \BM decoding algorithm. Thus we show that the relationship between these two decoding strategies for Reed-Solomon codes is much closer than previously thought. 

The \BM decoder is an iterative algorithm which computes a recursive sequence of polynomials converging to the error-locator polynomial $\sigma(x)$, whose degree is the number of errors and whose roots identify the error positions. The error values are usually computed by evaluating $\omega(x)$ in the roots of $\sigma(x)$, where $\omega(x)$ is the error-evaluator polynomial.  We prove that it is possible to improve the error value computation in the \BM decoder avoiding the additional operations needed to compute $\omega(x)$ from $\sigma(x)$. We achieve this result by  using a polynomial which is a byproduct of the computation of $\sigma(x)$ in the place of $\omega(x)$. This alternative method for the evaluation of errors in the \BM decoder had already been observed by T. Horiguchi, but in this thesis we obtain a more direct proof, recovering the new formula as a corollary of the correctness of the \BM decoding algorithm.

Nevertheless we study the techniques of linear algebra used by T. Horiguchi in order to generalize the applications of his new error evaluation method. Thanks to this generalization, we also improve the error value computation in the \PGZ decoder. In fact we prove that the error values can be calculated by evaluating polynomials that are intermediate outcomes of the  \fPGZ decoding algorithm, instead of solving a separate linear system as usual.

Moreover thanks to the study done on the structure of the syndrome matrix and its leading principal minors, we can state a new iterative formulation of the \PGZ decoder well suited to a parallel implementation on integrated microchips. Indeed we prove that the number, the positions and the values of the errors can be computed directly by using the leading principal minors of the syndrome matrix, whose computation can be accomplished iteratively via a parallel implementation of the Laplace expansion for determinants. We show that this parallel version of the \PGZ decoder is a decoding algorithm for Reed-Solomon codes with an $O(e)$ computational time complexity, where $e$ is the number of errors which occurred, although  a fairly large number (about $t\binom{t}{\pint{t/2}}$) of elementary circuit elements is needed.

Finally a parallel implementation for the \BM decoder is given, which is less expensive and simpler from the point of view of the hardware required. In this case we restructure the \BM decoding algorithm in order to avoid some instructions hindering the parallel implementation. The result is a decoding algorithm with an $O(t)$ computational time complexity.

Our conclusions are that the \PGZ and \BM decoder are both valid alternatives for decoding Reed-Solomon codes. As seen, the choice of a decoding algorithm and an architecture for its hardware implementation has to be influenced by the particular application of the code and its resulting conditions. We analyze several different cases throughout our thesis.\\

%

The organization of the thesis is as follows. In chapter 1, Reed-Solomon codes and the general outline of their decoding algorithms are described. In chapter 2, the Peterson-Gorenstein-Zierler decoder and its quadratic improved version are discussed. We also describe the necessary and sufficient conditions to avoid decoder malfunctions in the \fPGZ decoding algorithm. In chapter 3, we study in detail the Berlekamp-Massey decoder and its implementation as a $t$-bounded distance decoding algorithm and we compare the latter with the \fPGZ decoding algorithm. Chapter 4 deals with Horiguchi's formula to compute the error values and its applications to the \fPGZ and the \BM decoding algorithms and  continues the comparison between these two decoding strategies. Finally, in chapter 5 we present and compare two parallel decoding algorithms implementable as integrated circuits, one based on the \PGZ decoder and linear algebra and the other on the \BM decoding algorithm.


\mainmatter

\chapter{Coding Theory}

Coding theory deals with mathematical models and algorithms to transmit data across noisy channels. 
In sections \ref{sect:error-correting codes}, \ref{sect:linear codes} and \ref{sect:cyclic codes} we will give a brief introduction to error-correcting codes and how they work. For the proofs of the propositions and theorems stated in these sections, we refer to \cite{hoffman}. In section \ref{sect:RScodes} we will introduce Reed-Solomon codes, with the notations and the properties that will be used in the next chapters to investigate their decoding procedure. The general outline  of the latter is described in section \ref{sect:decodingRS}.

\section{Error-Correcting Codes}
\label{sect:error-correting codes}
Error-correcting codes are used to detect and correct errors that occur when data are transmitted across some noisy channel or stored on some medium. In many cases, the data are transmitted or stored as a sequence of words of fixed length $n$ and each word is made up of symbols taken from a finite alphabet. We describe this situation considering a finite field $\K$ and a fixed natural number $n$. A \emph{block code} (or simply code)  \emph{of length $n$} is a subset of $\kn$, denoted by $\C$. In this model, the elements of the field $\K$ represent the symbols of the finite alphabet and the vectors in $\kn$ represent all possible words of length $n$. For this reason, throughout this thesis, words (of length $n$) are synonymous with vectors of $\kn$. The vectors of $\C$ represent all words that can form the transmitted sequence of data and they will be called \emph{codewords}.

The noise of the channel used for the transmission may cause some errors, namely some components of a transmitted codeword, during the transmission, may be modified (but cannot be eliminated, indeed we suppose that a codeword of length $n$ is always received as a word length $n$). In this case the received word is of the same length but different from the transmitted one. In order to measure this difference, we introduce the following tools:
\begin{defi}
 The \emph{Hamming weight} of a word $\vct{v}=(v_0,v_1,\dots,v_{n-1}) \in \kn$ is the number of the nonzero components of $\vct{v}$, that is 
 $$\wt(\vct{v})=\#\{i \;|\; v_i\neq0\}$$
The \emph{Hamming distance} between two words $\vct{u},\vct{v} \in \kn$ is the Hamming weight of the vector $\vct{u}-\vct{v}$, that is
$$\dist(\vct{u},\vct{v})=\wt(\vct{u}-\vct{v})=\#\{i \;|\; u_i\neq v_i\}$$
\end{defi}

It can be easily shown that the Hamming distance defines a metric on the vector space $\kn$.\\

Note that if $\cw \in \C$ is sent over a channel and $\rw\in \kn$ is received, then the Hamming distance between $\rw$ and $\cw$ is the number of errors which occurred, that is the number of components of $\cw$ modified during the transmission. The vector $\e=\rw-\cw$ is called the \emph{error vector} and we have that $\wt(\e)=\dist(\rw,\cw)$.

We assume that the errors are distributed randomly, that is a component of the transmitted word can be wrong independently of the other components and the probability that a component is affected in transmission is the same for all the components. If we indicate with $p$ this probability ($0\leq p\leq1$) and with $\rw\in\K^n$ a received word, then the probability that the word $\vct{v}\in\K^n$ is the sent word is
$$p^d\pt{1-p}^{n-d}$$
where $d=\dist(\rw,\vct{v})$. We can reasonably suppose that $0< p<\frac{1}{2}$ and thus we have that 
$$p^{d_1}\pt{1-p}^{n-d_1}\leq p^{d_2}\pt{1-p}^{n-d_2} \Leftrightarrow d_1\geq d_2$$
namely the closest codeword (in the sense of the Hamming distance) to $\rw$ is the most likely to have been sent. For this reason, we correct a received word $\rw$ with the codeword that is the closest  to $\rw$. If there are several words in $\C$ at the same distance from $\rw$, we may arbitrarily choose one of these or we may ask for a retransmission. Any algorithm doing this will be denoted by a \emph{decoding algorithm}.\\

Clearly with this decoding strategy mistakes may happen since the closest codeword to $\rw$ may not be the one sent. In order to understand and avoid decoding mistakes we study and classify the error vectors in the following way:
\begin{defi}
A vector $\e \in\kn\setminus\{\vct{0}\}$ is a \emph{detectable error} if for any $\cw\in\C$ we have that $\cw+\e \notin \C$, while it is an \emph{correctable error} if for any $\cw\in\C$ we have that $\dist(\cw+\e,\cw)<\dist(\cw+\e,\vct{v})$ for any $\vct{v}\in\C\setminus\{\cw\}$.
\end{defi}
In other words, $\e$ is detectable if for any transmitted codeword $\cw$, we can recognize that $\cw+\e$ is not a codeword, whereas it is correctable if $\cw$ is closer to $\cw+\e$ than any other codeword. In order to distinguish between detectable and correctable errors, the following definition is useful:
\begin{defi}
 If $\C$ is a code with at least two different words, then the integer $\dist(\C)$ defined by
 $$\dist(\C)=\min\{\dist(\vct{u},\vct{v}) \;|\; \vct{u},\vct{v} \in\C \text{ and } \vct{u}\neq \vct{v}\}$$
 is called the \emph{distance of the code} $\C$.
\end{defi} 
The following two propositions show the relation between the distance of a code and the identification of detectable or correctable error vectors.  
\begin{prop}\label{prop:detectable}
 Let $\C \subseteq \kn$ be a code of distance $d$. All error vectors $\e\in \kn\setminus\{\vct{0\}}$ such that $\wt(\e)\leqq d-1$ are detectable and  there is at least one error vector in $\kn$ of weight $d$ that is not detectable.
 \end{prop}

\begin{prop}
 Let $\C \subseteq \kn$ be a code of distance $d$. All error vectors $\e\in \kn$ such that $\wt(\e)\leqq \pint{\frac{d-1}{2}}$ are correctable and  there is at least one error vector in $\kn$ of weight $\pint{\frac{d-1}{2}}+1$ that is not correctable.
 \end{prop}

For this reason, in a code of distance $d$ the integer $$t=\pint{\frac{d-1}{2}}$$   is called the  \emph{error correction capability} of the code. Indeed as seen if the number of errors happened is less or equal to $t$, then  a decoding algorithm corrects the received word with the sent codeword and no decoding  mistakes can happen.

This codes property leads to define an useful kind of decoding algorithm:  we call \emph{$t$-bounded  distance decoding algorithm} a decoding algorithm that either decodes a received word $\rw$ into the unique codeword $\cw$ at distance not more than $t$ from $\rw$ (if such codeword exists) or indicates that no such codeword exists, declaring a \emph{decoder failure}. 
To understand the behavior of a $t$-bounded  distance decoding algorithm, we introduce the set formed by the disjoint union of all the closed balls of radius $t$ centered at a codeword: $$\mathcal B =\bigsqcup_{\cw\in\C} \overline{B}_t(\cw)$$
There are three possible cases that may occur when the codeword $\cw\in\C$ is transmitted and $\rw=\cw+\e$ is received:
\begin{enumerate}
 \item If $\wt(\e)\leq t$, then $\rw \in \mathcal B$. A $t$-bounded distance decoding algorithm decodes correctly $\rw$, giving as output $\cw$.
 \item If $\wt(\e)> t$ and $\rw\in\mathcal B$, then there exists $\vct{\tl{c}}\in \C$ such that $\vct{\tl{c}}\neq \cw$ and $\rw \in \overline{B}_t(\vct{\tl{c}})$. In this case, the output of a $t$-bounded  distance decoding algorithm will be the ``wrong'' codeword $\vct{\tl{c}}$. Thus we talk about \emph{decoder error}. Note that evidently decoder errors cannot be detected in any way and they are unavoidable for any decoding algorithm.
 \item If $\rw\notin\mathcal B$, then a $t$-bounded  distance decoding algorithm must detect the error vector $e$ and declare the decoder failure,  even though it is not capable of correcting it. 
\end{enumerate}

As we will see later on, when   $\rw\notin\mathcal B$ some decoding algorithms may not indicate decoder failure as expected and  may instead produce an output vector that is not a codeword at all. This can happen even if they behave as  $t$-bounded  distance decoding algorithm in the case that $\rw\in\mathcal B$. We refer to this event as a \emph{decoder malfunction}. We will show that decoder malfunctions can be detected and avoided by adding conditions which declare a decoder failure.

\section{Linear Codes}
\label{sect:linear codes}
Linear codes are a large family of error-correcting codes, for which tools and techniques of linear algebra are used in  encoding and decoding algorithms. In this section we briefly recall the main proprieties of these codes.
\begin{defi}
Let $\C$ be a code of length $n$. If $\C$ is also a linear subspace of $\kn$, then it is called a \emph{linear code}. In this case, the dimension of $\C$ as a linear subspace of $\kn$ is called the \emph{dimension of the code} $\C$, 
\end{defi}
We indicate with $\C(n,k,d)$ a linear code of length $n$, dimension $k$ and distance $d$.
The structure of vector space allows a simpler description of the code. For example, if $\C$ is a linear code, then its distance can be calculated simply as 
$$\dist(\C)=\min\{\wt(\vct{v}) \;|\; \vct{v} \in\C \text{ and } \vct{v}\neq \vct{0}\}$$
and moreover it is possible to define the following useful matrices, which  characterize a linear code:
\begin{defi}
 Let $\C$ be a linear code of length $n$ and dimension $k$. Any $k\times n$ matrix whose rows form a basis for $\C$, as a linear subspace of $\kn$, is called a \emph{generator matrix}. While any $n\times (n-k)$ matrix whose columns form a basis for the orthogonal complement of $\C$ is called a \emph{parity-check matrix}.
\end{defi}
Obviously, if $G$ is a generator matrix for the code $\C(n,k,d)$, then we have $\rk(G)=k$ and 
$$\vct{v}\in\C \Longleftrightarrow \exists\, \vct{u}\in\K^k \text{ such that } \vct{v}=\vct{u}G$$
while if $P$ is a parity-check matrix for $\C(n,k,d)$, then $\rk(P)=n-k$ and
$$\vct{v}\in\C \Longleftrightarrow \vct{v}P=\vct{0}$$

The following proposition shows how to calculate the distance of a linear code knowing a parity-check matrix. 
\begin{prop}\label{prop:distanceParity}
 Consider a linear code $\C$ of length $n$ and dimension $k$. Let $P$ be a parity-check matrix for $\C$. Then  $\C$ has distance $d$ if and only if any $d-1$ rows of $P$ are linearly independent, and at least $d$ rows of $P$ are linearly dependent.
\end{prop}

As consequence, we obtain the following theorem that shows the existing relation between the three main characteristics of a linear code: length, dimension and distance.
\begin{thm}[Singleton Bound]
For any  linear code $\C(n,k,d)$, $$\,d\geq n-k+1.$$
\end{thm}
A linear code $\C(n,k,d)$ that realizes the equality $d=n-k+1$ is said to be a \emph{maximum distance separable} (or MDS) \emph{code}.
If we recall that the integer $t=\pint{\frac{d-1}{2}}$  represents the error correction capability of the code, we can understand that MDS codes are optimal codes, in the sense that they have the greatest error correction capability for the same length and dimension.

\section{Cyclic Codes}
\label{sect:cyclic codes}
In order to have very efficient decoding algorithms, linear codes may be no sufficient. So we introduce codes with more algebraic  structure:
\begin{defi}
A code $\C$ of length $n$ is \emph{cyclic} if it is closed under the operation cyclic right shift of a codeword. In other words $\C \subseteq \kn$ is a cyclic code if 
$$\vct{v}=(v_0, v_1,\dots,v_{n-1})\in\C\Longrightarrow \pi(\vct{v})\overset{def}{=}(v_{n-1},v_0,\dots,v_{n-2})\in\C$$ 
\end{defi}

To appreciate the algebraic structure of a cyclic code we introduce the representation of the elements in $\kn$ as polynomials.
There is an one-to-one correspondence between  vectors in $\kn$ and polynomials in $\K[x]$ of degree less than $n-1$, defined by
$$\vct{v}=(v_0, v_1,\dots,v_{n-1}) \longleftrightarrow \vct{v}(x)=\sum_{i=0}^{n-1}v_{i} x^i$$
Thus a code of length $n$ can be represented as a set of polynomials of degree at most $n-1$. Moreover if the vector $\vct{v}$ corresponds to the polynomial $\vct{v}(x)$, as showed above, the vector $\pi(\vct{v})$ corresponds to the polynomial $x\vct{v}(x) \Mod{x^n-1}$. So a cyclic code of length $n$ can be seen as a subspace of the ring $\K[x]/(x^n-1)$ closed under the multiplication by $x$. That is, a cyclic code of length $n$ corresponds to an ideal of the ring $\K[x]/(x^n-1)$.
Moreover any ideal of $\K[x]/(x^n-1)$ corresponds to a cyclic code and so we have the one-to-one correspondence given by:
$$\C \subseteq \kn \text{ cyclic code } \longleftrightarrow \C\subseteq \K[x]/(x^n-1) \text{ ideal } $$

Henceforth, when studying cyclic codes, we will refer to the elements in $\kn$ both as vectors and as polynomials. \\

Since $\K[x]/(x^n-1)$ is a principal ideal ring, every cyclic code, as an ideal, can be generated by just one element. More precisely, we have:
\begin{prop}\label{cyclic1}
Let $\C$ be a cyclic code, then there exists a unique nonzero polynomial $g(x)\in \C$ such that
$$\vct{v}(x)\in\C \Longleftrightarrow g(x)|\vct{v}(x)$$
\end{prop}

We define the \emph{generator polynomial} of the cyclic code $\C$ the polynomial $g(x)$ described above and we call the cyclic code  of length $n$ generated by the polynomial  $g(x)$, the code formed by all vectors $\vct{v}\in \K^n$ such that $g(x)|\vct{v}(x)$. Remark that not all polynomials in $\K[x]$ can be generator polynomials of a cyclic code. Indeed the set $\{\vct{v}\in \K^n \,:\, g(x)|\vct{v}(x)$\} is a cyclic code if and only if the polynomial $g(x)| x^n-1$.

\begin{prop}\label{prop:cycliccode}
If $\C$ is a cyclic code of length $n$ and $g(x)$ is its generator polynomial, then 
$$\dim(\C)=n-\deg(g(x))$$\\
\end{prop}

An important class of cyclic codes, still used a lot in practice, was discovered by R. C. Bose and D. K. Ray-Chaudhuri (1960) and independently by A. Hocquenghem (1959). The codes are known as \emph{BCH codes}.
\begin{defi}[BCH code]\label{def:BCH}
Let $\F$ be the finite field consisting of $q$ elements and let $n$ and $\delta$ be two positive integer. Given the smallest $a\in \N$ such that $\mathbb F_{q^a}$ has a primitive $n$th root of unity $\beta$, then a cyclic code of length $n$ over $\F$ is called a \emph{BCH code of parameters $n$ and $\delta$ generated by $\beta$} if its generator polynomial is the least common multiple of the minimal polynomials over $\F$ of $\beta^l, \beta^{l+1}, \dots, \beta^{l+\delta-2}$, for some $l\in \N$. This code is indicated by $BCH(n,\delta,\beta)$.
\end{defi}
We observe that the the least common multiple of the minimal polynomials of $\beta^l, \beta^{l+1}, \dots, \beta^{l+\delta-2}$  divides $x^n-1$ because  $\beta$ is a primitive $n$th root of unity. So the above definition is correct.


\begin{ex}\label{ex_bch}
Let $l=1$ and $q=2$. Let $\beta\in\mathbb F_{16}$ be a primitive fourth root of unity satisfying $\beta^4+\beta+1=0$, then its minimal polynomials over $\mathbb F_2$ is the irreducible polynomial $f_1(x)=x^4+x+1$. 
Since $f_1(\beta^2)=f_1(\beta)^2$ and $f_1(\beta^4)=f_1(\beta)^4$, $f_1(x)$ is also the minimal polynomials over $\mathbb F_2$ of $\beta^2$ and of $\beta^4$. Moreover it can be easily proved that the minimal polynomials over $\mathbb F_2$ of $\beta^{3}$ is $f_2(x)=x^4+x^3+1$.
Thus the code $BCH(15,5,\beta)$ over $\mathbb F_2$ is generated by the polynomial
$$g(x)=f_1(x)\cdot f_2(x)=x^8+x^7+x^5+x^4+x^3+x+1$$
\end{ex}

\section{Reed-Solomon Codes}
\label{sect:RScodes}
Reed-Solomon codes, named after their inventors \cite{reedSolomon}, are an extensively studied family of error-correcting codes, heavily used in theoretical and practical settings. Furthermore, the optimality of the Reed-Solomon codes, in terms of error correction capability, and their algebraic properties are the main reasons for the  great fame of these codes and their widespread use in coding theory and computer science.  \\

From here on, let $\F$ be the finite field consisting of $q$ elements ($\F=\mathbb F_{p^m}$ with $m,p\in\N$ and $p$ prime). We indicate with $n$ the cardinality of the multiplicative group of $\F$, that is $n=q-1$, and with $\alpha$ a primitive element of $\F$, i.e.\ an element with multiplicative order equal to $n$. Reed-Solomon codes are special case of cyclic codes of length $n$ over the field $\F$. More precisely:
\begin{defi}[Reed-Solomon code] \label{defiRS}
Let $n$ and $\alpha$ be as above and let $d$ be an integer less than or equal to $n$. The \emph{Reed-Solomon code of parameters $n$ and $d$ generated by $\alpha$} (denoted by $RS(n,d,\alpha)$) is the cyclic code of length $n$ over $\F$ generated by the polynomial $$g(x)=(x-\alpha^l)(x-\alpha^{l+1}) \cdots(x-\alpha^{l+d-2})$$ 
for some $l\in \N$
\end{defi}

We observe that Reed-Solomon codes are examples (the simplest ones) of BCH codes, namely in the case $n=q-1$.
For simplicity, we will always take $l=1$, but the general treatment is equivalent. The following proposition summarizes the properties of a Reed-Solomon code. 
\begin{prop}\label{prop:rscode}
Let $RS(n,d,\alpha)$ be as in definition \ref{defiRS}. Then the following holds:
 \begin{enumerate}[\ \ \ (1)]
  \item $\dim\left(RS(n,d,\alpha)\right)=n-d+1$;
  \item $\vct{v}\in RS(n,d,\alpha)\Leftrightarrow \vct{v}(\alpha^i)=0 \; \forall \, i=1,2,\dots,d-1$ and a parity-check matrix for the code $RS(n,d,\alpha)$  is
         $$ P=\begin{mymatrix}
                        1            & 1               & \cdots & 1 \\
                        \alpha       &\alpha^2         & \cdots & \alpha^{d-1}\\
                        \alpha^2     &(\alpha^2)^2     & \cdots & (\alpha^{d-1})^2\\
                        \vdots       & \vdots          &        & \vdots \\
                        \alpha^{n-1} & (\alpha^2)^{n-1} & \cdots &(\alpha^{d-1})^{n-1}\\
        \end{mymatrix};$$
  \item $\dist\left(RS(n,d,\alpha)\right)=d$;
  \item the code $RS(n,d,\alpha)$ is an MDS code.
  \end{enumerate}
\end{prop}
\begin{proof}
 \begin{enumerate}
 \item From definition \ref{defiRS} it follows that $\deg(g(x))=d-1$, thus (1) is an immediate consequence of proposition \ref{prop:cycliccode}.
 \item Since the code $RS(n,d,\alpha)$ is a cyclic code generated by the polynomial $g(x)$, the equivalence stated in (2) follows directly from proposition \ref{cyclic1}. To verify that $P$ is a parity-check matrix for the code  $RS(n,d,\alpha)$ is sufficient to observe that, if $\vct{v} \in(\F)^n $, then we may evaluate the polynomial $\vct{v}(x)$ in $\alpha^j$ with an inner product, in fact: 
 $$\vct{v}(\alpha^j)= \left(v_0, v_1,\dots,v_{n-1}\right)\begin{mymatrix}
                                                         1 \\ \alpha^j\\  \vdots\\ \alpha^{(n-1)j} 
                                                        \end{mymatrix}=
                                                        \sum_{i=0}^{n-1}v_i\alpha^{ij}$$
 Therefore we have that $\vct{v}(\alpha^i)=0 \; \forall \, i=1,2,\dots,d-1 \Leftrightarrow \vct{v}P=\vct{0}$.
 \item From the Singleton bound, using (1), it  follows that $$\dist\left(RS(n,d,\alpha)\right)\leq d$$ To conclude the proof it is sufficient to show that any $d-1$ rows of $P$ are linearly independent. In fact, by proposition \ref{prop:distanceParity}, this leads to $\dist\left(RS(n,d,\alpha)\right)\geq d$. Consider $d-1$ rows in $P$. In other words, for any $i=1,2,\dots,d-1$, we choose $m_i\in\{0,1,\dots,n-1\}$ (with the condition $m_i\neq m_j$, if $i\neq j$) and we define $a_i=\alpha^{m_i}$. So we obtain the square submatrix 
 $$M=\begin{mymatrix}
      a_1   &  a_1^2 & \cdots & a_1^{d-1}\\
      a_2   &  a_2^2 & \cdots & a_2^{d-1}\\
      \vdots & \vdots &       &\vdots\\
      a_{d-1}   &  a_{d-1}^2 & \cdots & a_{d-1}^{d-1}\\
   \end{mymatrix}$$ 
 and it is easy to show that $M$ is non-singular because it can be written as a product of a non-singular Vandermonde matrix and a non-singular diagonal matrix. That is
 $$M=\begin{mymatrix}
      1   &a_1 & a_1^2 & \cdots & a_1^{d-2}\\
      1& a_2   &  a_2^2 & \cdots & a_2^{d-2}\\
      \vdots & \vdots &  \vdots  &    &\vdots\\
      1 &a_{d-1}   &  a_{d-1}^2 & \cdots & a_{d-1}^{d-2}\\
   \end{mymatrix}  
   \begin{mymatrix}
      a_1 & 0     & \cdots   & 0\\
      0   &  a_2 & \ddots & \vdots\\
      \vdots & \ddots &   \ddots    & 0\\
      0   &   \cdots &0 & a_{d-1}\\
   \end{mymatrix}$$ 
\item Trivial using (1) and (3).
\end{enumerate}
\end{proof}

\section{Decoding Reed-Solomon Codes}\label{sect:decodingRS}
Consider a Reed-Solomon code of length $n$ and distance $d$ over $\F$ and let $\alpha$ be the primitive $n$th root of unity that defines the code. 

Suppose that $\cw=(c_0,c_1,\dots,c_{n-1})\in RS(n,d,\alpha)$ is the codeword sent and $\e=(e_0,e_1,\dots,e_{n-1})\in (\F)^n$ is the error vector, thus the received word is $\rw=\cw+\e=(r_0,r_1,\dots,r_{n-1})\in (\F)^n$. We represent this situation with the polynomials
$$\cw(x)=\sum_{i=0}^{n-1}c_ix^i \in \F[x]$$
$$\e(x)=\sum_{i=0}^{n-1}e_ix^i \in \F[x]$$
$$\rw(x)=\sum_{i=0}^{n-1}r_ix^i \in \F[x]$$
and, clearly, we have $\rw(x)=\cw(x)+\e(x)$.

From now on, we call $e$ the Hamming weight of the error vector $\e$ and we define $$\mathcal I=\{i\; |\; e_i\neq0\}=\{p_1<p_2<\cdots<p_e\}$$ the set of the positions where an error occurs. Thus for any 
$i\in\{1,2,\dots, e\} $, $E_i=e_{p_i}\in \F$ represents the value of the error that occurs in the position $p_i$.
In order to correct the received word $\rw$, we need to find the set $\mathcal I$ and the values $E_i$. For this purpose we introduce the following tools:
\begin{defi}\label{defi:sigma_omega}
For any $i\in\{1,2,\dots, e\} $, we define 
\begin{equation}\label{rel:X_i}
X_i=\alpha^{p_i}\in\F 
\end{equation}\\[-1cm]
Moreover we call
\begin{align*}
&\sigma(x)=\prod_{i=1}^e (1-X_ix) \quad \quad \text { the \emph{error-locator polynomial} } \\
&\omega(x)=\sum_{i=1}^e E_iX_i \prod_{j\neq i} (1-X_jx) \quad \quad \text{ the \emph{error-evaluator polynomial.}}\\
\end{align*}
\end{defi}

\noindent Evidently the elements $X_1^{-1}, X_2^{-1},\dots, X_e^{-1}$ are precisely the roots of $\sigma(x)$ in the field $\F$ and  they identify the errors positions thanks to the relation
$$X_i^{-1}=\alpha^{n-p_i} \quad \quad \text{ for any } i\in\{1,2,\dots,e\}$$
which follows by (\ref{rel:X_i}). Thus if we are able to calculate the error-locator polynomial, then we can know all the error positions $\{p_1,p_2,\dots,p_e\}$ from its roots.
While the error-locator polynomial permits calculating the error positions, the error-evaluator polynomial is linked to the computation of the error values $E_1,E_2,\dots, E_e$. Indeed it is easy to prove the following relation, stated for the first time in \cite{forney} and known as \emph{Forney's formula}: 
\begin{equation}\label{forney}
E_i=-\frac{\omega(X_i^{-1})}{\sigma'(X_i^{-1})}\quad \quad \forall\; i=1,2,\dots,e                                                                                                                                                                                                                                         \end{equation}
where $\sigma'(x)$ is the formal derivative of $\sigma(x)$. Thus the knowledge of the polynomials  $\sigma(x)$ and $\omega(x)$ permits a complete decoding of the received word $\rw$, since it allows the computation of the error positions and of error values. For this reason from here on our aim is the determination of the these polynomials.

Other simple properties of the error-locator and the error-evaluator follow immediately by their definition. We list them below:

\begin{rem} \hspace{3cm}
\begin{enumerate}\label{rem:sigmaproperties} 
 \item $\deg(\sigma(x))=e=\wt(\e)$ and   $\sigma(0)=1$;
 \item $\sigma(x)$ splits completely in $\F[x]$ and does not have multiple roots;
 \item $\deg(\omega(x))\leq e-1$;
 \item $\gcd\left(\sigma(x), \omega(x)\right)=1$.\\
\end{enumerate}
\end{rem}

In proposition \ref{prop:rscode} we saw that
$$\vct{v}\in RS(n,d,\alpha)\Leftrightarrow \vct{v}(\alpha^i)=0 \quad \forall \, i=1,2,\dots,d-1$$ 
 and this justifies the introduction of the following definition:
\begin{defi}\label{defi:syndromes}
For any $i\in\{1,2,\dots, d-1\}$ we define  
$$S_i=\rw(\alpha^i) \quad \quad \text{ the \emph{$i$th syndrome}}$$
Moreover $S(x)=\displaystyle\sum_{i=1}^{d-1} S_i x^{i-1}$ is the \emph{syndrome polynomial}.
\end{defi}

It is clear that if $S_1=S_2=\cdots=S_{d-1}=0$, then $\rw\in RS(n,d,\alpha)$ and thus either $\e=\vct{0}$ or $\e$ is a not detectable error. If we suppose that $e=\wt(\e)$ is less or equal to $t$, the error correction capability of the code  $ RS(n,d,\alpha)$, then the second case is not possible by proposition \ref{prop:detectable} and hence the condition $S_1=S_2=\cdots=S_{d-1}=0$ implies that no error has occurred.\\

Since $\cw\in RS(n,d,\alpha)$, $\cw(\alpha^i)=0$ for every $i\in\{1,2,\dots,d-1\}$, thus we can observe that 
\begin{equation}\label{rel:syndrome}
S_i=\rw(\alpha^i)=\underbrace{\cw(\alpha^i)}_{=0}+\e(\alpha^i)=\e(\alpha^i)=\displaystyle\sum_{j=1}^eE_jX_j^i
\end{equation}
for any $i\in\{1,2,\dots,d-1\}$. This characterization of the syndromes is of fundamental importance because it allows to write a relation, more precisely a congruence, stated in the following proposition and satisfied by the error-locator and error-evaluator polynomials.

\begin{prop}[key equation] If $\sigma(x)$, $\omega(x)$ and $S(x)$ are the polynomials respectively defined in definition \ref{defi:sigma_omega} and \ref{defi:syndromes}, then
\begin{equation}\label{keyequation}
 \sigma(x)S(x)\equiv\omega(x) \Mod{  x^{d-1}}
\end{equation}
\end{prop}
\begin{proof}
We observe that 
$$\dfrac{\omega(x)}{\sigma(x)}=
   \dfrac{\displaystyle\sum_{i=1}^e E_iX_i \prod_{j\neq i} (1-X_jx)}{\displaystyle\prod_{k=1}^e (1-X_kx)}
   =\sum_{i=1}^e E_iX_i \dfrac{1}{1-X_ix}$$
and using the identity $\displaystyle\sum_{i\geq 0}y^i=\frac{1}{1-y}$ we obtain
$$\dfrac{\omega(x)}{\sigma(x)}=\sum_{i=1}^e E_iX_i\sum_{j=0}^{+\infty} (X_ix)^j=\sum_{j=0}^{+\infty} \left(\sum_{i=1}^eE_iX_i^{j+1}\right)x^j$$
By (\ref{rel:syndrome}), we have
$$\dfrac{\omega(x)}{\sigma(x)}
=\underbrace{\sum_{j=0}^{d-2}S_{j+1}x^j}_{S(x)}+\sum_{j\geq d-1}\pt{\sum_{i=1}^eE_iX_i^{j+1}}x^j $$
Considering the last equality modulo $x^{d-1}$ the proof is completed.\\
\end{proof}

Given the syndrome polynomial, the pair $(\sigma(x),\omega(x))$ is not the unique possible solution of the key equation (\ref{keyequation}). For this reason we consider the set:
$$\mathcal M=\left\{(a(x),b(x))\in \F[x]^2 \;|\; a(x)S(x)\equiv b(x) \Mod{  x^{d-1}}\right\}$$
Evidently, $\mathcal M$ is a submodule of $\F[x]^2$ and $(\sigma(x),\omega(x))\in\mathcal M$.
Considering that $\Deg{\omega(x)}<\Deg{\sigma(x)}$, we are interested in the pairs $(a(x),b(x))\in \mathcal M$ such that $\Deg{b(x)}<\Deg{a(x)}$ and thus we introduce the followings definitions:
\begin{defi}\label{defi:validsolution}
We call a pair $(a(x),b(x))\in \F[x]^2$ a \emph{valid solution} of the key equation (\ref{keyequation}) if $ (a(x),b(x))\in \mathcal M$ and $\Deg{b(x)}<\Deg{a(x)}$ . Moreover we call a pair $(a(x),b(x))\in \mathcal M$ a \emph{valid solution of minimal degree} if it is a valid solution and for any other valid solution $(a_1(x),b_1(x))$, we have that $\deg(a(x))\leq \deg(a_1(x))$.
\end{defi}
Recalling that $e=\wt(\e)$ and $t$ is the error correction capability of the code $RS(n,d,\alpha)$, we have:

\begin{thm}\label{thm:keyequation2}
If $e\leq t$, then the pair $(\sigma(x),\omega(x))$ is the  unique, up to a multiplicative constant, valid solution of minimal degree of the key equation (\ref{keyequation}). 
\end{thm}
\begin{proof}
We have already observed that  $(\sigma(x),\omega(x))$  is a valid solution of the key equation (\ref{keyequation}). Now we show that $(\sigma(x),\omega(x))$ is a valid solution of minimal degree. If there were  a valid solution $(a(x),b(x))$ such that $\deg(a(x))<\deg(\sigma(x))$ then we would have 
\begin{equation}\label{aiutino0}
a(x)\omega(x)\equiv a(x)\sigma(x)S(x)\equiv b(x)\sigma(x) \Mod{ x^{d-1}} 
\end{equation}
and, since   $\deg(\sigma(x))=e$,
$$\deg(a(x)\omega(x))=\deg(a(x))+\deg(\omega(x))< 2e < d-1$$ 
$$\deg(b(x)\sigma(x))=\deg(b(x))+\deg(\sigma(x))< 2e < d-1$$
From the degree inequalities, it follows that the congruence (\ref{aiutino0}) is  an equality, i.e.\ :  $$a(x)\omega(x)= b(x)\sigma(x).$$ Therefore   $\sigma(x)| a(x)\omega(x)$.  Since $\gcd\left(\sigma(x),\omega(x)\right)=1$, it follows that $\sigma(x)| a(x)$ and this is absurd. Hence $(\sigma(x),\omega(x))$ is a valid solution of minimal degree. 
Finally we show that $(\sigma(x),\omega(x))$ is the unique (up to a multiplicative constant) valid solution of minimal degree. Let $(\tl{\sigma}(x),\tl{\omega}(x))\in \mathcal M$ be a valid solution of minimal degree of the key equation (\ref{keyequation}). Then $\Deg{\tl{\sigma}(x)}=\deg(\sigma(x))=e$ and 
$$\tl{\sigma}(x)\omega(x)\equiv \tl{\sigma}(x)\sigma(x)S(x)\equiv \tl{\omega}(x)\sigma(x) \Mod{ x^{d-1}}$$
with $\deg(\tl{\sigma}(x)\omega(x))< 2e < d-1$ and $\deg(\tl{\omega}(x)\sigma(x))< 2e < d-1$.
Thus we have that $$\tl{\sigma}(x)\omega(x)= \tl{\omega}(x)\sigma(x)$$
Since $\gcd\left(\sigma(x),\omega(x)\right)=1$ and $\Deg{\tl{\sigma}(x)}=\deg(\sigma(x))$, we can conclude that there exists $k\in\F$ such that $$\tl{\sigma}(x)=k\sigma(x) \text{ and } \tl{\omega}(x)=k\omega(x)$$
The proof is concluded.
\end{proof}

Thanks to theorem \ref{thm:keyequation2}, developing procedures that, given the syndrome polynomial $S(x)$, solve the key equation (\ref{keyequation}) finding its minimal valid solution, will allow us to be able to find the polynomials $\sigma(x)$ and $\omega(x)$, that is to decode $\rw$. \\

Any procedure that solves the key equation (\ref{keyequation}) and more in general any procedure that performs the decoding of Reed-Solomon codes, requires computations using finite field arithmetic. In this thesis we do not discuss in detail the implementations and the circuits that perform addition, multiplication and division over $\F$ (for which we refer to \cite{lin} and \cite{peterson}). We only recall that the elements of $\F=\mathbb F_{p^m}$ can be represented as vectors of $\pt{\mathbb F_p}^m$ with respect to a basis (for example $\left\{1,\alpha,\alpha^2,\dots,\alpha^{m-1}\right\}$ with $\alpha$ primitive element of $\F$) and that, with the vector representation, additions and multiplications by constant field element are simple, but unconstant  multiplication and division are not. 
In particular:
\begin{itemize}
 \item[-] To add two field elements, we simply add their vector representations over the base field $\mathbb{F}_p$. This operation is the simplest among $\F$ arithmetic.
  
 \item[-] Multiplying  a field element by a fixed element from the same field is simple enough because it can be seen as a $\mathbb F_p$ linear map of $\F$. Thus it involves only few additions and multiplications over the base field $\mathbb{F}_p$ (see example \ref{ex:multiplication}), whereas multiplying two arbitrary field elements is more expensive in terms of number of operations over $\mathbb{F}_p$ involved and in terms of implementation. 
 
 \item[-] Division can be handled by first computing the inverse and then multiplying by it. It is the most expensive  operation in $\F$ because it adds to the  multiplication complexity the time of an inverse calculation. The inversion in $\F$ is a quite complicated operation: a direct approach is to use a table of $mp^m$ positions in which the inverses of the field elements are stored, but there are also alternative methods that can be more advantageous for decoder where codes are defined over $\F$ for different values of $q$. For example there are several methods for computing an inverse based on the Extended Euclidean Algorithm. See \cite{bhargava} for details. \\
 \end{itemize}

 \begin{ex}\label{ex:multiplication}
  Suppose we want to multiply a field element $\beta$ in $\mathbb F_{2^4}$ by the primitive element $\alpha$ whose minimal polynomial is $f(x)=x^4+x+1$. The element $\beta$ can be expressed as $\beta=b_0+b_1\alpha+b_2\alpha^2+b_3\alpha^3$, with $b_i\in \mathbb F_2$. Thus
  $$\beta\cdot\alpha=b_0\alpha+b_1\alpha^2+b_2\alpha^3+b_3\underbrace{\alpha^4}_{=\,\alpha+1}=
  b_3+(b_0+b_3)\alpha +b_1\alpha^2+b_2\alpha^3$$
  This multiplication can be carried out by one shift of the feedback register shown in figure \ref{fig:multiplier} with the cost of only one addition in the base field $\mathbb F_2$. The circle represents an adder over $\mathbb F_2$, while the rectangular elements are storage devices.\\
 \end{ex}
 
 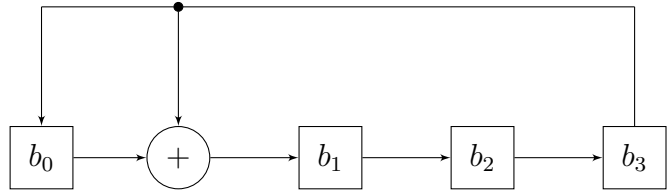
\begin{figure*}[h!]
\begin{center}
\begin{tikzpicture}[auto, node distance=2cm,>=latex']
    \node [coordinate](pall1){};
    \node [pall, right of=pall1, node distance=6cm](pall2){};
    \node [coordinate, node distance=6cm, right of=pall2](pall3){};
    \node [block, below of=pall1](b0){$b_0$};
    \node [sum, below of=pall2](somma){$+$};
    \node [block, right of=somma] (b1){$b_1$};
    \node [block, right of=b1](b2){$b_2$};
    \node [block, below of=pall3](b3){$b_3$};
     
    \draw [->] (b0) --  (somma);
    \draw [->] (somma) --  (b1);
    \draw [->] (b1) --  (b2);
    \draw [->] (b2) -- (b3);
    \draw [->] (pall1) -- (b0);
    \draw [->] (pall2) -- (somma);
    \draw [] (pall2) -- (pall1); 
    \draw [-](pall3) -- (pall2);
    \draw [](b3) -- (pall3);
\end{tikzpicture}
\caption{circuit for multiplying arbitrary element in $\mathbb F_{2^4}$ by $\alpha$}\label{fig:multiplier}
\end{center}
\end{figure*}
First the vector representation $(b_0,b_1,b_2,b_3)$ of $\beta$ is loaded into the register. Then the register is pulsed: all the stored elements shift right and the element $b_3$ goes from the fourth storage device  into the first one. At the same time,  the element $b_0$ is moved into the adder where it is added with $b_3$. The sum is stored in the second device. The final configuration contains the vector representation of $\beta\alpha$ (see figure \vref{fig:multiplier2}).\\

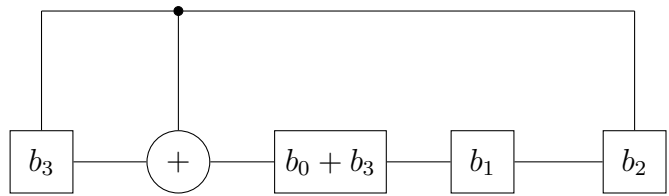
\begin{figure*}[h!]
\begin{center}
\begin{tikzpicture}[auto, node distance=2cm,>=latex']
    \node [coordinate](pall1){};
    \node [pall, right of=pall1, node distance=6cm](pall2){};
    \node [coordinate, node distance=6cm, right of=pall2](pall3){};
    \node [block, below of=pall1](b0){$b_3$};
    \node [sum, below of=pall2](somma){$+$};
    \node [block, right of=somma] (b1){$b_0+b_3$};
    \node [block, right of=b1](b2){$b_1$};
    \node [block, below of=pall3](b3){$b_2$};
      
     \draw [] (b0) --  (somma);
     \draw [] (somma) --  (b1);
     \draw [] (b1) --  (b2);
     \draw [] (b2) -- (b3);
     \draw [] (pall1) -- (b0);
     \draw [] (pall2) -- (somma);
     \draw [] (pall2) -- (pall1); 
     \draw [](pall3) -- (pall2);
     \draw [](b3) -- (pall3);
\end{tikzpicture}
\caption{final configuration of the register of fig.\ \ref{fig:multiplier}}\label{fig:multiplier2}
\end{center}
\end{figure*}

At this point, we can outline the general procedure of the decoding algorithms for Reed-Solomon codes, which consists of four major steps:
\begin{enumerate}[\bfseries Step 1.]
 \item \textbf{Computation of the syndromes.} Since for any $i\in\{1,2,\dots,d-1\}$ $S_i=\rw(\alpha^i)$, this first task may be accomplished using \emph{Horner's method} for evaluating polynomials. This technique is based on the polynomial decomposition given by
 $$r_{n-1}x^{n-1}+\cdots+r_1x+r_0= \big(\cdots(r_{n-1}x+r_{n-2})x+\cdots\big)x+r_0$$
 that allows to compute $S_j$ gradually as $r_i$'s are received, as it is shown by the following calculations:\\[0,3cm]
 \texttt{Input}:  $\rw=(r_0,r_1,\dots,r_{n-1})$, $\alpha^j$;\\
 \texttt{Output}:  $S_j$;\\[-0,8cm]
 \begin{tabbing}
 \texttt{Begin}\\
   $S_j:=r_{n-1}$;\\
  \texttt{for } \= $i=n-2,n-3,\dots,0$ { do }\\
  \> $S_j:=\alpha^j\cdot S_j+ r_i$;\\
  \texttt{endfor};\\ \texttt{End}
 \end{tabbing}
 We note that all syndromes can be computed simultaneously and that $r_{n-1}$ is the first received symbol. Thus after $r_0$ is received, all $d-1$ syndrome computations are completed at the same time.
  
 Figure \ref{fig:horner} represents a circuit that after $n$ shits contains $S_i=\rw(\alpha^i)$. The vector $\rw=(r_0,r_1,\dots,r_{n-1})$ is shifted into the circuit one component a time. After the first shift, the storage device contains $r_{n-1}$, after the second it contains  $r_{n-1}\alpha^i+r_{n-2}$ and so on.\\
  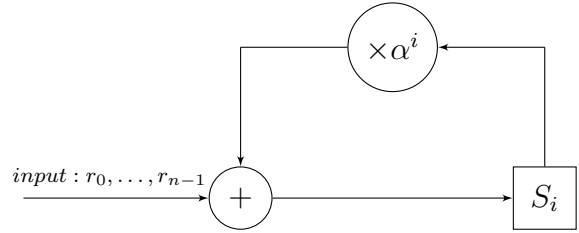
\begin{figure*}[h!]
\begin{center}
\begin{tikzpicture}[auto, node distance=2cm,>=latex']
    \node [coordinate](pall0){};
    \node [coordinate, right of=pall0, node distance=3cm](pall1){};
    \node [sum,  right of=pall1](alpha){$\times \alpha^i$};
    \node [coordinate, right of=alpha](pall2){};
    \node [output, below of=pall0](input){};
    \node [sum, below of=pall1](somma){$+$};
    \node [block, below of=pall2](s){$S_i$};
     
    \draw [->] (pall1) -- (somma);
    \draw [->] (pall2) -- (alpha);
    \draw [->] (somma) -- (s);
    \draw [->] (input) -- node {\scriptsize $input:r_0,\dots,r_{n-1}\;\;$}(somma);
    \draw [] (alpha) -- (pall1);
    \draw [] (s) -- (pall2);
\end{tikzpicture}
\caption{circuit for computing $\rw(\alpha^i)$}\label{fig:horner}
\end{center}
\end{figure*}
  
 \item \textbf{Determination of the error-locator polynomial and of the number of errors that occurred.} This is done exploiting the characterization of $\sigma(x)$ due to the key equation (\ref{keyequation}). In chapters 2 and 3 we will study two different strategies to accomplish this point.
  
 \item \textbf{Finding the error positions.}  Recalling that the elements $X_i=\alpha^{p_i}$ are the inverses of the roots of $\sigma(x)$, the error positions  $p_1,p_2,\dots,p_e$ are computed using \emph{Chien's search}. This is an exhaustive search over all the elements in $\F$ that finds the roots of $$\sigma(x)=\sigma_ex^e+\sigma_{e-1}x^{e-1}+\cdots+\sigma_1x+1$$ based on the fact that 
 $$\alpha^i \text{ is a root of } \sigma(x) \Longleftrightarrow 1+\sum_{j=1}^e \sigma_j\alpha^{ij}=0$$
 Thus if $i$ goes from $1$ to $n-1$, at each step it is necessary only to multiply the $j$th addend $\sigma_j\alpha^{ij}$ by $\alpha^j$ in order to obtain the next set of addends.
  Chien's search procedure can be summarized as follow:
 \begin{alg}[Chien's search for error positions]\label{alg:chien}\hspace{3cm}
 \begin{description}
  \item[\texttt{Input}:]  the length of the code $n$, the error number $e$ and the coefficients of 
  $\sigma(x)=\sigma_ex^e+\sigma_{e-1}x^{e-1}+\cdots+\sigma_1x+1$.
  \item[\texttt{Output}:]  the error positions $p_1<p_2<\dots<p_e$.\\[-1cm]
  \end{description}
  \begin{tabbing}
  \texttt{Begin}\\[0,2cm]
   $k:=e$;\\
   \texttt{for }\= $i=1,2,\dots,n$ \texttt{ do }\\
   \> $\sigma_l:=\sigma_l\cdot \alpha^l$ \texttt{ for any } $l=1,2,\dots,e$;\\
   \> $\displaystyle S:=1+\sum_{j=1}^e \sigma_j$;\\
   \> \texttt{if }\= $S=0$ \texttt { then }\\ 
   \>\> $p_k:=n-i$;\\
   \>\> $k:=k-1$;\\
   \> \texttt{endif}\\
    \texttt{endfor}\\
    \texttt{If } $k\neq0$ \texttt{ then declare a decoder malfunction;}\\[0,2cm]
   \texttt{End}
  \end{tabbing}
 \end{alg}
 Chien's search can be implemented in a single circuit (see figure \ref{fig:chien}) with $t$ multipliers for multiplying by $\alpha, \alpha^2, \dots,\alpha^t$ respectively and few adders, with a computational complexity upper bounded by $n$ multiplications and $en$ additions. More details are given in \cite{chien}.\\
 
\begin{figure*}[h!]
\begin{center}
\begin{tikzpicture}[auto, node distance=15mm,>=latex']
    \node [coordinate](pall1){};
    \node [sum, right of=pall1,node distance=23mm](pall2){$+$};
    \node [coordinate, right of=pall2](spazio){};
    \node [sum, right of=spazio](pall3){$+$};
    \node [sum, right of=pall3, node distance=3cm](pall4){$S=0?$};
    \node [output, right of=pall4](pall5){};
    \node [coordinate](aiuto1){};
    \node [block, below of=pall1](sigma1){$\sigma_1$};
    \node [coordinate, right of=sigma1](aiuto1){};
    \node [block, below of=pall2](sigma2){$\sigma_2$};
    \node [coordinate, right of=sigma2](aiuto2){};
    \node [coordinate, right of=aiuto2](aiutot){};
    \node [block, below of=pall3] (sigmat){$\sigma_t$};
    \node [coordinate, right of=sigmat](aiutot){};
    \node [sum, below of=sigma1] (alpha1){$\times\alpha^{\phantom{3}}$};
    \node [sum, below of=sigma2] (alpha2){$\times\alpha^2$};
    \node [sum, below of=sigmat] (alphat){$\times\alpha^t$};
   
    \draw [->] (sigma1) --  (alpha1);
    \draw [->] (sigma2) --  (alpha2);
    \draw [->] (sigmat) --  (alphat);
    \draw [->] (alpha1) -|(aiuto1) -- (sigma1);
    \draw [->] (alpha2) -|(aiuto2) -- (sigma2);
    \draw [->] (alphat) -|(aiutot) -- (sigmat);
    \draw [] (sigma1) -- (pall1); 
    \draw [->](pall1) -- (pall2);
    \draw [->, dashed](pall2)-. (pall3);
    \draw [->] (sigma2) -- (pall2);
    \draw [->] (sigmat) -- (pall3);
    \draw [->] (pall3) -- (pall4);
    \draw [->] (pall4) -- node{{\scriptsize $\quad output$}}(pall5);
\end{tikzpicture}
\caption{circuit for Chien's search}\label{fig:chien}
\end{center}
\end{figure*}
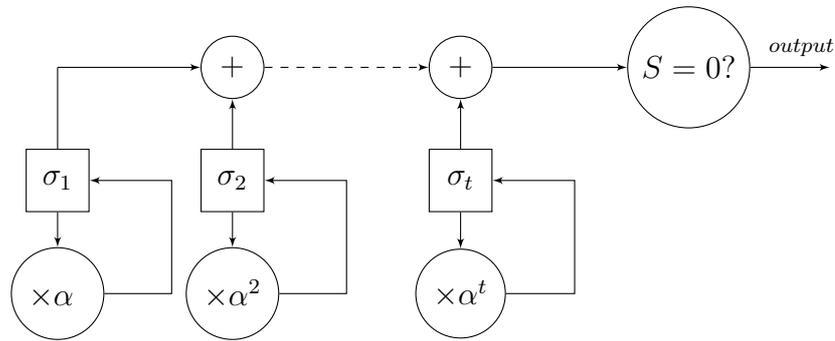

 \item \textbf{Finding the error values.}\\
 This task can be accomplished using Forney's formula (\ref{forney}), but throughout this thesis we will  also see other alternative  methods to compute error values. We observe that, depending on the decoding algorithm used, this last step may be implemented not only after the step 3, but also during the step 3: each time that Chien's search finds an error position $p_i$,  the decoding algorithm corrects the $p_i$th components of $\rw$ computing $E_i$ immediately. 
 \end{enumerate}

It is important to note that steps 1 and 3 of the general procedure involve only additions and multiplications by fixed elements (the first $d-1$ powers of $\alpha$) and their implementation is simple, so these steps requires a negligible computational costs compared with steps 2 and 4. Moreover steps 1 and 3  are essentially the same in all the decoding algorithms that we will study in this thesis. For these reasons they  will be not discussed over and they will be not counted in estimating of the computational cost of  the decoding algorithms that we will analyze later on.\\
 
\begin{figure*}[h!]
\begin{center}
\begin{tikzpicture}[auto, node distance=1cm,>=latex']
    \node [output](pall0){};
    \node [pall, right of=pall0, node distance=6cm](pall1){};
    \node [coordinate, right of=pall1, node distance=2cm](pall2){};
    \node [block, right of=pall2] (delay){\footnotesize Delay};
    \node [block, below of=pall2, node distance=2cm, text width=22mm](b0){\footnotesize Syndrome  Computation};
    \node [block, right of=b0, text width=2cm, node distance=3cm] (b1){\footnotesize Solving Key Equation};
    \node [block, right of=b1, text width=22mm, node distance=3cm](b2){\footnotesize Chien's Search \& Error values Computation};
    \node [sum, right of=delay, node distance=7cm](somma){$+$};
    \node [output, right of=somma, node distance=26mm](output){};
      
     \draw [->] (b0) --  (b1);
     \draw [->] (b1) --  (b2);
     \draw [->] (b2) -| node[pos=0.6]{$\e$} (somma);
     \draw [->] (delay) -- (somma);
     \draw [-] (pall1) -- (delay);
     \draw [->] (pall1) |- (b0);
    \draw [->] (pall0) -- node[text width=2cm]{{\footnotesize Received word} $\rw$}(pall1);
    \draw [->] (somma) -- node[text width=2cm]{\footnotesize Decoded word}(output);
\end{tikzpicture}
\caption{a block diagram for the general decoding procedure}\label{fig:procedure}
\end{center}
\end{figure*}
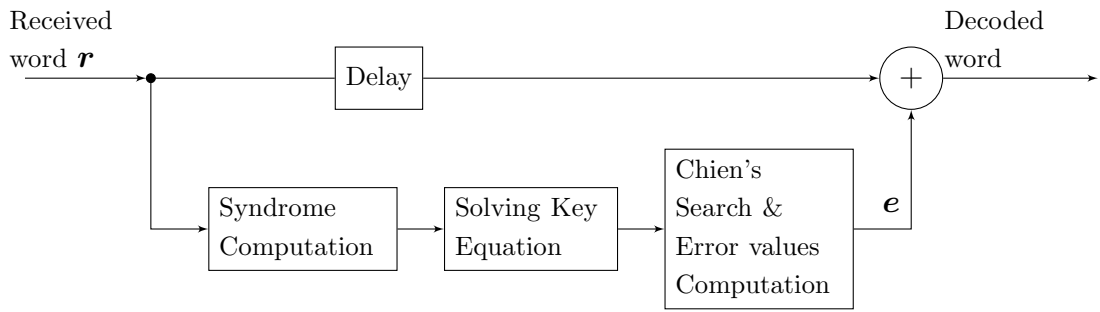

%
%

\chapter{Peterson-Gorenstein-Zierler Decoder} 
W. Wesley Peterson in 1960 developed the first practical decoder for Reed-Solomon codes based on syndrome computation only using tools of linear algebra. In particular, Peterson developed an algorithm for binary BCH codes that finds the location of the errors solving a linear system and Daniel E. Gorenstein and Neal Zierler, a year later, extended this algorithm for nonbinary codes.

In this chapter we will analyze the Peterson-Gorenstein-Zierler decoding algorithm  and we will present an important improvement \cite{schmidt} obtained using a fast inversion  method of an Hankel matrix. Finally we will  describe the necessary and sufficient conditions to avoid decoder malfunctions in the fast implementation of the Peterson-Gorenstein-Zierler decoder.

As discussed in section \ref{sect:decodingRS}, while studying the Peterson-Gorenstein-Zierler decoding algorithm  we will focus our attention only  in how it carries out  step 2 (computing the error-locator polynomial) and step 4 (computing the error values) of the general outline. 

\section{\PGZ Decoding Algorithm}\label{sect:PGZ}
Consider a code $RS(n,d,\alpha)$ over the field $\F$ and let $\cw$, $\e$ and $\rw=\cw+\e$ be the codeword sent, the error vector and the word received  in a data transmission, like in section \ref{sect:decodingRS}. We recall that $t$ represents the error correction capability of the code $RS(n,d,\alpha)$ (i.e.\ $t=\pint{\frac{d-1}{2}}$) and from now on  we assume that $e=\wt(\e)\leq t$. If $e>t$ we do not expect to be able to correct the errors.  By remark \ref{rem:sigmaproperties}, we can express the error-locator polynomial $\sigma(x)$ and the error-evaluator polynomial $\omega(x)$ as:
\begin{align*}
\sigma(x)&=\sigma_ex^e+\sigma_{e-1}x^{e-1}+\cdots+\sigma_1x+1\\
\omega(x)&=\omega_{e-1}x^{e-1}+\omega_{e-2}x^{e-2}+\cdots+\omega_1 x+\omega_0 
\end{align*}
The validity of the key equation (\ref{keyequation}) for the polynomials $\sigma(x)$ and $\omega(x)$ implies that the coefficients of $x^e,x^{e+1},\dots,x^{d-2}$ in the polynomial $\sigma(x)S(x)$ are equal to zero. Thus we have the following linear system of $d-1-e$ equations for the unknowns $\sigma_1, \sigma_2,\dots,\sigma_e$:
\begin{equation} \label{sistema1}
 \begin{cases}
  \displaystyle \sum_{i=1}^e S_{e+1-i}\,\sigma_i+S_{e+1}=0\\
  \displaystyle \sum_{i=1}^e S_{e+2-i}\,\sigma_i+S_{e+2}=0\\
  \displaystyle \cdots\\
  \displaystyle \sum_{i=1}^e S_{d-1-i}\,\sigma_i+S_{d-1}=0\\
 \end{cases}
\end{equation}
whose solution (if it exists and is unique) determines the coefficients of the error-locator polynomial $\sigma(x)$.
The linear system (\ref{sistema1}) may be expressed in matrix form as
\begin{equation} \label{sistema1Matrice}
 \begin{mymatrix}
        S_1          & S_2             & \cdots & S_{e} \\
        S_2          & S_3             & \cdots & S_{e+1}\\
        \vdots       & \vdots          &        & \vdots \\
        S_{d-1-e}    & S_{d-e}        & \cdots & S_{d-2}\\
 \end{mymatrix} 
 \begin{mymatrix} \sigma_e\\ \sigma_{e-1} \\ \vdots \\ \sigma_1\\
 \end{mymatrix}= - 
 \begin{mymatrix} S_{e+1}\\ S_{e+2} \\ \vdots \\ S_{d-1}\\
 \end{mymatrix} 
\end{equation}
and in order to study its solutions, we define the following useful matrices:

\begin{defi}\label{def:syndromematrix}
Let $t$ be the error correction capability of the code $RS(n,d,\alpha)$, we call
$$ A=\begin{mymatrix}
                        S_1          & S_2             & \cdots & S_{t+1} \\
                        S_2          & S_3             & \cdots & S_{t+2}\\
                        \vdots       & \vdots          &        & \vdots \\
                        S_{d-1-t}    & S_{d-t}        & \cdots & S_{d-1}\\
        \end{mymatrix}$$
the \emph{ syndrome matrix } and we indicate with $A_i$ its $i\times i$ leading principal minor. That is
$$ A_i=\begin{mymatrix}
                        S_1          & S_2             & \cdots & S_{i} \\
                        S_2          & S_3             & \cdots & S_{i+1}\\
                        \vdots       & \vdots          &        & \vdots \\
                        S_{i}    & S_{i+1}        & \cdots & S_{2i-1}\\
        \end{mymatrix}$$
\end{defi}

In the following proposition we study the properties of the syndrome matrix and of its minors $A_i$:
\begin{prop}\label{prop:syndromematrix}
Let $t$ be the error correction capability of $RS(n,d,\alpha)$ and $e=\wt(\e)$. If $e\leq t$, then
  \begin{enumerate}[\ \ \ (1)] 
  \item $A_e$ is a non-singular matrix;
  \item the rank of the syndrome matrix $A$ is equal to $e$;
  \item if $\rho=\max\{i\leq t \;|\; \det(A_i)\neq 0\}$, then the rank of $A$ is equal to $\rho$. So we have $e=\rho$.
 \end{enumerate}
\end{prop}
\begin{proof}
 \begin{enumerate}
  \item We consider the matrices
        $$ V=\begin{mymatrix}
                               1          &   1          & \cdots & 1 \\
                               X_1          & X_2            & \cdots & X_e\\
                               \vdots     & \vdots          &        & \vdots \\
                                X_1^{e-1}         &  X_2^{e-1}            & \cdots & X_e^{e-1}\\
            \end{mymatrix}\quad
        D=\begin{mymatrix}
                               E_1X_1       & 0                & \cdots       & 0 \\
                               0            & E_2X_2           &   \ddots     & \vdots \\
                               \vdots       & \ddots           &\ddots        & 0 \\
                               0            & \cdots           & 0            & E_eX_e\\
         \end{mymatrix}$$  
         Since the values $X_i$'s are all distinct and different from zero and  the value $E_i$'s are all different from zero, the Vandermonde matrix $V$ and the diagonal matrix $D$ are both  non-singular. The proof can be easily concluded verifying that $A_e=VDV^T$ by (\ref{rel:syndrome}).
  \item Since the matrix $A_e$ is a submatrix of $A$, it follows from (1) that $\rk(A)\geq e$. We consider the matrices
  $$ V_1=\begin{mymatrix}
                               1              & 1               & \cdots   & 1 \\
                               X_1            & X_2             & \cdots & X_e \\
                               \vdots         & \vdots          &        & \vdots \\
                               X_1^{d-2-t}    & X_2^{d-2-t}     & \cdots & X_e^{d-2-t}\\
        \end{mymatrix} \quad
       V_2=\begin{mymatrix}
                               1          & X_1            & \cdots & X_1^{t} \\
                               1          & X_2            & \cdots & X_2^{t}\\
                               \vdots     & \vdots          &        & \vdots \\
                               1          & X_e             & \cdots & X_e^{t}\\
        \end{mymatrix}
          $$  
        It is easy to verify that $A=V_1DV_2$. Recalling that the rank of a product of matrices is less or equal to the rank of each factor matrix and observing that $\rk(D)=e$, we obtain that $\rk(A)\leq e$. So the proof is complete.
  
  \item Since, by definition,  $\det(A_\rho)\neq 0 $, we have that $\rk(A)\geq \rho$. On the other hand, we saw in (1) that $\det(A_e)\neq 0$, therefore we have $e\leq\rho$. Hence, from (2) it follows that $\rk(A)\leq\rho$ and this concludes the proof.
 \end{enumerate}
\end{proof}

An immediate consequence is that, given the syndrome matrix $A$,  we can calculate the number of errors which occurred $e$ as
$$e=\max\{i\leq t \;|\; \det(A_i)\neq 0\}$$
Furthermore by proposition \ref{prop:syndromematrix} we know that there exists a unique solution  of  the linear system (\ref{sistema1}) and  the first $e$ equations are sufficient to calculate it.  In other words, the  coefficients $\sigma_1,\sigma_2,\dots,\sigma_e$ of the error-locator polynomial are completely determined by solving
\begin{equation} \label{sistema2Matrice}
 A_e\begin{mymatrix} \sigma_e\\ \sigma_{e-1} \\ \vdots \\ \sigma_1\\
   \end{mymatrix}= - 
   \begin{mymatrix} S_{e+1}\\ S_{e+2} \\ \vdots \\ S_{2e}\\
   \end{mymatrix} 
\end{equation}  
Once that the error-locator polynomial is known,  we will calculate its roots  and the error positions $p_1,p_2,\dots, p_e$ using (\ref{rel:X_i}) and Chien's search as seen in section \ref{sect:decodingRS}.\\

At this stage it remains to calculate the error values, knowing their positions. For this aim we consider the relations (\ref{rel:syndrome}) about the syndromes $S_1,S_2,\dots, S_{d-1}$. As the elements $X_1,X_2,\dots, X_e$ are known,  we obtain a set of $d-1$ linear equations in the unknowns $E_1,E_2,\dots,E_e$:
\begin{equation*}
\begin{cases}
 \displaystyle S_1=\sum_{j=1}^e E_j\,X_j\\
 \displaystyle S_2=\sum_{j=1}^e E_j\,X_j^2\\
 \cdots\\
 \displaystyle S_{d-1}=\sum_{j=1}^e E_j\,X_j^{d-1}
\end{cases}
\end{equation*}
The first $e$ equations can be solved for the error values $E_i$'s considering the following linear system:
\begin{equation} \label{syst:GZ}
 \begin{mymatrix} X_1   & X_2  & \cdots & X_e \\
                 X_1^2   & X_2^2  & \cdots & X_e^2 \\
                 \vdots &\vdots & &\vdots\\
                 X_1^e   & X_2^e  & \cdots & X_e^e \\
  \end{mymatrix}
 \begin{mymatrix} E_1\\ E_2 \\ \vdots \\ E_e\\
   \end{mymatrix}= 
   \begin{mymatrix} S_{1}\\ S_{2} \\ \vdots \\ S_{e}\\
   \end{mymatrix} 
\end{equation} 
The coefficient matrix of linear system  (\ref{syst:GZ}) has a special structure, indeed  it is the product of a Vandermonde matrix and a diagonal matrix:
\begin{equation*}
 \begin{mymatrix} X_1   & X_2  & \cdots & X_e \\
                 X_1^2   & X_2^2  & \cdots & X_e^2 \\
                 \vdots &\vdots & &\vdots\\
                 X_1^e   & X_2^e  & \cdots & X_e^e \\
  \end{mymatrix}=
  \underbrace{\begin{mymatrix}
                               1          &   1          & \cdots & 1 \\
                               X_1          & X_2            & \cdots & X_e\\
                               \vdots     & \vdots          &        & \vdots \\
                                X_1^{e-1}         &  X_2^{e-1}            & \cdots & X_e^{e-1}\\
  \end{mymatrix}}_{\textstyle V}
  \underbrace{\begin{mymatrix}
                               X_1       & 0                & \cdots       & 0 \\
                               0            & X_2           &   \ddots     & \vdots \\
                               \vdots       & \ddots           &\ddots        & 0 \\
                               0            & \cdots           & 0            & X_e\\
   \end{mymatrix}}_{\textstyle X}
\end{equation*}
Since the $X_i$'s are distinct and nonzero by definition, both the Vandermonde matrix $V$ and the diagonal matrix $X$ are invertible. Thus,
$$\begin{mymatrix}
   E_1\\ E_2 \\ \vdots \\E_e
  \end{mymatrix}=X^{-1}V^{-1}\underbrace{\begin{mymatrix}
   S_1\\ S_2 \\ \vdots \\S_e
  \end{mymatrix}}_{\textstyle \vct{s}}
$$
Since $V$ is a Vandermonde matrix, the step $V^{-1}\vct{s}$ can be efficiently carried out using the \emph{Bj\"{o}rck-Pereyra algorithm} for the solution of Vandermonde systems, while the step $X^{-1}\left(V^{-1}\vct{s}\right)$ consists simply in $e$ divisions. The Bj\"{o}rck-Pereyra algorithm transforms the right side vector of a Vandermonde system by a sequence of simple transformations into the solution vector. The components are modified in a suitable order so that no extra storage is needed. For more details and a complete exposition of the Bj\"{o}rck-Pereyra algorithm, we refer to \cite{bjorck}. In the following we present an algorithm that first compute the vector $\vct{y}=(y_0,y_1,\dots,y_{e-1})^T\in(\F)^e$ such that $\vct{y}=V^{-1}\vct{s}$ (implementing the Bj\"{o}rck-Pereyra algorithm for the matrix $V$ and the right side vector $\vct{s}$) and after it solves the diagonal system. 

\begin{alg}[\textbf{\BP algorithm for computing error values}]\hspace{2cm}
\label{alg:BP}
\begin{description}
 \item[\texttt{Input}:] the number of error $e$, the elements $X_1, X_2, \dots, X_e$ and the syndrome vector  $\vct{s}=\left(S_1,S_2,\dots,S_{e}\right)^T$;\\[-0.8cm]
 \item[\texttt{Output}:] the error values $E_1, E_2,\dots,E_e$;
\end{description}
\vspace*{-0.2cm}
\texttt{Begin}
\begin{enumerate}[\bfseries \BP.1]
  \item  \begin{tabbing}
         \textit{(finding  $\vct{y}=V^{-1}\vct{s}$)}\\
         \=$\vct{y}:=\vct{s}$;\\
         \>\texttt{for }\= $k=1,2,\dots, e-1$ \texttt{ do }  \\
         \>\>\texttt{for }\= $i=e-1, e-2, \dots, k$ \texttt{ do }  \\
                    \>\> \>$y_i:=y_i-X_ky_{i-1}$;\\
         \>\>\texttt{endfor}\\
         \>\texttt{endfor}\\ \\
         \>\texttt{for }\= $k=e-1, e-2, \dots, 1$ \texttt{ do }  \\
         \>\>\texttt{for }\= $i=k, k+1, \dots, e-1$ \texttt{ do } \\
                    \>\> \>$y_i:=\dfrac{y_i}{X_{i+1}-X_{i+1-k}}$;\\
         \>\>\texttt{endfor}\\
         \>\>\texttt{for }\= $i=k-1, k, \dots, e-2$ \texttt{ do }  \\
                    \>\> \>$y_i:=y_i-y_{i+1}$;\\
        \>\>\texttt{endfor}\\
        \>\texttt{endfor}\\
    \end{tabbing}
  \item \begin{tabbing}
         \textit{(solving the diagonal system)}\\
         \texttt{for }\= $i=1,2,\dots, e$ \texttt{ do } \\
            \>$E_i:=y_{i-1}X_i^{-1}$;
        \end{tabbing}
 \end{enumerate}
\texttt{End}
\end{alg}

By direct counting, it is easy to verify that algorithm \ref{alg:BP} has a computational time complexity of $O(e^2)$ arithmetic operations.\\

\noindent With this observations, the Peterson-Gorenstein-Zierler (\PGZ) decoding algorithm consists of the followings steps:\\
\begin{alg}[\textbf{\PGZ decoding algorithm} for $\RS$]\hspace{2cm}
\label{alg:PGZ}
\begin{description}
 \item[\texttt{Input}:] the received word $\rw(x)$;
 \item[\texttt{Output}:] the  codeword $\cw(x)$;
\end{description}
\texttt{Begin}
\begin{enumerate}[\bfseries \PGZ.1]
         \item \label{pgz1}
               \begin{tabbing}
                \textit{(syndrome computation)}\\
                \texttt{for }\= $i=1,2,\dots,d-1$ \texttt{ do }\\
                \> $S_i:=\rw(\alpha^i)$;\\
                \texttt{endfor}               
               \end{tabbing}
  
  \item  \label{pgz2}
         \begin{tabbing}
          \textit{(error-locator polynomial computation)}\\
          \=$i:=\pint{\frac{d-1}{2}}$;\\
          \>$det:=\det(A_i)$;\\
          \>\texttt{whi}\=\texttt{le }  $(det=0)$ \texttt{ repeat} \\
          \>\>$i:=i-1$;\\
          \>\>$det:=\det(A_i)$;\\
          \>\texttt{endwhile};\\
          \> $e:=i$;
        \end{tabbing}
        \texttt{solve the linear system (\ref{sistema2Matrice}) to find} $\sigma_1, \sigma_2, \dots,\sigma_e$;\\
        $\sigma(x):=\sigma_ex^e+\sigma_{e-1}x^{e-1}+\cdots+\sigma_1x+1$;
  
  \item \label{pgz3}
        \textit{(finding error positions)}\\
        \texttt{calculate  the error positions $p_1,p_2,\dots,p_e$ and the elements $X_1^{-1}, X_2^{-1},\dots,X_e^{-1}$ using Chien's search (alg.\ \ref{alg:chien});}
        
  \item \label{pgz4}
        \textit{(finding error values)}\\
        \texttt{solve linear system (\ref{syst:GZ}) with algorithm \ref{alg:BP}};\\[0,1cm]
        $\begin{mymatrix}
   E_1\\ E_2 \\ \vdots \\E_e
  \end{mymatrix}=X^{-1}V^{-1}\begin{mymatrix}
   S_1\\ S_2 \\ \vdots \\S_e
  \end{mymatrix}$
 \end{enumerate}
\texttt{Return} $\quad \cw(x):=\displaystyle \rw(x)-\sum_{i=1}^e E_{i}x^{p_i}$;\\
\texttt{End}
\end{alg}

The correctness of the \PGZ decoding algorithm \ref{alg:PGZ} follows immediately from proposition \ref{prop:syndromematrix}. We are now interested in estimating its computational cost. 
The calculation of $e$ in \refpgz{pgz2} is done by  testing the non-singularity of $A_i$, beginning from $i=t$ and decrementing $i$. Since the computational complexity of $\det(A_i)$ is $O(i^3)$ and 
$$\sum_{i=e}^t i^3= \left(\frac{t(t+1)}{2}\right)^2-\left(\frac{(e-1)e}{2}\right)^2$$
then calculating $e$ has a computational cost of order $O(t^4)$.
The computational cost of finding the coefficients $\sigma_1,\sigma_2,\dots,\sigma_e$ is the cost of solving an $e\times e$ linear system, that is  $O(e^3)$, whereas the linear system solved in \refpgz{pgz4} has a computational complexity of $O(e^2)$ operations due to algorithm \ref{alg:BP}. However the total computational time complexity of \PGZ decoding algorithm \ref{alg:PGZ} is of order $O(t^4)$ and this makes it not applicable for large $t$.\\

We also observe  that in \PGZ decoder the steps 3 and 4 of the general outline are sequential. Indeed step \refpgz{pgz4} can be implemented only after the end of \refpgz{pgz3}, because algorithm \ref{alg:BP} used in  \refpgz{pgz3} to calculate the error values needs all the elements $X_i$'s together. In chapter 4 we will see a modified version of \PGZ decoder with a different strategy for the error value computation.


\section{\fPGZ Decoding Algorithm} \label{sect:fastPGZ}

Although the \PGZ decoding algorithm \ref{alg:PGZ} involves simple linear algebra, its implementation is computationally expensive in general. For this reason in this section we will study and formalize an alternative implementation of the \PGZ decoder, shown in \cite{schmidt}, in which some special proprieties of the syndrome matrix are used to decrease the computational cost of finding the number of error $e$ and solving the linear system (\ref{sistema2Matrice}) (step \refpgz{pgz2}).\\

We observe that the syndrome matrix $A$ is an \emph{Hankel matrix}, that is a matrix with constant positive sloping diagonals. Techniques for fast inversion of Hankel matrices are well known (see \cite{heinigRost}) and this section we  use them to calculate the rank of the syndrome matrix $A$ and the solutions of the linear system (\ref{sistema2Matrice}). For a more convenient notation, we will indicate with $\vct{0}^{(j)}$ the zero vector of $j$ components, for any $j\in \N$. 

The algorithm that we present is iterative and is based on the following idea: 
if for an index $i\leq t$ we have $\det(A_i)\neq 0$ and we already know the vectors $\vct{w}^{(i)}$ and $\vct{y}^{(i)}$ in $(\F)^i$ such that
 \begin{equation}\label{fastPGZw}
 A_i\vct{w}^{(i)}=-\begin{mymatrix} S_{i+1}\\ S_{i+2}\\\vdots \\ S_{2i}\\
   \end{mymatrix}
 \end{equation}
 \begin{equation}\label{fastPGZy}
  A_i\vct{y}^{(i)}=\begin{mymatrix}\vct{0}^{(i-1)}\\1\\
   \end{mymatrix}
 \end{equation}
then, using the Hankel form of the syndrome matrix, we can calculate the first index $j>i$ such that $\det(A_j)\neq 0$ and the vectors $\vct{w}^{(j)}$ and $\vct{y}^{(j)}$ in $(\F)^j$ such that the systems (\ref{fastPGZw}) and (\ref{fastPGZy}) also hold  for $i=j$.
Clearly, if we are able to do this step, repeating it,  we can implement an iterative algorithm which calculates the maximum index $\rho\leq t$ such that $\det(A_{\rho})\neq 0$ and the vectors $\vct{w}^{(\rho)}$ and $\vct{y}^{(\rho)}$ that solve (\ref{fastPGZw}) and (\ref{fastPGZy}) for $i=\rho$. As we saw in proposition \ref{prop:syndromematrix}, supposing  $e\leq t$, we have that the index $\rho$ is equal to $e$. So the vector $\vct{w}^{(\rho)}$ given as output is the unique vector that satisfies the linear system (\ref{sistema2Matrice}). In other words $$\vct{w}^{(\rho)}=\begin{mymatrix}                                                                                                                                                                                                                                                                                                                                                                                                                                           
                                                                                                                                                         
                                \sigma_e \\ \sigma_{e-1}\\ \vdots \\ \sigma_1                                                                                                                                                                                                                                                                                                                                                                                                                                                                                                                                                                   \end{mymatrix}$$
We will call the \emph{fast Peterson-Gorenstein-Zierler} (\fPGZ) \emph{decoding algorithm} the decoding algorithm obtained by solving the linear system (\ref{sistema2Matrice}) with this iterative procedure.\\

In order to study the iterative step on which is based the \fPGZ decoding algorithm, we need the following definitions and a propriety of Hankel matrices.
\begin{defi} \label{defi:epsilon}
 Let $i\leq t$ such that $\det(A_i)\neq0$.
 If $$\displaystyle\vct{w}^{(i)}=\begin{mymatrix}
                                   w^{(i)}_0\\w^{(i)}_1\\ \vdots \\ w^{(i)}_{i-1}
                                 \end{mymatrix} \in (\F)^i$$
 is the solution of system (\ref{fastPGZw}), then for any $j\in\{1,2,\dots,t-i\}$ we define 
 $$\varepsilon_j=S_{i+j}w_0^{(i)}+S_{i+j+1}w_1^{(i)}+\cdots+S_{2i+j-1}w_{i-1}^{(i)}+S_{2i+j}\in\F $$ 
 Furthermore if the elements $\varepsilon_j$'s are not all equal to zero, then we call
 $$r=\min\left\{j\in\{1,2,\dots,t-i\} \;|\; \varepsilon_j\neq 0\right\}\in\N^+$$
 the \emph{singularity gap} of the step $i$.
\end{defi}

The following lemma is stated and proved in section 1 chapter 1  of \cite{heinigRost} for Hankel matrices with entries in $\mathbb C$, but it is easy to verify that the same proof  also holds for Hankel matrices with entries in the field $\F$.
\begin{lem}\label{lemma:hankelinvertibile}
 Let  $$M=\begin{mymatrix} 
       m_1 & m_2 & \cdots & m_l\\
       m_2 & m_3 & \cdots & m_{l+1}\\
       \vdots &\vdots & &\vdots \\
       m_l & m_{l+1} & \cdots&  m_{2l-1}\\
     \end{mymatrix}\in (\F)^{l\times l}$$ 
 a general $l\times l$ Hankel matrix. If there exist $m\in\F$,  $\vct{w}\in (\F)^l$ and $\vct{y}\in(\F)^l$  such that
 $$M\vct{w}=\begin{mymatrix} m_{l+1}\\m_{l+2}\\ \vdots \\ m_{2l-1}\\m\\
   \end{mymatrix} \quad \text{ and } \quad
  M\vct{y}=\begin{mymatrix}\vct{0}^{(i-1)}\\1\\
   \end{mymatrix}$$ 
 then $M$ is non-singular.
\end{lem}

At this stage we can  enunciate the theorem that shows formally how the iterative step of \fPGZ  decoding algorithm works. The proof of this theorem is a constructive proof and gives the formulas necessary to implement the algorithm.
\begin{thm}\label{thm:fastPGZ}
Let $i$ be an index less or equal to $t$ such that $\det(A_i)\neq 0$. Let $\vct{w}^{(i)}$ and $\vct{y}^{(i)}$ be the vectors in $(\F)^i$ that solve respectively (\ref{fastPGZw}) and (\ref{fastPGZy}). Recalling the notation of definition \ref{defi:epsilon}, we have that:
\begin{enumerate}[\ \ \ (1)]
 \item If $\varepsilon_j=0$ for all $j\in\{1,2,\dots,t-i\}$, then 
           $$i=\max\{j\leq t \;|\; \det(A_j)\neq 0\}$$
 \item If there exists $j\in\{1,2,\dots,t-i\}$ such that $\varepsilon_j\neq0$ and $r$ is the singularity gap, then we have that
           $$r=\min\left\{j\in\{1,2,\dots,t-i\} \;|\; \det(A_{i+j})\neq 0\right\}$$
\end{enumerate}
\end{thm}
\begin{proof} 
By the linear system (\ref{fastPGZw}) and by definition \ref{defi:epsilon} it follows that for any $j\in\{1,2,\dots,t-i\}$:
$$ A_{i+j}\begin{mymatrix}\vct{w}^{(i)}\\ 1\\ \vct{0}^{(j-1)}\\ 
          \end{mymatrix}=
          \begin{mymatrix} \vct{0}^{(i)}\\ \varepsilon_1 \\ \vdots\\ \varepsilon_j\\
          \end{mymatrix}$$
If the $\varepsilon_j$'s are all equal to zero, then $\ker(A_{i+j}) \neq\{\vct{0}\}$ for any $j\in\{1,2,\dots,t-i\}$.
In particular $\det(A_{i+j})=0$ for any $j\in\{1,2,\dots,t-i\}$ and hence (1) is proved.
If not all the $\varepsilon_j$'s are zero, then  by definition of $r$ we have that $\varepsilon_1 = \varepsilon_2=\cdots =\varepsilon_{r-1}=0$. Hence $\det(A_{i+j})=0$ for any $j\in\{1,2,\dots,r-1\}$ and to complete the proof of (2) it is sufficient  to show that $\det(A_{i+r})\neq0$.   By lemma \ref{lemma:hankelinvertibile}, it is sufficient to show the existence of vectors $\vct{w}^{(i+r)}$ and $\vct{y}^{(i+r)}$ in $(\F)^{i+r}$ such that
 \begin{equation}\label{fastPGZw2}
     A_{i+r}\vct{w}^{(i+r)}=-\begin{mymatrix} 
                                S_{i+r+1}\\ S_{i+r+2}\\\vdots \\ S_{2i+2r}\\
                             \end{mymatrix}
 \end{equation}
 \begin{equation}\label{fastPGZy2}                
     A_{i+r}\vct{y}^{(i+r)}=\begin{mymatrix} 
                               \vct{0}^{(i+r-1)}\\ 1\\
                             \end{mymatrix}
 \end{equation}
Since $ A_{i+r}\begin{mymatrix} \vct{w}^{(i)}\\ 1\\ \vct{0}^{(r-1)}\\ 
            \end{mymatrix}=
            \begin{mymatrix} \vct{0}^{(i+r-1)}\\ \varepsilon_r 
            \end{mymatrix}$
and $\varepsilon_r\neq 0$, it follows that the vector
\begin{equation}\label{formula:y^i+r}
\vct{y}^{(i+r)}=\frac{1}{\varepsilon_r}\begin{mymatrix} \vct{w}^{(i)}\\1\\\vct{0}^{(r-1)}
                                                 \end{mymatrix}\in (\F)^{i+r} 
\end{equation}
is well defined  and satisfies the system (\ref{fastPGZy2}).
Moreover we define the vectors $\vct{a}^{(0)},\vct{a}^{(1)},\dots, \vct{a}^{(r-1)}\in (\F)^ {i+r}$ as
 \begin{equation}\label{formulaA}
   \begin{cases}
       \vct{a}^{(0)}=\vct{y}^{(i+r)}\\
       \vct{a}^{(j)}=\displaystyle \begin{mymatrix}
                                     0\\ a^{(j-1)}_{0}\\ a^{(j-1)}_{1}\\\vdots \\ a^{(j-1)}_{i+j-1}\\ \vct{0}^{(r-j-1)}\\ 
                                   \end{mymatrix}
                     -\alpha_j\vct{a}^{(0)}                  \quad \quad\forall\, j=1,2, \dots,r-1       
   \end{cases}
 \end{equation}
where
\begin{equation}\label{rel:alpha}
 \alpha_j=S_{i+r+1}a_0^{(j-1)}+S_{i+r+2}a_1^{(j-1)}+\cdots+S_{2i+r+j}a_{i+j-1}^{(j-1)}
\end{equation}
It can be easily proved by induction that, for any $j\in\{0,1,\dots,r-1\}$, the product $A_{i+r}\vct{a}^{(j)}$ is equal to the vector with all the components zero except the component $i+r-j$, which is equal to 1. Indeed the base case is implied by (\ref{fastPGZy2}) and if
$$A_{i+r}\vct{a}^{(j)}=\begin{mymatrix}
                          \vct{0}^{(i+r-j-1)}\\ 1\\ \vct{0}^{(j)}
                       \end{mymatrix}$$
then            
 \begin{align*}
 \displaystyle
 A_{i+r}\vct{a}^{(j+1)}&=A_{i+r}\begin{mymatrix}
                                   0\\ a^{(j)}_{0}\\ a^{(j)}_{1}\\ \vdots \\ a^{(j)}_{i+j}\\ \vct{0}^{(r-j-2)}\\ 
                                 \end{mymatrix}-\alpha_{j+1}A_{i+r}\vct{a}^{(0)}\; =\\ \\
                       &=\begin{mymatrix}
                             S_2 & S_3 & \cdots & S_{i+j+2}\\
                             S_3 & S_4 & \cdots & S_{i+j+3}\\
                             \vdots &\vdots & &\vdots\\
                             S_{i+r} & S_{i+r+1} & \cdots & S_{2i+j+r}\\
                         \end{mymatrix}
                         \begin{mymatrix}
                             a^{(j)}_{0}\\a^{(j)}_{1}\\ \vdots \\ a^{(j)}_{i+j}\\ 
                         \end{mymatrix}-\alpha_{j+1}A_{i+r}\vct{y}^{(i+r)}
                  \overset{\displaystyle\overset{\overset{\displaystyle \text{{\tiny inductive}}}{\displaystyle \text{{\tiny hypothesis}}}}
                  {\big\downarrow}}{=}\\ \\
                      &=\begin{mymatrix}
                           \vct{0}^{(i+r-j-2)}\\ 1\\ \vct{0}^{(j)}\\ \alpha_{j+1}
                        \end{mymatrix}-\alpha_{j+1}\begin{mymatrix} 
                                                      \vct{0}^{(i+r-1)}\\ 1\\
                                                     \end{mymatrix}=
                                                  \begin{mymatrix}
                                                   \vct{0}^{(i+r-j-2)}\\ 1\\ \vct{0}^{(j+1)}\\                                  
                                                   \end{mymatrix} \\
\end{align*}
In a similar way, we recursively define the vectors $\vct{b}^{(0)},\vct{b}^{(1)},\dots, \vct{b}^{(r)}$ in $(\F)^ {i}$ as
 \begin{equation}\label{formulaB}
   \begin{cases}
      \vct{b}^{(0)}=\vct{w}^{(i)}\\
      \vct{b}^{(j)}=\begin{mymatrix}
                        0 \\ b_0^{(j-1)}\\  b_1^{(j-1)}\\\vdots \\ b_{i-2}^{(j-1)}   \\                                 
                      \end{mymatrix}
                    - b_{i-1}^{(j-1)}\vct{w}^{(i)}-\beta_j \vct{y}^{(i)}  \quad \quad\forall\, j=1,2, \dots,r
   \end{cases}
 \end{equation}
where 
\begin{equation}\label{rel:beta}
\beta_j=S_{i+1}b_0^{(j-1)}+S_{i+2}b_1^{(j-1)}\cdots+S_{2i}b_{i-1}^{(j-1)}+S_{2i+j} 
\end{equation}
It can be proved, again by induction, that for any $j\in\{0,1,\dots,r\}$ we have 
$$A_i\vct{b}^{(j)}=-\begin{mymatrix}
                        S_{i+j+1} \\S_{i+j+2}\\ \vdots \\S_{2i+j}
                    \end{mymatrix}$$
Finally, for any $j\in\{1,2,\dots,r\}$ we define 
\begin{equation}\label{rel:gamma}
\gamma_j=S_{i+j}b_0^{(r)}+S_{i+j+1}b_1^{(r)}+\cdots+ S_{2i+j-1}b_{i-1}^{(r)} 
\end{equation}
and we observe that 
        $$A_{i+r}\begin{mymatrix}
                  \vct{b}^{(r)} \\ 0\\ \vdots \\0 \\
                 \end{mymatrix}=
                 \begin{mymatrix}
                  -S_{i+r+1} \\ -S_{i+r+2}\\\vdots \\-S_{2i+r}\\ \gamma_1\\ \vdots\\ \gamma_r
                 \end{mymatrix}$$       
Thus we conclude the proof defining the vector $\vct{w}^{(i+r)}$ as
 \begin{equation}\label{formula:w^i+r}
         \vct{w}^{(i+r)}=\begin{mymatrix}
                              \vct{b}^{(r)} \\ \vct{0}^{(r)} \\
                           \end{mymatrix} -
                           \displaystyle\sum_{l=0}^{r-1} \left[\gamma_{r-l}+S_{2i+2r-l}\right] \vct{a}^{(l)}
 \end{equation}
\end{proof}

In  order to work, the \fPGZ  decoding algorithm needs a base step. For this reason we observe that:
\begin{rem}\label{rem:i0}
If $e=\wt(\e)\leq t$, then either $S_1=S_2=\cdots=S_{d-1}=0$ or there exists $i\leq e$ such that $S_i\neq 0$. In fact, if not all the syndromes were zero and $S_1=S_2=\cdots=S_e=0$ held, then we would have at the same time that $\det(A_e)\neq0$ (by proposition \ref{prop:syndromematrix}) and that all the entries of the first row of $A_e$ are equal to zero. This obviously is absurd. 
 
Thus we can define \begin{equation*} 
i_0=\min\left\{j\in\N^+ \;|\; S_{j}\neq 0\right\}                                                               
\end{equation*}
and we have that  $$A_{i_0}=\begin{mymatrix}
                        0          &  \cdots      & 0          &0           & S_{i_0} \\
                        0          &  \cdots      & 0          & S_{i_0}    & S_{i_0+1}\\
                        \vdots     &              & \vdots     & \vdots     & \vdots \\
                        S_{i_0}    &  \cdots      & S_{2i_0-3} & S_{2i_0-2} & S_{2i_0-1}\\
        \end{mymatrix}$$
is a non-singular lower triangular matrix. So 
$$\vct{y}^{(i_0)}=\begin{mymatrix}
                   S_{i_0}^{-1}\\
                   \vct{0}^{i_0-1}
                  \end{mymatrix}$$
satisfies (\ref{fastPGZy}) for $i=i_0$ and the vector  $\vct{w}^{(i_0)}$, which satisfies (\ref{fastPGZw}) for $i=i_0$, can be calculated with linear algebra algorithms for triangular matrix of complexity $O((i_0)^2)$.\\
\end{rem}

\noindent We are now ready to summarize the \fPGZ  decoding algorithm as follows.
Since when $r=1$  the formula (\ref{formula:w^i+r}) to calculate $\vct{w}^{(i+r)}$ is simpler than in the general case, in the following decoding algorithm we prefer to consider separately the case $\varepsilon_1\neq 0$.\\

\begin{alg}[\textbf{\fPGZ  decoding algorithm} for $\RS$]\hspace{2cm}
\label{alg:fastPGZ}
\begin{description}
 \item[\texttt{Input}:]  the received word $\rw(x)$;
 \item[\texttt{Output}:] the codeword $\cw(x)$;
\end{description}
\texttt{Begin}
\begin{enumerate}[\bfseries \fPGZ.1]
        \item \label{fpgz1}
               \begin{tabbing}
               \textit{(syndrome computation)}\\
               \texttt{for }\= $i=1,2,\dots, d-1$ \texttt{ do } \\
               \> $S_i:=\rw(\alpha^i)$;\\
               \texttt{endfor}               
              \end{tabbing}

        \item \label{fpgz2}
              \textit{(error-locator polynomial computation)}
              \begin{enumerate}[(a)]
                \item  \begin{tabbing}   
                          {\small\textit{(base step)}}\\              
                          $t:=\pint{\frac{d-1}{2}}$;\\
                          $i_0:=1$;\\  
                          \texttt{whi}\=\texttt{le} ($S_{i_0}=0$ \texttt{and} $i_0\leq d-1$) \texttt{repeat} \\
                          \> $i_0:=i_0+1$;\\
                          \texttt{endwhile} \\
                          \texttt{if} $i_0=d$ \texttt{then return} $\rw$;\\
                          \texttt{if} $i_0>t$ \texttt{then declare a failure};\\
                          \texttt{calculate}   $\vct{w}^{(i_0)}$ \texttt{and}  $\vct{y}^{(i_0)}$;                          
                       \end{tabbing}
                \item  \begin{tabbing}  
                         {\small\textit{(iterative procedure)}}\\
                         $i:=i_0$;\\
                         \texttt{whi}\=\texttt{le} ($i<t$) \texttt{repeat}\\
                        \>$\varepsilon_1:=S_{i+1}w_0^{(i)}+S_{i+2}w_1^{(i)}+\cdots+S_{2i}w_{i-1}^{(i)}+S_{2i+1}$;\\
                         \>\texttt{if} \=($\varepsilon_1\neq0$)  \texttt{then}\\[0,2cm]
                         \>\>$\vct{y}^{(i+1)}:=\frac{1}{\varepsilon_1}\begin{mymatrix} \vct{w}^{(i)}\\1\\ \end{mymatrix}$;\\[0,2cm]
                         \>\>$\eta:=S_{i+1}y_0^{(i)}+S_{i+2}y_1^{(i)}+\cdots+S_{2i}y_{i-1}^{(i)}$;\\[0,2cm]  
                         \>\>$\varepsilon_2:=S_{i+2}w_0^{(i)}+S_{i+3}w_1^{(i)}+\cdots+S_{2i+1}w_{i-1}^{(i)}+
                             S_{2i+2}$;\\[0,2cm] 
                         \>\>$\vct{w}^{(i+1)}:=\begin{mymatrix} 0\\ \vct{w}^{(i)}\\    \end{mymatrix}-\varepsilon_1\begin{mymatrix} \vct{y}^{(i)}\\0\\ \end{mymatrix}+(\varepsilon_1\eta-\varepsilon_2) \vct{y}^{(i+1)}$;\\                                                                                                                                       \>\>$i:=i+1$;\\
                         \>\texttt{else}\=\\
                         \>\>$r:=1$; \\
                         \>\> \texttt{whi}\=\texttt{le} ($\varepsilon_r=0$) \texttt{repeat} \\
                         \>\>\>$r:=r+1$;\\ 
                         \>\>\>\texttt{if} \= $r>t-i$ \texttt{then} \\
                         \>\>\>\> \texttt{go to step c};\\
                         \>\>\> \texttt{else }  $\varepsilon_r:=S_{i+r}w_0^{(i)}+S_{i+r+1}w_1^{(i)}+\cdots+S_{2i+r-1}w_{i-1}^{(i)}+S_{2i+r}$;\\
                         \>\> \texttt{endwhile}\\
                         \>\> $\vct{y}^{(i+r)}:=\frac{1}{\varepsilon_r}\begin{mymatrix} \vct{w}^{(i)}\\1\\                        						\vct{0}^{(r-1)} \end{mymatrix}$;\\
                         \>\> \texttt{calculate  vectors  $\vct{a}^{(j)}$'s  using  (\ref{formulaA})};\\
                         \>\> \texttt{calculate  vectors  $\vct{b}^{(j)}$'s  using  (\ref{formulaB})};\\
                         \>\> $l:=1$;\\
                         \>\> \texttt{for } \= $l=1,2,\dots, r$\\
                         \>\>\>$\gamma_l:=S_{i+l}b_0^{(r)}+S_{i+l+1}b_1^{(r)}+\cdots+ S_{2i+l-1}b_{i-1}^{(r)}$;\\
                         \>\> \texttt{endfor}\\
                         \>\> $\vct{w}^{(i+r)}:=\begin{mymatrix} \vct{b}^{(r)} \\ \vct{0}^{(r)} \\
                           \end{mymatrix} -\displaystyle\sum_{l=0}^{r-1} \left[\gamma_{r-l}+S_{2i+2r-l}\right] \vct{a}^{(l)}$;\\
                         \>\>$i:=i+r$;\\                  
                         \texttt{endwhile}
                          \end{tabbing}
                   \item  $e:=i$;\\
                         $\sigma(x):=w_0^{(i)}x^i+w_1^{(i)}x^{i-1}+\cdots+w_{i-1}^{(i)}x+1$;
                   \end{enumerate}                                                  
      \item \label{fpgz3}
            \textit{(finding error positions)}\\
           \texttt{calculate  the error positions $p_1,p_2,\dots,p_e$ and the elements $X_1^{-1}, X_2^{-1},\dots,X_e^{-1}$ using Chien's search (alg.\ \ref{alg:chien});}
      \item \label{fpgz4}
        \textit{(finding error values)}\\
        \texttt{solve linear system (\ref{syst:GZ}) with algorithm \ref{alg:BP}};\\[.2cm]
         $\begin{mymatrix}
    E_1\\ E_2 \\ \vdots \\E_e
   \end{mymatrix}=X^{-1}V^{-1}\begin{mymatrix}
    S_1\\ S_2 \\ \vdots \\S_e
   \end{mymatrix}$
 \end{enumerate}
\texttt{Return} $\quad \cw(x):=\displaystyle \rw(x)-\sum_{i=1}^e E_{i}x^{p_i}$;\\
\texttt{End}
\end{alg}

From theorem \ref{thm:fastPGZ} and from what we have seen until now, it follows that if $e\leq t$, then the \fPGZ  decoding algorithm \ref{alg:fastPGZ} correctly gives as output the transmitted codeword $\cw$. Now we will estimate its computational complexity. We consider step \reffpgz{fpgz2}, by remark \ref{rem:i0} it follows that the computational cost of (a) is of order $O(e^2)$, thus we can focus our attention on the iterative step from $i$ to $i+r$ in (b):
\begin{itemize}
 \item the computation of the gap $r$ and of the elements $\varepsilon_j$'s needs at most $ir$ multiplications and $ir$ additions      
 \item from (\ref{formula:y^i+r}) we deduce that the vector $\vct{y}^{(i+r)}$ can be calculated from $\vct{w}^{(i)}$ with one inversion and $i$ multiplications;
 \item from (\ref{rel:alpha}) we find that the computational cost of the elements  $\alpha_j$'s is upper bounded by $ir+r^2$ multiplications and $ir+r^2$ additions, while from (\ref{rel:beta}) and (\ref{rel:gamma}) we find that the  elements  $\beta_j$'s and $\gamma_j$'s are computed with at most $2ir$ multiplications and $2ir$ additions;
\item from (\ref{formulaA}) we deduce that the cost of the vectors $\vct{a}^{(j)}$'s is at most  $ir+r$ multiplications and $ir+r^2$ additions, while  (\ref{formulaB}) implies that the vectors  $\vct{b}^{(j)}$'s can be calculated with $ir$ multiplications and $ir$ additions;
\item finally (\ref{formula:w^i+r}) implies that the number of operations necessary to calculate the vector $\vct{w}^{(i+r)}$ from the vectors $\vct{b}^{(r)}$ and $\vct{a}^{(j)}$'s is upper bounded by $ir+r^2$ multiplications and $ir+r^2+r+i$ additions.
\end{itemize}
Considering that $i,r\leq t$, the iterative step from $i$ to $i+r$ requires at most one inversion, $9tr+t+r$ multiplications and $10tr+t+r$ additions. Evidently the total computational cost of (b) in \reffpgz{fpgz2} is given by the sum of the costs of the iterative steps and since the sum of the singularity gaps is clearly equal to $e-i_0$ (i.e.\ less or equal to $e$), we can conclude that in (b) the total number of operations is 
\begin{center}
 \begin{tabular}{ll}
  $e$ & inversions\\
  $10et+e$ & multiplications\\
  $11et+e$ & additions\\
 \end{tabular}
\end{center}
Step \reffpgz{fpgz4} is the same of \refpgz{pgz4} and has a computational cost is  of $O(e^2)$ operations. Hence the computational complexity of the \fPGZ  decoding algorithm \ref{alg:fastPGZ} is $O(et)$, much better than in the case of the \PGZ decoding algorithm \ref{alg:PGZ}.\\

To conclude we give an example of decoding of a Reed-Solomon code using the \fPGZ  decoding algorithm \ref{alg:fastPGZ}:
 \begin{ex}\label{ex_fastpgz} 
  Let $\alpha$ be a primitive element of $\mathbb{F}_{2^4}$ satisfying $\alpha^4+\alpha+1=0$. Consider  over $\mathbb{F}_{2^4}$ the code, $RS(15,9,\alpha)$ generated by 
  $$g(x)=(x-\alpha)(x-\alpha^2)\cdots(x-\alpha^8)$$
  The code has distance $9$, so $t=4$.  Suppose that the codeword sent is $\cw=\vct{0}$ and the error vector is
  $$\e(x)=\alpha^2x^2+\alpha x^8+\alpha^7x^{13}$$
  Clearly $\rw(x)=\e(x)$ and step \reffpgz{fpgz1} reads out the syndromes
  $$\begin{array}{llll}
    S_1=\rw(\alpha)=\alpha^{12} & S_2=\rw(\alpha^2)=0    &S_3=\rw(\alpha^3)=0  & S_4=\rw(\alpha^4)=\alpha^{5}\\
    S_5=\rw(\alpha^5)=\alpha^{11}   & S_6=\rw(\alpha^6)=\alpha^{13} &S_7=\rw(\alpha^7)=\alpha^{3} &S_8=\rw(\alpha^8)=\alpha
  \end{array}$$ 
 Step \reffpgz{fpgz2} in (a) sets $i_0=1$, $\vct{w}^{(1)}=0$, $\vct{y}^{(1)}=\alpha^3$, while in (b) calculates:\\
  \begin{tabbing}
   ($i=1$)$\xrightarrow{\hspace*{1cm}}\,$ \= $\varepsilon_1=S_2w_0^{(1)}+S_3=0$\\
                           \>$\varepsilon_2=S_3w_0^{(1)}+S_4=0+\alpha^{5}=\alpha^{5}$\\ \\
                           \>$\varepsilon_1=0 \text{ and }\varepsilon_2\neq 0\Rightarrow$ \= the singularity gap is $r=2$, hence\\
                           \>\> $i=i+2=3$;\\
                           \>\>$\vct{y}^{(3)}=\frac{1}{\varepsilon_2}\begin{mymatrix}
                                    w_0^{(1)}\\1\\0
                                   \end{mymatrix}=\begin{mymatrix}
                                    0\\ \alpha^{10}\\0
                                   \end{mymatrix}$;\\                 
                           \>\>$\vct{w}^{(3)}=\begin{mymatrix}
                                    \alpha^8\\ \alpha^9\\ \alpha^6
                                   \end{mymatrix}$;\\
  \end{tabbing}
  \begin{tabbing}
   ($i=3$)$\xrightarrow{\hspace*{1cm}}\,$ \= $\varepsilon_1=S_4w_0^{(3)}+S_5w_1^{(3)}+S_6w_2^{(3)}+S_7=
                                                            \alpha^{13}+\alpha^5+\alpha^4+\alpha^3=0$\\ 
                           \>$\varepsilon_1=0 \Rightarrow$ \= the singularity gap is $r>t-i=1$, hence\\
                           \>\> $e=i=3$;\\
                           \>\> $\sigma(x)=\alpha^8x^3+\alpha^9x^2+\alpha^6x+1$;            
  \end{tabbing}
  Since $\sigma(x)=\alpha^8x^3+\alpha^9x^2+\alpha^6x+1=
  \left(1+\alpha^{2}x\right)\left(1+\alpha^8x\right)\left(1+\alpha^{13}x\right)$,  step \reffpgz{fpgz3} calculates that  $$X_1=\alpha^{2}, \quad X_2=\alpha^{8}, \quad X_3=\alpha^{13}$$
  Therefore the error positions are 
  $$p_1=2, \quad p_2=8, \quad p_3=13$$
  Finally in \reffpgz{fpgz4} the linear system 
  $$\begin{mymatrix}
     \alpha^{2} &\alpha^{8}  & \alpha^{13}\\
     \left(\alpha^{2}\right)^2 &\left(\alpha^{8}\right)^2  & \left(\alpha^{13}\right)^2\\
     \left(\alpha^{2}\right)^3 &\left(\alpha^{8}\right)^3  & \left(\alpha^{13}\right)^3\\
    \end{mymatrix}
    \begin{mymatrix}
     E_1 \\ E_2\\ E_3 
    \end{mymatrix}=
    \begin{mymatrix}
     \alpha^{12}\\0\\0
    \end{mymatrix}$$
  is solved to find that the error values are
  $$E_1=\alpha^2,\quad E_2=\alpha, \quad E_3=\alpha^7 $$
  Thus $\e(x)=E_1x^{p_1}+E_2x^{p_2}+E_3x^{p_3}=\alpha^2x^2+\alpha x^8+\alpha^7x^{13}$
  and the received word $\rw$ is correctly decoded in $\cw=\rw-\e=\vct{0}$.\\
\end{ex}
                  
\section{Avoiding Malfunctions in \fPGZ}
In section \ref{sect:error-correting codes}, we studied  the definition and the behavior of a $t$-bounded distance decoding algorithm. In this section we want to understand whether the \fPGZ  decoding algorithm \ref{alg:fastPGZ} for Reed-Solomon codes is $t$-bounded  distance or not. In particular  we will describe the necessary and sufficient conditions  to avoid  malfunctions of this decoding algorithm.\\

Consider the code $RS(n,d,\alpha)$ and let $t$ be its error correction capability. Let $\cw$, $\e$ and $\rw=\cw+\e$ respectively denote the transmitted codeword, the vector error and the received word. Without any conditions about $e=\wt(\e)$, we suppose that  the received word $\rw$ is the input given to the \fPGZ  decoding algorithm \ref{alg:fastPGZ}.  In section \ref{sect:error-correting codes}, we defined  the set $\mathcal{B}$ as:
$$\mathcal B =\bigsqcup_{\cw\in RS(n,d,\alpha)} \overline{B}_t(\cw)$$
By what we have already proved, if $\wt(\e)\leq t$, then $\rw \in \mathcal B$ and the \fPGZ  decoding algorithm \ref{alg:fastPGZ} decodes correctly $\rw$, giving as output $\cw$. Whereas  if  $\wt(\e)>t$ and $\rw \in \mathcal B$, then an unavoidable decoder error has happened. So to know whether the \fPGZ  decoding algorithm \ref{alg:fastPGZ} is $t$-bounded  distance or not, we have to study what happens when  $\rw \notin \mathcal B$, recalling  that when $\rw \notin \mathcal B$ then a $t$-bounded  distance decoding algorithm does not correct $\rw$ and declare the decoder failure. 

For this aim, we introduce the following notation: denote by $\tl{e} $ the assumed number of occurred errors computed  at the end of \reffpgz{pgz2}  and  recall that  $$\tl{e}=\max\{i\in\{1,2,\dots,t \,|\,\det(A_i)\neq0\}\}\leq t$$
Similarity we assume that:
\begin{itemize}
 \item[-] $\tl{\sigma}(x)=\tl{\sigma}_{\tl{e}}x^{\tl{e}}+\cdots+\tl{\sigma}_1 x+1$ is the polynomial calculated in \reffpgz{pgz2},
 \item[-] $\tl{X}_1,\tl{X}_2,..., \tl{X}_{\tl{e}}$ are the inverses of the roots of $\tilde\sigma(x)$ calculated in \reffpgz{pgz3},
 \item[-] $\tl{E}_i$ is the $i$th assumed error value,
 \item[-] $\vct{\tl{e}}$ is the assumed error vector. Note that  $\tl{e}=\wt(\vct{\tl{e}})$.
\end{itemize}
and we call $\vct{\tl{c}}=\rw-\vct{\tl{e}}$ the output vector. 
We observe that, since $\tl{e}\leq t$, it holds that
\begin{equation}\label{condition1}
\rw\in\mathcal B \Longleftrightarrow \vct{\tl{c}}\in RS(n,d,\alpha)
\end{equation}
and we show in the following example a malfunction of the \fPGZ decoder.


 \begin{ex}\label{ex_fastpgz_mal} 
  Let $\alpha$ denote a primitive element of $\mathbb{F}_{2^4}$ satisfying $\alpha^4+\alpha+1=0$. Consider  over $\mathbb{F}_{2^4}$ the code, $RS(15,9,\alpha)$ generated by 
  $$g(x)=(x-\alpha)(x-\alpha^2)\cdots(x-\alpha^8)$$
  The code has distance $9$, so $t=4$.  Suppose that the sent codeword is $\vct{0}$ and the error vector is
  $$\e(x)=\alpha^3x+\alpha^3x^2+\alpha^{14}x^{10}+\alpha^{5}x^{12}+\alpha^{8}x^{13}$$
  Clearly $e=5$ and $\rw(x)=\e(x)$. The syndrome values are:
  $$\begin{array}{llll}
    S_1=\rw(\alpha)=\alpha^{10} & S_2=\rw(\alpha^2)=\alpha^2 &S_3=\rw(\alpha^3)=\alpha^8 & S_4=\rw(\alpha^4)=\alpha^{7}\\
    S_5=\rw(\alpha^5)=0   & S_6=\rw(\alpha^6)=\alpha^{3} &S_7=\rw(\alpha^7)=\alpha^{9} &S_8=\rw(\alpha^8)=\alpha^8
  \end{array}$$ 
  So \reffpgz{pgz2} starts setting $i_0=1$, $\vct{w}^{(1)}=\alpha^7$ and $\vct{y}^{(1)}=\alpha^5$ and continues with the iterations:
  \begin{tabbing}
   ($i=1$)$\xrightarrow{\hspace*{1cm}}$   											\=$\varepsilon_1=S_2w_0^{(1)}+S_3=\alpha^{2}\alpha^7+\alpha^{8}=\alpha^{12}$;\\[.3cm]
           \>$\varepsilon_1\neq 0\Rightarrow$ \= $\displaystyle\vct{y}^{(2)}=\frac{1}{\varepsilon_1}\begin{mymatrix}
                             w_0^{(1)}\\ 1                                                                                                                                                       \end{mymatrix}
              =\frac{1}{\alpha^{12}}\begin{mymatrix}
                       \alpha^{7}\\1
                      \end{mymatrix}
                      =\begin{mymatrix}
                        \alpha^{10}\\ \alpha^3
                       \end{mymatrix}$;\\[0.3cm]
                      \>\>$\varepsilon_2=S_3w_0^{(1)}+S_4=\alpha^{8}\alpha^7+\alpha^{7}=\alpha^{9}$;\\[.3cm]
                      \>\>$\eta=S_2y_0^{(1)}+S_4=\alpha^{2}\alpha^5=\alpha^{7}$;\\[.3cm]
                      \>\>$\vct{w}^{(2)}$\=$=\begin{mymatrix}
                                          0\\ \alpha^{7}
                                         \end{mymatrix}-\alpha^{12}\begin{mymatrix}\alpha^{5}\\0
                                                                   \end{mymatrix}+\left[\varepsilon_1 \eta-\varepsilon_2\right]
                                         \begin{mymatrix}
                                          \alpha^{10}\\ \alpha^3
                                           \end{mymatrix}=$\\[.3cm]
                      \>\>\>$=\begin{mymatrix}\alpha^{11}\\ \alpha^{12}
                                                          \end{mymatrix}$;\\
   \end{tabbing}
   \begin{tabbing}   
     ($i=2$)$\xrightarrow{\hspace*{1cm}}$
          \=$\varepsilon_1=S_3w_0^{(2)}+S_4w_1^{(2)}+S_5=\alpha^{8}\alpha^{11}+\alpha^{7}\alpha^{12}+0=0$;\\[.3cm]
          \>$\varepsilon_2=S_4w_0^{(2)}+S_5w_1^{(2)}+S_6=\alpha^{7}\alpha^{11}+0+\alpha^{3}=0$;\\[.3cm]
          \>$\varepsilon_1=\varepsilon_2=0\Rightarrow$\= $ \tl{e}=2$;\\[.3cm]
                                                      \>\>$\tl{\sigma}(x)=\alpha^{11}x^2+\alpha^{12}x+1$;
  \end{tabbing}
  Since $\tl{\sigma}(x)=(1+x)(1+\alpha^{11}x)$, in \reffpgz{pgz3} we have  $\tl{X}_1=1=\alpha^0$ and $\tl{X}_2=\alpha^{11}$. So the assumed error positions are $i_1=0$ and $i_2=11$. Finally in \reffpgz{pgz4} calculates the assumed error values solving
  $$\begin{mymatrix}
     1 & \alpha^{11}\\
     1 & \alpha^{7} 
    \end{mymatrix}\begin{mymatrix} \tl{E}_1 \\ \tl{E}_2\\
                 \end{mymatrix}=\begin{mymatrix} \alpha^{10}\\ \alpha^{2}
                               \end{mymatrix} $$   
  Thus we have that $\tl{E}_1=\alpha^6$ and $\tl{E}_2=\alpha^{11}$ and the output vector is 
  \begin{align*}
   \vct{\tl{c}}(x)&=\rw(x)-(\alpha^{11}x^{11}+\alpha^{6})=\\
   &=\alpha^{6}+\alpha^3x+\alpha^3x^2+\alpha^{14}x^{10}
         +\alpha^{11}x^{11}+\alpha^{5}x^{12}+\alpha^{8}x^{13}\neq \cw(x)
  \end{align*}
  We note that 
  $$\vct{\tl{c}}(\alpha^7)=\alpha^7\neq0 \Rightarrow \vct{\tl{c}}\notin RS(15,9,\alpha)\Rightarrow 
  \rw\notin \mathcal B$$
  Thus a decoder malfunction has occurred.\\
 \end{ex}
 
To avoid malfunctions of this kind in the \fPGZ decoding algorithm we need to add in its implementation the conditions necessary and sufficient to assure that the received word $\rw$ belongs to the set $\mathcal B$. 
If one of these conditions will be not satisfied, then the \fPGZ decoder must declare a failure without carry out the decoding.  
By (\ref{condition1}), one way to prevent malfunction from occurring is simply to add a step at the end of the algorithm which checks if $\vct{\tl{c}}\in RS(n,d,\alpha)$. This involves checking $d-1$ equations.  We will see later in this section that the same goal may be achieved with a lower number of equations to check. First a lemma about Hankel matrices is necessary:
 \begin{lem}\label{lemma:hankel}
  Let $M$ be a $k\times k$ or $k\times (k+1)$ Hankel matrix and let $m$ be its rank. If the $m\times m$ leading principal minor in $M$ is non-singular, then all the entries of $M$ are completely determined by the first $2m$ entries.
 \end{lem}
We leave the proof of lemma \ref{lemma:hankel} to the reader since it uses only the symmetric structure of an Hankel matrix and some easy linear algebra.

We state now the proposition that fixes the conditions to add in the \fPGZ  decoding algorithm \ref{alg:fastPGZ} in order to make it $t$-bound distance:
\begin{prop}\label{prop:malfunctionPGZ}
 With previous notations, we have that $\vct{\tl{c}}\in RS(n,d,\alpha)$ if and only if
 \begin{enumerate}[\ \ \  (1)]
   \item \label{ca} there exists $i_0\leq t$ as defined in remark \ref{rem:i0};
   \item \label{cb}$\tl{\sigma}(x)$ has exactly  $\tl{e}$ distinct roots, all belonging to the field $\F$ and all different from zero;
   \item \label{cc} $\rk(A)=\tl{e}$;
   \item \label{cd} for any $i\in\{1,2,\dots,\tl{e}\}$, $\tl{E}_i\neq0$.
 \end{enumerate}
\end{prop}
Condition (\ref{ca}) has been already included in the \fPGZ  decoding algorithm \ref{alg:fastPGZ}, but it is not sufficient as seen in example \ref{ex_fastpgz_mal}.
\begin{proof} If  $\vct{\tl{c}}\in RS(n,d,\alpha)$, then $\rw=\vct{\tl{c}}+\vct{\tl{e}}$ with $\wt(\vct{\tl{e}})\leq t$. Thus we  can repeat the construction of section \ref{sect:PGZ}, with $\vct{\tl{c}}$ and  $\vct{\tl{e}}$ in place of $\cw$ and $\e$, concluding that conditions (\ref{ca}), (\ref{cb}), (\ref{cc}) and (\ref{cd}) are satisfied. On the other hand, supposing that conditions (\ref{ca}), (\ref{cb}), (\ref{cc}) and (\ref{cd}) are satisfied, we have to prove that $\vct{\tl{c}}\in RS(n,d,\alpha)$, which is equivalent to prove that $\vct{\tl{c}}(\alpha^i)=0$ for any $i=1,2,\dots,d-1$.
We observe that
$$\vct{\tl{c}}(\alpha^i)=\rw(\alpha^i)-\vct{\tl{e}}(\alpha^i)= S_i-\vct{\tl{e}}(\alpha^i)$$
and we  define $\tl{S}_i=\vct{\tl{e}}(\alpha^i)$. Condition (\ref{ca}) ensures that step \reffpgz{fpgz2} can start correctly, while condition (\ref{cb}) ensures that in step \reffpgz{fpgz3} there exist $p_1,p_2\dots,p_{\tl{e}} $ in $\{1,2,\dots,n\}$ such that  $\tl{X}_i=\alpha^{p_i}$ and $j_a\neq j_b$ if $a\neq b$. Thus $\vct{\tl{e}}(x)=\displaystyle\sum_{i=1}^{\tl{e}}\tl{E}_{i}x^{p_i}$.
Since the assumed error values $\tl{E}_i$'s are computed in \reffpgz{fpgz4} solving the linear system
 $$\begin{mymatrix}\tl{X}_1   & \tl{X}_2  & \cdots & \tl{X}_{\tl{e}} \\
                 \tl{X}_1^2   & \tl{X}_2^2  & \cdots & \tl{X}_{\tl{e}}^2 \\
                 \vdots &\vdots & &\vdots\\
                 \tl{X}_1^{\tl{e}}   & \tl{X}_2^{\tl{e}}  & \cdots & \tl{X}_{\tl{e}}^{\tl{e}} \\
  \end{mymatrix}
    \begin{mymatrix} \tl{E}_1\\ \tl{E}_2 \\ \vdots \\ \tl{E}_{\tl{e}}\\
    \end{mymatrix}= 
   \begin{mymatrix} S_{1}\\ S_{2} \\ \vdots \\ S_{\tl{e}}\\
   \end{mymatrix} $$
we know that 
$$S_j=\sum_{i=1}^{\tl{e}} \tl{E}_i \tl{X}_i^j=\sum_{i=1}^{\tl{e}} \tl{E}_i \tilde (\alpha^{p_i})^j=\vct{\tl{e}}(\alpha^j) \quad \quad \forall \,j=1,2,\dots, \tl{e} $$
So we have already proved that $S_i=\tl{S}_i$ for any $i=1,2,...,\tl{e}$. We will prove now that the same equality  also holds for $i=\tl{e}+1,\tl{e}+2,\dots,d-1$.\\
From $$A_{\tl{e}}\begin{mymatrix}
                                                           \tl{\sigma}_{\tl{e}}\\ \vdots \\ \tl{\sigma}_1\\
                                                                \end{mymatrix}=-\begin{mymatrix}
                                                                                 S_{\tl{e}+1}\\ \vdots \\ S_{2\tl{e}}\\    
                                                                                \end{mymatrix}$$ 
it follows that the first $2\tl{e}$ syndromes are completely determined by the first $\tl{e}$ syndromes, in fact:
\begin{equation}\label{2}
-S_{\tl{e}+l}=\sum_{i=1}^{\tl{e}}S_{\tl{e}+l -i}\,\tl{\sigma}_i \quad \quad \forall\;l=1,2,\dots, \tl{e}
\end{equation}
From $\tl{\sigma}(\tl{X}_i^{-1})=0$ for any $i=1,2,\dots,\tl{e}$ it holds that
$$ \tl{X}_i^{\tl{e}}+\tl{\sigma}_1\tl{X}_i^{\tl{e}-1}+\cdots+\tl{\sigma}_{\tl{e}}=0$$
Let $j\neq0$ be a natural number, we multiply the former equation by $\tl{E}_i\tl{X}_i^j$ (different from zero by conditions (\ref{cb}) and (\ref{cd})) and after we sum over $i=1,2,\dots,\tl{e}$. Thus we obtain
$$\left(\sum_{i=1}^{\tl{e}}\tl{E}_i\tl{X}_i^{\tl{e}+j}\right)+
 \left(\sum_{i=1}^{\tl{e}}\tl{E}_i\tl{X}_i^{\tl{e}+j-1}\right)\tl{\sigma}_1+\dots+
 \left(\sum_{i=1}^{\tl{e}}\tl{E}_i\tl{X}_i^{j}\right)\tl{\sigma}_{\tl{e}}=0 \quad \forall\,j\geq1$$
that is
\begin{equation}\label{3}
  -\tl{S}_{\tl{e}+j}=\sum_{i=1}^{\tl{e}}\tl{S}_{\tl{e}+j -i}\,\tl{\sigma}_i \quad \quad \forall\;j\geq 1
\end{equation}
Hence from (\ref{2}) and (\ref{3}) we  also find that $S_i=\tl{S}_i$ for $i=1,2,\dots, 2\tl{e}$. Finally we consider the matrix $\tl{A}$ defined by $$\tl{A}=\begin{mymatrix}                                                  
                        \tl{S}_1          & \tl{S}_2             & \cdots & \tl{S}_{t+1} \\
                        \tl{S}_2          & \tl{S}_3             & \cdots & \tl{S}_{t+2}\\
                        \vdots       & \vdots          &        & \vdots \\
                        \tl{S}_{d-1-t}    & \tl{S}_{d-t}        & \cdots & \tl{S}_{d-1}\\
        \end{mymatrix}$$             
Since $\tl{A}$ is the syndrome matrix for the error vector $\vct{\tl{e}}$, which has weight less or equal to $t$, we know that $\rk(\tl{A})=\wt\left(\vct{\tl{e}}\right)=\tl{e}$. So we can conclude the proof using condition (\ref{cc}) and lemma \ref{lemma:hankel}, which implies that  $A$ and  $\tl{A}$ are equal. 
\end{proof}

\noindent Note that in the example \ref{ex_fastpgz_mal} the syndrome matrix is 
$$A=\begin{mymatrix}
     \alpha^{10}& \alpha^{2}& \alpha^{8}& \alpha^{7}& 0\\
     \alpha^{2}& \alpha^{8}& \alpha^{7}& 0& \alpha^3\\
     \alpha^{8}& \alpha^{7}& 0& \alpha^3& \alpha^9\\
     \alpha^{7}& 0& \alpha^3& \alpha^9& \alpha^7
    \end{mymatrix}
$$
and it holds that $\rk(A)\geq3$ since
$$\det\begin{mymatrix}
     \alpha^{10}& \alpha^{2}&  0\\
     \alpha^{2}& \alpha^{8}&  \alpha^3\\
     \alpha^{8}& \alpha^{7}& \alpha^9\\
     \end{mymatrix}=\alpha^{14}\neq0
$$ but $\tl{e}=2$, thus the condition $\rk(A)=\tl{e}$ cannot hold.\\

While conditions (\ref{ca}), (\ref{cb}) and (\ref{cd}) in proposition \ref{prop:malfunctionPGZ} may be  easily implemented without any significant extra cost, condition (\ref{cc}) is not good from an implementation point of view. For this reason  we observe that the following equivalent condition holds:
\begin{prop}\label{aiutino2}
 Supposing that conditions (\ref{ca}), (\ref{cb}) and (\ref{cd}) of proposition \ref{prop:malfunctionPGZ} hold, we have that
 $$\rk(A)=\tl{e} \Longleftrightarrow  A\begin{mymatrix}
                                     \vct{0}^{i}\\ \tl{\sigma}_{\tl{e}}\\ \vdots \\ \tl{\sigma}_1 \\ 1 \\ \vct{0}^{t-\tl{e}-i}
                  \end{mymatrix}=\vct{0}\quad\quad \forall\,i=0,1,\dots,t-\tl{e}$$
\end{prop}     
\begin{proof}
 If $\rk(A)=\tl{e}$, then from what we saw in the proof of proposition \ref{prop:malfunctionPGZ} it follows that
 $$-S_{\tl{e}+j}=\sum_{i=1}^{\tl{e}}S_{\tl{e}+j -i}\,\tl{\sigma}_i \quad \quad \forall\;j=1,2,\dots, d-1-\tl{e} $$ and we obtain straightly the right side condition.
 On the other hand, since $\tl{e}$ is calculated as $\max\left\{i\in\{1,2,\dots, t\} \,|\, \det(A_i)\neq0\right\}$, it is always true that $\rk(A)\geq \tl{e}$ and it remains to prove only that $\rk(A)\leq \tl{e}$. If the right side condition is accomplished, then there are $t+1-\tl{e}$ linearly independent vectors in  $\ker(A)$. It follows that $\dim(\ker(A))\geq t+1-\tl{e}$. By definition of the syndrome matrix, we have that $t+1=\dim(\ker(A))+\rk(A)$, hence $\rk(A)\leq \tl{e}$ and the proof is complete. 
\end{proof}

Thanks to the Hankel structure of the syndrome matrix, checking the right side condition in proposition \ref{aiutino2} is equivalent to checking the following $d-1-2\tl{e}$ equations:
$$\begin{cases}
 \displaystyle S_{2\tl{e}+1}+\sum_{i=1}^{\tl{e}}S_{2\tl{e}+1 -i}\,\tl{\sigma}_i=0\\
 \displaystyle S_{2\tl{e}+2}+\sum_{i=1}^{\tl{e}}S_{2\tl{e}+2 -i}\,\tl{\sigma}_i=0\\
 \cdots\\
 \displaystyle  S_{d-1}+\sum_{i=1}^{\tl{e}}S_{d-1-i}\,\tl{\sigma}_i=0\\
\end{cases}$$
Note that the first $t-\tl{e}$ equations hold by the condition $$\varepsilon_1=\varepsilon_2=\cdots=\varepsilon_{t-\tl{e}}=0$$
that happens at the end \reffpgz{pgz2}. Checking the remaining equations has a computational complexity of $O(\tl{e}t)$ operations, so we may add this check to  step \reffpgz{pgz2} without increasing its total complexity.\\

Finally we observe that the \fPGZ  decoder can malfunction because it may not use all the syndrome values in the decoding, indeed step \reffpgz{pgz2} calculates $\tl{e}$ and $\tl{\sigma}(x)$ using only the first $\tl{e}+t$ syndromes. Thus it is understandable that the controls necessary to avoid malfunctions involve the remaining syndromes. 

\chapter{Berlekamp-Massey Decoder}
Elwyn R. Berlekamp published his algorithm to solve the key equation (\ref{keyequation}) in 1968 and in 1969  James Massey gave a simplified version of it. 
The described algorithm found the polynomials  $\sigma(x)$ and $\omega(x)$ with an iterative procedure based on polynomial families  defined by recursion. 
In section \ref{sect:BMalg} we will study in detail this procedure, while in section \ref{sect:BMmalf} we will  describe the necessary and sufficient conditions to avoid decoder malfunctions in the implementation given of the Berlekamp-Massey decoder. Finally in section \ref{sect:comparison} we will compare the latter with the \fPGZ decoding algorithm.
%

\section{\BM Decoding Algorithm}\label{sect:BMalg}
The Berlekamp-Massey (\BM) decoding algorithm is iterative and solves the key equation (\ref{keyequation}) in successively higher degrees. In other words, for $i=0,1,\dots,d-1$  we attempt to find polynomials of ``small'' degree
$$\sigma^{(i)}(x)=\sum_{j=0}^i \sigma^{(i)}_j x^j\in\F[x]$$
$$\omega^{(i)}(x)=\sum_{j=0}^{i-1} \omega^{(i)}_j x^j\in\F[x]$$
which satisfy
\begin{equation} \label{keyequation_i}
\sigma^{(i)}(x)S(x)\equiv\omega^{(i)}(x) \Mod{  x^{i}}
\end{equation}
We look for  solutions of small degree because we have already seen in the theorem \ref{thm:keyequation2} that the pair $(\sigma(x), \omega(x))$ is the valid solution of minimal degree for the key equation (\ref{keyequation}).
The iterative nature of the algorithm follows from the observation that if we know a solution to (\ref{keyequation_i}) that satisfies the conditions $\Deg{\sigma^{(i)}(x)}\leq i$ and $\Deg{\omega^{(i)}(x)}\leq i-1$, 
then we may write
$$\sigma^{(i)}(x)S(x)\equiv\omega^{(i)}(x)+\Delta_ix^i \Mod{  x^{i+1}}$$
where $\Delta_i$ is  the coefficient of $x^i$ in the polynomial $\sigma^{(i)}(x)S(x)$, that is
\begin{equation}\label{defi:Delta}
\Delta_i\displaystyle\overset{def}{=}\sum_{j=0}^{i}S_{i+1-j}\,\sigma_j^{(i)}
\end{equation}
and it is called \emph{$i$th discrepancy}. If $\Delta_i=0$, then we may evidently proceed taking  $\sigma^{(i+1)}(x)=\sigma^{(i)}(x)$ and $\omega^{(i+1)}(x)=\omega^{(i)}(x)$. In order to define $\sigma^{(i+1)}(x)$ and $\omega^{(i+1)}(x)$ also in the case when $\Delta_i\neq 0$, we introduce the \emph{auxiliary polynomials} $\tau^{(i)}(x)$ and $\gamma^{(i)}(x)$, which will be chosen so that they are a solution of the \emph{auxiliary equation}
\begin{equation}\label{auxiliaryequation}
\tau^{(i)}(x)S(x)\equiv \gamma^{(i)}(x) + x^{i-1} \Mod {x^i}
\end{equation}
Supposing the existence of the auxiliary polynomials $\tau^{(i)}(x)$ and $\gamma^{(i)}(x)$ for any $i\in\{0,1,\dots,d-2\}$, then
\begin{defi}\label{defi:sigma_omega_i} We define
$$\begin{cases}\sigma^{(0)}(x)=1\\
\sigma^{(i+1)}(x)=\sigma^{(i)}(x)-\Delta_ix\tau^{(i)}(x)
\end{cases}$$
$$\begin{cases}\omega^{(0)}(x)=0\\
\omega^{(i+1)}(x)=\omega^{(i)}(x)-\Delta_ix\gamma^{(i)}(x)
\end{cases}$$
for any $i\in\{0,1,\dots,d-2\}$.
\end{defi}  
We will show later on in proposition \ref{prop:BM1} that the polynomials $\sigma^{(i)}(x)$'s and $\omega^{(i)}(x)$'s just defined satisfy (\ref{keyequation_i}) for every $i\in\{0,1,\dots,d-1\}$ and moreover we will prove that $\sigma^{(d-1)}(x)=\sigma(x)$ and $\omega^{(d-1)}(x)=\omega(x)$.

The auxiliary polynomials are recursively defined during the algorithm. Indeed if $\tau^{(i)}(x)$ and $\gamma^{(i)}(x)$ satisfy the congruence (\ref{auxiliaryequation}), then we have two obvious ways to define $\tau^{(i+1)}(x)$ and $\gamma^{(i+1)}(x)$:
$$\begin{cases}\tau^{(i+1)}(x)=x\tau^{(i)}(x)\\
                \gamma^{(i+1)}(x)=x\gamma^{(i)}(x)                                                                                                                                                                                                                                                               
  \end{cases} \quad \text{ or } \quad
  \begin{cases} \tau^{(i+1)}(x)=\frac{\sigma^{(i)}(x)}{\Delta_i}\\
                \gamma^{(i+1)}(x)=\frac{\omega^{(i)}(x)}{\Delta_i}
  \end{cases}$$
If $\Delta_i=0$, then our choice is forced, while, if $\Delta_i \neq0$, then our choice must be based upon the aim to minimize the degrees of $\sigma^{(i+1)}(x)$ and $\omega^{(i+1)}(x)$. For this reason we introduce the function $D:\N\rightarrow\N$, defined by
\begin{defi}\label{defi:D(i)}
Let $D(0)=0$ and
$$D(i+1)= \begin{cases} D(i)  \quad &\text{ if } \;\Delta_i=0 \text{ or } 2D(i)\geq i+1\\
                      i+1-D(i) \quad &\text{ if } \;\Delta_i\neq0 \text{ and } 2D(i)< i+1
        \end{cases}$$
for any $i\geq0$.
\end{defi}
The function $D$ is evidently  nondecreasing and nonnegative and $D(i)\leq i$ for any $i\in\N$. Moreover, as we will prove later on in this section, it represents an upper bound for  $\Deg{\sigma^{(i)}(x)}$ and permits the right definition of the auxiliary polynomials as follows:
\begin{defi}\label{defi:tau_gamma}
Let $\tau^{(0)}(x)=1$ and  $\gamma^{(0)}(x)=-x^{-1}$. Then
$$\displaystyle
\tau^{(i+1)}(x)=\begin{cases} x\tau^{(i)}(x)\quad&\text{ if }\;\Delta_i=0 \text{ or } 2D(i)\geq i+1\\
                 \frac{\sigma^{(i)}(x)}{\Delta_i}\quad&\text{ if }\;\Delta_i\neq0\text{ and }2D(i)<i+1
        \end{cases}$$
$$\displaystyle
\gamma^{(i+1)}(x)=\begin{cases}x\gamma^{(i)}(x)\quad&\text{ if }\;\Delta_i=0\text{ or }2D(i)\geq    i+1\\
                  \frac{\omega^{(i)}(x)}{\Delta_i}\quad &\text{ if } \;\Delta_i\neq0 \text{ and } 2D(i)< i+1
        \end{cases}$$
for any $i\in\{0,1,\dots,d-2\}$.
\end{defi}

With these definitions it is immediate to show that:
\begin{prop}\label{prop:BM0}
If $\sigma^{(i)}(x)$, $\omega^{(i)}(x)$, $\tau^{(i)}(x)$ and $\gamma^{(i)}(x)$ are the polynomials defined respectively in definitions \ref{defi:sigma_omega_i} and \ref{defi:tau_gamma}, then
$$\sigma^{(i)}(x)S(x)\equiv \omega^{(i)}(x) \Mod{ x^i}$$
$$\tau^{(i)}(x)S(x)\equiv \gamma^{(i)}(x) + x^{i-1} \Mod {x^i}$$
for any $i\in\{0,1\dots,d-1\}$. 
Moreover 
\begin{equation}\label{eq:bm0}
\omega^{(i)}(x)\tau^{(i)}(x)-\sigma^{(i)}(x)\gamma^{(i)}(x)=x^{i-1}
\end{equation}
and thus we have that $\gcd\left(\sigma^{(i)}(x),\omega^{(i)}(x)\right)=1$.
\end{prop}
\begin{proof} 
 We will prove the proposition by induction. If $i=0$ the congruences are both trivial. Suppose the congruences true for $i$, it follows that 
  $$\begin{cases}
  \sigma^{(i)}(x)S(x)\equiv \omega^{(i)}(x)+\Delta_ix^i \Mod {x^{i+1}}\\
  x^i\equiv x\tau^{(i)}(x)S(x)-x\gamma^{(i)}(x)\Mod{ x^{i+1}} \end{cases}$$
  where $\Delta_i$ is defined by (\ref{defi:Delta}).  Substituting the expression for $x^i$ in the first congruence  we obtain that 
  $$
  \sigma^{(i)}(x)S(x)\equiv \omega^{(i)}(x)+\Delta_i\left( x\tau^{(i)}(x)S(x)-x\gamma^{(i)}(x)\right) \Mod {x^{i+1}}$$
  and using that
  $$\sigma^{(i+1)}(x)=\sigma^{(i)}(x)-\Delta_ix\tau^{(i)}(x)$$
  $$\omega^{(i+1)}(x)=\omega^{(i)}(x)-\Delta_ix\gamma^{(i)}(x)$$
  we can conclude that
  $$\sigma^{(i+1)}(x)S(x)\equiv \omega^{(i+1)}(x) \Mod{ x^{i+1}}$$
  Now we consider the relations that can  define $\tau^{(i+1)}(x)$ (definition \ref{defi:tau_gamma}). \\ 
  In the first case 
  $$\tau^{(i+1)}(x)S(x)=x\tau^{(i)}(x)S(x)\equiv  x\gamma^{(i)}(x)+  x^i\Mod{ x^{i+1}}$$
  In the other case
  \begin{align*}\tau^{(i+1)}(x)S(x)=\frac{\sigma^{(i)}(x)}{\Delta_i}S(x)&\equiv  \frac{\omega^{(i)}(x)+ \Delta_ix^i}{\Delta_i}
  \Mod{ x^{i+1}}\\
  &\equiv \gamma^{(i+1)}(x) + x^{i} \Mod {x^{i+1}}
  \end{align*} 
  and this concludes the proof of the inductive step. Thus the congruences are proved.
  In order to prove (\ref{eq:bm0}) we note that  base case is again trivial, so we focus our attention on the inductive step:
  $$\omega^{(i+1)}\tau^{(i+1)}-\sigma^{(i+1)}\gamma^{(i+1)}=\left(\omega^{(i)}-\Delta_ix\gamma^{(i)}\right)\tau^{(i+1)}- \left(\sigma^{(i)}-\Delta_ix\tau^{(i)}\right)\gamma^{(i+1)}$$
  Depending on the values of $\Delta_i$ and $D(i)$ we will substitute in the former expression the right values for $\tau^{(i+1)}(x)$ and $\gamma^{(i+1)}(x)$. It is an easy computation to verify that, in both cases, using the inductive hypothesis we obtain that $$\omega^{(i+1)}(x)\tau^{(i+1)}(x)-\sigma^{(i+1)}(x)\gamma^{(i+1)}(x)=x^{i}$$
  From the relation just proved, it follows that $\gcd\left(\sigma^{(i)}(x),\omega^{(i)}(x)\right)| x^{(i-1)}$, but it is immediate to prove that $\sigma^{(i)}(0)=1$ for every $i$, that is $x\nmid\sigma^{(i)}(x)$. Thus we conclude that $\gcd\left(\sigma^{(i)}(x),\omega^{(i)}(x)\right)=1$.
\end{proof}

In order to prove that the polynomials $\sigma^{(d-1)}(x)$ and $\omega^{(d-1)}(x)$ are respectively the error-locator polynomial and the error-evaluator polynomial, we need the followings propositions and theorems:
\begin{prop}\label{prop:BM1}
If $\sigma^{(i)}(x)$, $\omega^{(i)}(x)$, $\tau^{(i)}(x)$ and $\gamma^{(i)}(x)$ are the polynomials defined respectively in definitions \ref{defi:sigma_omega_i} and \ref{defi:tau_gamma}, then $ \forall\;i\in\{0,1,\dots,d-1\}$ we have that 
\begin{enumerate}[\ \ \ (1)]
\item $\begin{cases} 
       \deg\left(\sigma^{(i)}(x)\right)\leq D(i) \\
       \deg\left(\tau^{(i)}(x)\right)\leq i-D(i)
       \end{cases}$
\item $\begin{cases}
        \deg\left(\omega^{(i)}(x)\right)\leq D(i)-1\\
        \deg\left(\gamma^{(i)}(x)\right)\leq i-D(i)-1
       \end{cases}$       
\end{enumerate}
where  $D$ is the integer function defined in  definition \ref{defi:D(i)}.
\end{prop}
\begin{proof} We will prove both claims by induction. The base case of (1)  is  trivial. Suppose that the inequalities are true for $i$ and  consider the polynomial                                                                                                   $\tau^{(i+1)}(x)$. Depending on the values of $\Delta_i$ and $D(i)$ we have:
  $$\deg\left(\tau^{(i+1)}\right)=
  \begin{cases}
    1+\deg\left(\tau^{(i)}\right)\leq 1+i-D(i)=1+i-D(i+1)\\
    \deg\left(\sigma^{(i)}\right)\leq D(i)=i+1-D(i+1)
  \end{cases}$$  
So, in both cases, $\deg\left(\tau^{(i+1)}(x)\right)\leq i+1-D(i+1)$. Moreover, using that 
$$\sigma^{(i+1)}(x)=\sigma^{(i)}(x)-\Delta_ix\tau^{(i)}(x)$$
we have that
\begin{align*}
  &\text{if }\Delta_i =0 \Longrightarrow \deg\left(\sigma^{(i+1)}\right)=\deg\left(\sigma^{(i)}\right)\leq D(i)=D(i+1)\\
   &\text{if }\Delta_i \neq 0 \Longrightarrow \deg\left(\sigma^{(i+1)}\right)\leq\max\left\{D(i), i+1-D(i)\right\}=D(i+1)
\end{align*}
and (1) is proved.
The proof of (2) is almost identical to that of (1), using the relation that defines $\omega^{(i+1)}(x)$ (definition \ref{defi:sigma_omega_i}).\\
\end{proof}

Consider now the set of all possible solutions of the key equation (\ref{keyequation_i}), that is
$$\mathcal M_i=\left\{(a(x),b(x))\in \F[x]^2 \;|\; a(x)S(x)\equiv b(x) \Mod{x^i}\right\}$$
Clearly $\mathcal M_i$ is a free submodule of rank 2 of $\F[x]^2$, since $\left\{(1, S(x)), (0,x^i)\right\}$ is obviously a basis of $\mathcal M_i$. In the following theorem we will show that even the set
$\left\{\left(\sigma^{(i)}(x),\omega^{(i)}(x)\right), \left(x\tau^{(i)}(x), x\gamma^{(i)}(x)\right)\right\}$ is a base for $\mathcal M_i$.

\begin{thm}\label{thm:BM2}
 Let $i$ be a fixed index in $\{1,\dots,d-1\}$. If $\left(a(x),b(x)\right)\in \mathcal M_i$ then there exist polynomials $u(x)$, $v(x) \in \F[x]$ such that $u(0)=a(0)$ and
 $$\begin{cases} a(x)=u(x)\sigma^{(i)}(x)+v(x)x\tau^{(i)}(x)\\
                 b(x)=u(x)\omega^{(i)}(x)+v(x)x\gamma^{(i)}(x)
   \end{cases}$$
 \noindent Moreover, if there exists $\delta\in\N^+$ such that $\deg(a(x))\leq\delta$ and $\deg(b(x))\leq\delta-1$ then we have $\deg(u(x))\leq \delta-D(i) \text{ and } \deg(v(x))\leq \delta+D(i)-i-1$.
\end{thm}
\begin{proof}
By hypothesis we have that 
$$a(x)\omega^{(i)}(x)\equiv a(x)\sigma^{(i)}(x)S(x)\equiv b(x)\sigma^{(i)}(x) \Mod{ x^{i}}$$
$$a(x)\gamma^{(i)}(x)\equiv a(x)\left[\tau^{(i)}(x)S(x)-x^{i-1}\right]\equiv b(x)\tau^{(i)}(x)-a(0)x^{i-1} \Mod{ x^{i}}$$
so there exist polynomials $\tilde{u}(x)$ and $v(x)$ such that 
$$a(x)\omega^{(i)}(x)- b(x)\sigma^{(i)}(x)= x^{i}v(x)$$
$$a(x)\gamma^{(i)}(x)- b(x)\tau^{(i)}(x)= x^{i-1}\left[x\tilde u(x)-a(0)\right]$$
We define $u(x)=a(0)-x\tilde u(x)$, hence we have that $u(0)=a(0)$ and
\begin{equation}\label{eq:v}
 a(x)\omega^{(i)}(x)- b(x)\sigma^{(i)}(x)= x^{i}v(x)
\end{equation}
\begin{equation}\label{eq:u}
-a(x)\gamma^{(i)}(x)+b(x)\tau^{(i)}(x)= x^{i-1}u(x)
\end{equation}
from which, using (\ref{eq:bm0}), it follows that:
\begin{align*}
x^{i-1}&a(x)=\left[\omega^{(i)}(x)\tau^{(i)}(x)-\sigma^{(i)}(x)\gamma^{(i)}(x)\right]a(x)=\\
&=\tau^{(i)}(x)\left[a(x)\omega^{(i)}(x)-b(x)\sigma^{(i)}(x)\right]-\sigma^{(i)}(x)\left[a(x)\gamma^{(i)}(x)-b(x)\tau^{(i)}(x)\right]=\\
&=x^{i-1}\left[u(x)\sigma^{i)}(x)+v(x)x\tau^{(i)}(x)\right]
\end{align*}
and
\begin{align*}
x^{i-1}&b(x)=\left[\omega^{(i)}(x)\tau^{(i)}(x)-\sigma^{(i)}(x)\gamma^{(i)}(x)\right]b(x)=\\
&=\gamma^{(i)}(x)\left[a(x)\omega^{(i)}(x)-b(x)\sigma^{(i)}(x)\right]-\omega^{(i)}(x)\left[a(x)\gamma^{(i)}(x)-b(x)\tau^{(i)}(x)\right]=\\
&=x^{i-1}\left[u(x)\omega^{i)}(x)+v(x)x\gamma^{(i)}(x)\right]
\end{align*}
So the first part of the theorem \ref{thm:BM2} is proved. From (\ref{eq:v}) and (\ref{eq:u}) it also follows  that
$$i-1+\deg(u(x))\leq \max\left\{ \deg\left(\gamma^{(i)}(x)\right)+\deg\left(a(x)\right), \deg\left(\tau^{(i)}(x)\right)+\deg(b(x))\right\}$$
$$i+\deg(v(x))\leq \max\left\{ \deg\left(\omega^{(i)}(x)\right)+\deg(a(x)), \deg\left(\sigma^{(i)}(x)\right)+\deg(b(x))\right\}$$
If $\deg(b(x))<\deg(a(x))$, using the degree inequalities stated in proposition \ref{prop:BM1}, we obtain that
$$\max\left\{\deg\left(\gamma^{(i)}(x)\right)+\deg(a(x)),\deg\left(\tau^{(i)}(x)\right)+\deg(b(x))\right\}\leq \delta+i-D(i)-1$$
$$\max\left\{ \deg\left(\omega^{(i)}(x)\right)+\deg(a(x)), \deg\left(\sigma^{(i)}(x)\right)+\deg(b(x))\right\}\leq D(i)+\delta-1$$
so we have $\deg(u(x))\leq \delta-D(i)$ and $\deg(v(x))\leq \delta+D(i)-i-1$ and this concludes the proof.\\
\end{proof}

\begin{cor}\label{cor:BM2}
Let $i$ be a fixed index in $\{1,\dots,d-1\}$ and let $\delta\in\N^+$ be less or equal to $\pint{\frac{i}{2}}$. If  $(a(x),b(x))\in \mathcal M_i$, $a(0)\neq 0$, $\deg(a(x))\leq \delta$ and $\deg(b(x))\leq\delta-1$,  then we have 
$D(i)\leq \delta$ and
$$\begin{cases} a(x)=u(x)\sigma^{(i)}(x)\\
                      b(x)=u(x)\omega^{(i)}(x)
        \end{cases}$$
where $u(x)=\gcd(a(x), b(x))\in \F[x]$.

\end{cor}
\begin{proof}
 Since $(a(x),b(x))\in \mathcal M_i$, from theorem \ref{thm:BM2} it follows that there are $u(x)$ and $v(x)$ in $\F[x]$ such that $u(0)=a(0)$ and
 $$\begin{cases} a(x)=u(x)\sigma^{(i)}(x)+v(x)\tau^{(i)}(x)\\
                 b(x)=u(x)\omega^{(i)}(x)+v(x)\gamma^{(i)}(x)
   \end{cases}$$
 Since $u(0)=a(0)\neq 0$, we have that $\deg(u(x))\geq 0$ hence, from  (\ref{eq:u}) it follows that 
 $$\deg\left(-a(x)\gamma^{(i)}(x)+b(x)\tau^{(i)}(x)\right)\geq i-1$$
 that is 
 $$\Deg{a(x)\gamma^{(i)}(x)}\geq i-1 \quad \text{ or } \quad \Deg{b(x)\tau^{(i)}(x)}\geq i-1$$
 Using proposition \ref{prop:BM1}, we see that
 $$\Deg{a(x)}\geq D(i) \quad \text{ or } \quad \Deg{b(x)}\geq D(i)-1$$
 In both cases, using the hypothesis on the degree of $a(x)$ or $b(x)$, we conclude that $D(i)\leq \delta$. Since $\delta\leq\pint{\frac{i}{2}}$ and $D(i)\leq \delta$, using theorem \ref{thm:BM2}, we have that  
 $$\Deg{v(x)}\leq \delta+D(i)-i-1\leq 2\delta-i-1\leq2\pint{\dfrac{i}{2}}-i-1<0$$ 
 that is $v(x)=0$.  Moreover, since  $\gcd\left(\sigma^{(i)}(x),\omega^{(i)}(x)\right)=1$, it is clear that $u(x)=\gcd(a(x), b(x))$ and the proof is  concluded.\\
\end{proof}

\begin{cor}\label{cor:BM2bis}
 We recall the definitions given on page \pageref{defi:sigma_omega_i} of the polynomials $\sigma^{(i)}(x)$'s and $\omega^{(i)}(x)$'s for $i\in\{0,1,\dots,d-2\}$:\\
 
 $\begin{cases}\sigma^{(0)}=1\\
 \sigma^{(i+1)}(x)=\sigma^{(i)}(x)-\Delta_ix\tau^{(i)}(x)
 \end{cases}$
 $\begin{cases}\omega^{(0)}=0\\
 \omega^{(i+1)}(x)=\omega^{(i)}(x)-\Delta_ix\gamma^{(i)}(x)
 \end{cases}$\\

 \noindent If $e\leq t$, then we have  
 \begin{equation*}\begin{cases} 
     \sigma^{(d-1)}(x)=\sigma(x)\\
     \omega^{(d-1)}(x)=\omega(x)
   \end{cases}\end{equation*}
 Moreover $D(d-1)=e$.
\end{cor}
\begin{proof}
 By the key equation (\ref{keyequation}), we already know that $(\sigma(x), \omega(x))\in \mathcal{M}_{d-1}$. Hence, using remark \ref{rem:sigmaproperties} and corollary \ref{cor:BM2}, we immediately conclude that $\sigma^{(d-1)}(x)=\sigma(x)$,  $\omega^{(d-1)}(x)=\omega(x)$ and $D(d-1)\leq e$. 
 Since $$e=\Deg{\sigma(x)}=\Deg{\sigma^{(d-1)}(x)}\leq D(d-1)$$
 the proof is completed.
\end{proof}

Before entering into the details of the \BM decoding algorithm and stating it explicitly, we observe that
if we know the polynomial $\sigma(x)$, then the polynomial $\omega(x)$ is completely determined, since the key equation (\ref{keyequation})  also implies that the coefficients of the error-evaluator polynomial are equal to the coefficients of the product $\sigma(x)S(x)$ for all terms $1,x,\dots, x^{e-1}$. In other words the coefficients of $\omega(x)$ are determined by:
\begin{equation}\label{omega}
 \begin{cases}
  \omega_0=S_1\\
  \omega_i= S_{i+1}+\displaystyle\sum_{j=1}^i S_{i+1-j}\,\sigma_j \quad\quad \text{ for } i=1,2,\dots,e-1
 \end{cases}
\end{equation}
Thus the computation of the polynomials $\omega^{(i)}(x)$'s and $\gamma^{(i)}(x)$'s is not necessary and, in order to save arithmetic operations and memory space, we prefer to calculate the error-evaluator polynomial $\omega(x)$ by (\ref{omega}) after we have computed the error-locator polynomial $\sigma(x)$. \\

\begin{alg}[\textbf{\BM decoding algorithm} for $\RS$]\hspace{2cm}
\label{alg:BM}
\begin{description}
\item[\texttt{Input}:] the received word $\rw(x)$;
 \item[\texttt{Output}:] the codeword $\cw(x)$;
\end{description}
\texttt{Begin}
\begin{enumerate}[\bfseries \BM.1]
        \item \label{BM1}
              \textit{(syndrome computation)}\\[-1cm]
              \begin{tabbing}
               \texttt{for }\= $i=1,2,\dots, d-1$ \texttt{ do}\\
               \> $S_i:=\rw(\alpha^i)$;\\
               \texttt{endfor}               
              \end{tabbing}
        \item \label{BM2}
              \textit{(error-locator polynomial computation)}\\
              $\sigma^{(0)}:=1$;\\ 
              $\tau^{(0)}:=1$;\\
              $D(0):=0$;                   
               \begin{tabbing} 
                 \texttt{for} \= $i=0,1,\dots,d-2$ \texttt{ do}\\
                 \>$\Delta_i:\displaystyle=\sum_{j=0}^{D(i)}S_{i+1-j}\,\sigma_j^{(i)}$;\\
                 \>$\sigma^{(i+1)}(x):=\sigma^{(i)}(x)-\Delta_ix\tau^{(i)}(x)$;\\[0,2cm]
                 \>\texttt{if} \= ($\Delta_i=0$ or $2D(i)\geq i+1$)   \texttt{then}\\ 
                 \>\>$D(i+1):=D(i)$;\\
                 \>\>$\tau^{(i+1)}(x):=x\tau^{(i)}(x)$;\\
                 \>\texttt{else}  \= \\
                 \>\>$D(i+1):=i+1-D(i)$;\\
                 \>\>$\tau^{(i+1)}(x):=\frac{\sigma^{(i)}(x)}{\Delta_i}$;\\
                 \texttt{endfor}\\
                 $e:=D(d-1)$;\\
                 $\sigma(x):=\sigma^{(d-1)}(x)$;
               \end{tabbing}
       \item \label{BM3}
	     \textit{(finding error positions)}\\
             \texttt{calculate  the error positions $p_1,p_2,\dots,p_e$ and the elements $X_1^{-1}, X_2^{-1},\dots,X_e^{-1}$ using Chien's search (alg.\ \ref{alg:chien});}
       \item \label{BM4}
             \textit{(finding error values)}\\[0,4cm]
             $\displaystyle\omega(x):=S_1+\sum_{i=1}^{e-1}\left(S_{i+1}+\sum_{j=1}^i S_{i+1-j}\,\sigma_j\right)x^i\,;$\\
             \begin{tabbing}
               \texttt{for} \= $i=1,2,\dots,e$ \texttt{ do }\\[0,2cm]
               \>$E_i:=-\dfrac{\omega(X_i^{-1})}{\sigma'(X_i^{-1})}$;\\
               \texttt{endfor}
             \end{tabbing}            
\end{enumerate}
\texttt{Return } $\cw(x):=\displaystyle \rw(x)-\sum_{i=1}^e E_{i}x^{p_i};$\\
\texttt{End}\\
\end{alg}

When $e\leq t$, corollary \ref{cor:BM2bis} ensures the correctness of the \BM decoding algorithm. Its computational complexity (ignoring \refbm{BM1} and \refbm{BM3}) is of order $O(t^2)$. In fact, the cost of the $i$th iterative step in \refbm{BM2} is due to:
\begin{itemize}
 \item the computation of the $i$th discrepancy ($i$ multiplications and $i$ additions);
 \item the computation of the coefficients of $\sigma^{(i+1)}(x)$ as 
       $$\begin{cases}
          \sigma^{(i+1)}_0:=1\\     
          \sigma^{(i+1)}_j:=\sigma^{(i)}_j-\Delta_i\tau_{j-1}^{(i)}
         \end{cases}$$
       for any $j=1,2,\dots,i+1$ ($i+1$ multiplications and $i+1$ additions);
 \item the computation of  the coefficients of $\tau^{(i+1)}(x)$ as
       $$\begin{cases}
           \tau^{(i+1)}_0:=0\\
           \tau^{(i+1)}_j:=\tau^{(i)}_{j-1} \\
         \end{cases} \quad\text{ or }\quad 
         \begin{cases}
           \delta:=\Delta_i^{-1}\\
           \tau^{(i+1)}_l:=\delta\cdot\sigma^{(i)}_{l} \\
         \end{cases}     $$
       for  $j=1,2,\dots,i+1$ and  $l=0,1,\dots,i$   (at most one inversion and $i+1$ multiplications);
\end{itemize}
Hence the $i$th iterative step in \refbm{BM2} has a complexity upper bounded by $3i+2$ multiplications, $2i+1$ additions and one inversion. Thus summing for $i=0$ to $i=d-2$ and recalling that $t=\pint{\frac{d-1}{2}}$ (that is $d-1=2t+1$ or $d-1=2t$), we obtain that the total complexity of \refbm{BM2} is upper bounded by:
\begin{center}
 \begin{tabular}{ll}
  $2t+1$ & inversions\\
  $6t^2+7t+4$ & multiplications\\
  $4t^2+4t+1$ & additions
 \end{tabular}
\end{center}
Step \refbm{BM4} does not increase the order of this upper bound because:
\begin{itemize}
 \item the coefficient $\omega_i$ is computed by $i$  multiplications and $i$ additions, thus the computation of $\omega(x)$ requires a number of multiplications and additions less than $e^2$;
 \item the computation of the error $E_i$ requires two polynomial evaluations and one division; using the same procedure seen for Chien's search (see algorithm \ref{alg:chien}) we can evaluate polynomials $\omega(x)$ and $\sigma'(x)$ in all the field elements with a negligible cost compared to the cost of computing the coefficients $\omega_i$'s and the following $e$ divisions $\dfrac{\omega(X_i^{-1})}{\sigma'(X_i^{-1})}$.
 \end{itemize}
 
\begin{rem}
 In section \ref{sect:decodingRS}, we studied Horner's method to evaluate polynomials for computing the syndromes $S_i=\rw(\alpha^i)$. Despite that we prefer use Chien's search to evaluate the polynomials $\omega(x)$ and $\sigma'(x)$ because these polynomials have lower degree than $\rw(x)$ ($n>d-1>t\geq e$) and are evaluated over a wider range of elements than $\rw(x)$. These two differences make it inefficient to implement the evaluation of $\omega(x)$ and $\sigma'(x)$ in a manner similar to used for syndrome computations, since it would takes $n-1$  circuits like figure \ref{fig:horner}. While only two circuits like figure \ref{fig:chien} are sufficient, if we use Chien's search.
\end{rem}

\begin{rem}\label{rem:pipelinedBM}
We note the \BM decoding algorithm \ref{alg:BM} is an example of decoder for Reed-Solomon codes in which the steps 3 and 4 of the general outline are not necessarily consecutive (see section \ref{sect:decodingRS}). Indeed we note that if the computation of $\omega(x)$ is brought forward at the end of \refbm{BM2}, then the evaluations of the polynomials $\omega(x)$ and $\sigma'(x)$ in $\alpha^i$ can be simultaneously implemented with the same evaluation of $\sigma(x)$ done in step \BM.3.
Moreover, each time that the evaluation of $\sigma(x)$ gives a value equal to zero, the division necessary to calculate $E_i$ in \BM.4 can be implement in a separate divider, while the polynomial evaluations go on. For example if $\sigma(\alpha)=0$, then the divider calculates $\dfrac{\omega(\alpha)}{\sigma'(\alpha)}$ in \BM.4 while the circuits that implement Chien's search  continue to evaluate $\sigma(x)$, $\omega(x)$, and $\sigma'(x)$ in $\alpha^2$, $\alpha^3$, \dots
Thanks to this pipelined strategy, after the $n$ loops in which Chien's search evaluates the polynomials 
$\sigma(x)$, $\omega(x)$, and $\sigma'(x)$ in all the field elements, some of the divisions needed to compute the error values may have been already executed. Thus in practice the number of divisions that remain to perform in \BM.4, after the end of step \BM.3, is less than $e$.\\
\end{rem}

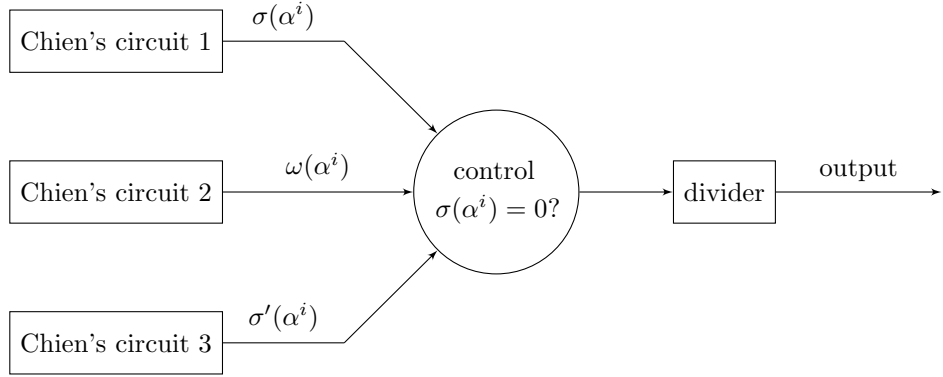
\begin{figure*}[h!]
\begin{center}
\begin{tikzpicture}[auto, node distance=2cm,>=latex']
    \node [block] (c1){\footnotesize Chien's circuit 1};
    \node [block, below of=c1](c2){\footnotesize Chien's circuit 2};
    \node [block, below of=c2](c3){\footnotesize Chien's circuit 3};
    \node [coordinate, right of=c1, node distance=3cm](con){};
  \node [sum, right of=c2, node distance=5cm, text width=16mm](mem){\footnotesize $\;$ control \\  $\sigma(\alpha^i)=0?$};
    \node [coordinate, right of=c3, node distance=3cm](c4){}; 
    \node [block, right of=mem, node distance=3cm](div){\footnotesize divider};
    \node [output, right of=div, node distance=3cm](output){};

    \draw [-] (c1) -- node{{\footnotesize $\sigma(\alpha^i)$}}(con);
    \draw [->] (c2) -- node{{\footnotesize $\omega(\alpha^i)$}}(mem);
    \draw [->] (c4) -- (mem);
    \draw (c3)-- node{{\footnotesize $\sigma'(\alpha^i)$}}(c4);
    \draw [->] (con) -- (mem);
   \draw [->](mem) -- (div);
   \draw [->] (div) -- node{{\footnotesize output}}(output);
\end{tikzpicture}
\caption{pipelined architecture for \BM.3 and \BM.4}\label{fig:pipelinedBM}
\end{center}
\end{figure*}

\begin{rem}
 One last observation on the polynomial evaluation: if the field $\F$ has characteristic equal to $2$, then $$\sigma'(x)=\sum_{\substack{j=1\\ j \text{ odd}}}^e \sigma_jx^j$$
 Thus the value $\sigma'(\alpha^i)$ can be obtained as a byproduct of the computation of  $\sigma(\alpha^i)$ summing the odd terms of the latter.\\
\end{rem}

In addition, we show with the next proposition that it is possible to obtain a new formula to calculate error values in which the error-evaluator polynomial $\omega(x)$ doesn't appear. 
\begin{prop}\label{prop:BM3} Let $\tau^{(d-1)}(x)$ be the last polynomial of the polynomial sequence $\left\{\tau^{(j)}(x)\right\}_{j=0}^{d-1}$ (definition \ref{defi:tau_gamma}) and let $\sigma'(x)$ be the formal derivative of the error-locator polynomial $\sigma(x)$. Then for any $i\in\{1,2,\dots,e\}$, we have that
\begin{equation}\label{forney2}
E_i=-\frac{\left(X_i^{-1}\right)^{d-2}}{\sigma'(X_i^{-1})\tau^{(d-1)}(X_i^{-1})}
\end{equation}
\end{prop}
\begin{proof}
We write the relation (\ref{eq:bm0}) seen in proposition \ref{prop:BM0} setting $i=d-1$ and we obtain that
$$\omega(x)\tau^{(d-1)}(x)-\sigma(x)\gamma^{(d-1)}(x)=x^{d-2}$$
Evaluating the former equality in $X_i^{-1}$, we have 
$$\omega(X_i^{-1})=\frac{\left(X_i^{-1}\right)^{d-2}}{\tau^{(d-1)}(X_i^{-1})}$$
At last  substituting this expression in Forney's formula  we can conclude the proof.\\
\end{proof}
For this reason, it is possible to implement the \BM decoder without computing the error-evaluator polynomial $\omega(x)$ saving memory space and operations. In this case  we can substitute  step \refbm{BM4}  with
 \begin{enumerate}[\bfseries \BM.4b]
   \item \label{BM4b} 
         \begin{tabbing}
           \texttt{for } \= $i=1,2,\dots,e$ \texttt{ do }\\[0,2cm]
           \>$E_i:=-\dfrac{\left(X_i^{-1}\right)^{d-2}}{\sigma'(X_i^{-1})\tau^{(d-1)}(X_i^{-1})}$; \\
           \texttt{endfor}
         \end{tabbing}        
 \end{enumerate}
Step \BM.4b can be implemented via the same pipelined strategy described in remark \ref{rem:pipelinedBM} for \refbm{BM4}. In this way the polynomial evaluation needed by (\ref{forney2}) are simultaneously implemented with step \BM.3 and afterwards step \BM.4b requires at most $e$ divisions and $e$ multiplications.\\

We give an example of decoding of a Reed-Solomon code using the \BM decoding algorithm \ref{alg:BM}:
\begin{ex}\label{ex_bm} 
Let $\alpha$ be a primitive element of $\mathbb{F}_{2^4}$ satisfying $\alpha^4+\alpha+1=0$. Consider  over $\mathbb{F}_{2^4}$ the code, $RS(15,9,\alpha)$ generated by 
$$g(x)=(x-\alpha)(x-\alpha^2)\cdots(x-\alpha^8)$$
The code has distance $9$, so $t=4$.  Suppose that the sent codeword is $\cw=\vct{0}$ and the error vector is
$$\e(x)=\alpha^2x^2+\alpha x^8+\alpha^7x^{13}$$
Clearly $\rw(x)=\e(x)$ and the syndrome values are:
$$\begin{array}{llll}
    S_1=\rw(\alpha)=\alpha^{12} & S_2=\rw(\alpha^2)=0    &S_3=\rw(\alpha^3)=0  & S_4=\rw(\alpha^4)=\alpha^{5}\\
    S_5=\rw(\alpha^5)=\alpha^{11}   & S_6=\rw(\alpha^6)=\alpha^{13} &S_7=\rw(\alpha^7)=\alpha^{3} &S_8=\rw(\alpha^8)=\alpha
  \end{array}$$ 
Hence the iterations in \refbm{BM2} are:
\begin{tabbing}
  ($i=0$)$\xrightarrow{\hspace*{1cm}}\,$ \= $\Delta_0=S_1=\alpha^{12}$;\\
                          \>$\sigma^{(1)}(x)=\sigma^{(0)}(x)-\Delta_0x\tau^{(0)}=1+\alpha^{12}x$; \\
                          \>$\Delta_0\neq0 \text{ and }2D(0)<1\Longrightarrow$ \= $D(1)=1$;\\
                          \>\> $\tau^{(1)}(x)=\dfrac{1}{\Delta_0}=\alpha^3$;\\
\end{tabbing}
\begin{tabbing}
  ($i=1$)$\xrightarrow{\hspace*{1cm}}\,$ \= $\Delta_1=S_2+S_1\sigma_1^{(1)}=\alpha^{9}$;\\
                          \>$\sigma^{(2)}(x)=\sigma^{(1)}(x)-\Delta_1x\tau^{(1)}=1$; \\
                          \>$2D(1)\geq2\Longrightarrow$ \= $D(2)=1$;\\
                          \>\> $\tau^{(2)}(x)=\alpha^3x $;\\
\end{tabbing}
\begin{tabbing}
  ($i=2$)$\xrightarrow{\hspace*{1cm}}\,$ \= $\Delta_2=S_3+S_2\sigma_1^{(2)}=0$;\\
                          \>$\sigma^{(3)}(x)=\sigma^{(2)}(x)=1$; \\
                          \>$\Delta_2=0 \Longrightarrow$ \= $D(3)=1$;\\
                          \>\> $\tau^{(3)}(x)=\alpha^3x^2 $;\\
\end{tabbing}
\begin{tabbing}
  ($i=3$)$\xrightarrow{\hspace*{1cm}}\,$   \= $\Delta_3=S_4+S_3\sigma_1^{(3)}=\alpha^{5}$;\\
                          \>$\sigma^{(4)}(x)=\sigma^{(3)}(x)-\Delta_3x\tau^{(3)}=1+\alpha^8x^3$; \\
                          \>$\Delta_3\neq0 \text{ and }2D(3)<4\Longrightarrow$ \= $D(4)=3$;\\
                          \>\> $\tau^{(4)}(x)=\dfrac{\sigma^{(3)}}{\Delta_3}=\alpha^{10}$;\\
\end{tabbing}
\begin{tabbing}
  ($i=4$)$\xrightarrow{\hspace*{1cm}}\,$ \= $\Delta_4=S_5+S_4\sigma_1^{(4)}+S_3\sigma_2^{(4)}+S_2\sigma_3^{(4)}=\alpha^{11}$;\\
                          \>$\sigma^{(5)}(x)=\sigma^{(4)}(x)-\Delta_4x\tau^{(4)}=1+\alpha^6x+\alpha^8x^3$; \\
                          \>$2D(4)\geq5\Longrightarrow$ \= $D(5)=3$;\\
                          \>\> $\tau^{(5)}(x)=\alpha^{10}x$;\\
\end{tabbing}
\begin{tabbing}
  ($i=5$)$\xrightarrow{\hspace*{1cm}}\,$\= $\Delta_5=S_6+S_5\sigma_1^{(5)}+S_4\sigma_2^{(5)}+S_3\sigma_3^{(5)}=\alpha^{14}$;\\
                   \>$\sigma^{(6)}(x)=\sigma^{(5)}(x)-\Delta_5x\tau^{(5)}=1+\alpha^6x+\alpha^9x^2+\alpha^8x^3$; \\
                          \>$2D(5)\geq6\Longrightarrow$ \= $D(6)=3$;\\
                          \>\> $\tau^{(6)}(x)=\alpha^{10}x^2$;\\
\end{tabbing}
\begin{tabbing}
 ($i=6$)$\xrightarrow{\hspace*{1cm}}\,$\=$\Delta_6=S_7+S_6\sigma_1^{(6)}+S_5\sigma_2^{(6)}+S_4\sigma_3^{(6)}=0$;\\
                          \>$\sigma^{(7)}(x)=\sigma^{(6)}(x)=1+\alpha^6x+\alpha^9x^2+\alpha^8x^3$; \\
                          \>$\Delta_6=0 \Longrightarrow$ \= $D(7)=3$;\\
                          \>\> $\tau^{(7)}(x)=\alpha^{10}x^3$;\\
\end{tabbing}
\begin{tabbing}
 ($i=7$)$\xrightarrow{\hspace*{1cm}}\,$\=$\Delta_7=S_8+S_7\sigma_1^{(7)}+S_6\sigma_2^{(7)}+S_5\sigma_3^{(7)}=0$;\\
                         \>$\sigma^{(8)}(x)=\sigma^{(7)}(x)=1+\alpha^6x+\alpha^9x^2+\alpha^8x^3$; \\
                          \>$\Delta_7=0\Longrightarrow$ \= $D(8)=3$;\\
                          \>\> $\tau^{(8)}(x)=\alpha^{10}x^4$;
\end{tabbing}
At the end  \refbm{BM2}  sets $e=D(8)=3$ and 
$$\sigma(x)=\alpha^8x^3+\alpha^9x^2+\alpha^6x+1=
\left(1+\alpha^{2}x\right)\left(1+\alpha^8x\right)\left(1+\alpha^{13}x\right)$$
After that \refbm{BM3} calculates that
$$X_1=\alpha^{2}, \quad X_2=\alpha^{8}, \quad X_3=\alpha^{13}$$
Hence the error positions are 
$$p_1=2, \quad p_2=8, \quad p_3=13$$
Finally in \refbm{BM4} the error-evaluator polynomial is computed 
$$\omega(x)=\alpha^6x^2+\alpha^3x+\alpha^{12}$$
and the error values are found by  Forney's formula:
$$E_1=\alpha^2,\quad E_2=\alpha, \quad E_3=\alpha^7 $$
Thus 
$$\e(x)=E_1x^{p_1}+E_2x^{p_2}+E_3x^{p_3}=\alpha^2x^2+\alpha x^8+\alpha^7x^{13}$$
and the received word $\rw$ is correctly decoded in $\cw=\rw-\e=\vct{0}$.\\
\end{ex}

\subsection*{Inversionless \BM decoder}
 As remembered in section \ref{sect:decodingRS}, the inversion is one of the most complicated operations  among those of finite field arithmetic. For this reason, an implementation  in which inversions can be replaced by some multiplications is preferable.
 
 For the \BM decoder the discrepancy inversions executed the step \BM.2 can be avoided introducing the following tools:
 \begin{defi}\label{defi:inversionless}
  We recursively define the function $\beta:\N\rightarrow\F^{\,*}$ and the polynomial families $\left\{\myhat\sigma^{(i)}(x)\right\}_{i=0}^{d-1}$ and $\left\{\myhat\tau^{(i)}(x)\right\}_{i=0}^{d-1}$ as $\beta(0)=1$, $\myhat\sigma^{(0)}(x)=1$,\\ $\myhat\tau^{(0)}(x)=1$ and
   \begin{align*}
    &\myhat\sigma^{(i+1)}(x)=\beta(i)\myhat\sigma^{(i)}(x)-\myhat\Delta_ix\myhat\tau^{(i)}(x)\\[0,2cm]
  &\myhat\tau^{(i+1)}(x)=\begin{cases} x\myhat\tau^{(i)}(x)\quad&\text{ if }\;\myhat\Delta_i=0 \text{ or } 2D(i)\geq i+1\\
                 \myhat\sigma^{(i)}(x)\quad&\text{ if }\;\myhat\Delta_i\neq0\text{ and }2D(i)<i+1
        \end{cases}\\
  &\beta(i+1)=\begin{cases} \beta(i)\quad&\text{ if }\;\myhat\Delta_i=0 \text{ or } 2D(i)\geq i+1\\
                \myhat\Delta_i\quad&\text{ if }\;\myhat\Delta_i\neq0\text{ and }2D(i)<i+1
        \end{cases} 
 \end{align*}
where $\myhat\Delta_i$ is  the coefficient of the term $x^i$ in the polynomial $\myhat\sigma^{(i)}(x)S(x)$.
\end{defi}

It is immediate to prove by induction that:
\begin{prop}\label{prop:inversionless}
 If $i\in\{0,1,\dots,d-1\}$, then 
 $$\myhat\tau^{(i)}(x)=\beta(i)\tau^{(i)}(x)\quad \text{ and }\quad \myhat\sigma^{(i)}(x)=\displaystyle\pt{\prod_{j=0}^{i-1}\beta(j)}\sigma^{(i)}(x)$$
 \end{prop}
 \begin{proof}
  The base case is trivial. Suppose that the equalities are true for $i$, then
  $$\myhat\Delta_i=\displaystyle\pt{\prod_{j=0}^{i-1}\beta(j)}\Delta_i$$
  and thus
  \begin{align*}
   \myhat\sigma^{(i+1)}&=\beta(i)\myhat\sigma^{(i)}(x)-\myhat\Delta_ix\myhat\tau^{(i)}(x)=\\
   &=\beta(i)\pt{\prod_{j=0}^{i-1}\beta(j)}\sigma^{(i)}(x)
   -\displaystyle\pt{\prod_{j=0}^{i-1}\beta(j)}\Delta_ix\beta(i)\tau^{(i)}(x)=\\
   &=\pt{\prod_{j=0}^{i}\beta(j)}\pt{\sigma^{(i)}(x)-\Delta_ix\tau^{(i)}(x)}=\pt{\prod_{j=0}^{i}\beta(j)}\sigma^{(i+1)}
 \end{align*}
 Moreover we have that either 
 $$\myhat\tau^{(i+1)}(x)= x\myhat\tau^{(i)}(x)=x\beta(i)\tau^{(i)}(x)=\beta(i+1)\tau^{(i+1)}(x)$$
 or
 $$\myhat\tau^{(i+1)}(x)= \myhat\sigma^{(i)}(x)=\displaystyle\pt{\prod_{j=0}^{i-1}\beta(j)}\sigma^{(i)}(x)=
 \dfrac{\myhat\Delta_i}{\Delta_i}\sigma^{(i)}(x)=\beta(i+1)\tau^{(i+1)}(x)$$
 This concludes the proof.\\
 \end{proof}

The \BM decoding algorithm \ref{alg:BM} implemented with the polynomials $\myhat\sigma^{(i)}(x)$ and $\myhat\tau^{(i)}(x)$ in the place of the polynomials $\sigma^{(i)}(x)$ and $\tau^{(i)}(x)$ actually finds scalar multiples $b\cdot\sigma(x)$ and $b\cdot\omega(x)$ instead of the polynomials $\sigma(x)$ and $\omega(x)$, where
$$b=\prod_{j=0}^{d-2}\beta(j)$$
However, it is obvious that Chien's search will find the same error positions and it follows from Forney's formula (\ref{forney}) that the same error vales are obtained. While formula (\ref{forney2}) has to be modified in
$$E_i=-\frac{b\,\beta(d-1)\left(X_i^{-1}\right)^{d-2}}{\myhat\sigma'(X_i^{-1})\myhat\tau^{(d-1)}(X_i^{-1})}$$

\section{Avoiding Malfunctions in  \BM} \label{sect:BMmalf}
In this section we will study the necessary and sufficient conditions to make the \BM decoding algorithm \ref{alg:BM} for Reed-Solomon codes a $t$-bounded distance decoder.\\

Consider again the code $RS(n,d,\alpha)$ with error correction capability equal to $t$. Let $\cw$, $\e$ and $\rw=\cw+\e$ respectively be the transmitted codeword, the vector error and the received word. 
The input given to the \BM decoding algorithm \ref{alg:BM} is $\rw$ and no assumption is made on the weight of $\e$. Let $\mathcal{B}$ be the set defined in section \ref{sect:error-correting codes}, we have already shown that if $\rw\in\mathcal B$ the \BM decoding algorithm \ref{alg:BM} either decodes correctly $\rw$ (when $\wt(\e)\leq t$) or runs into an inevitable decoder error occurs (if $\wt(\e)> t$).  It remains to study what happens when $\rw \notin \mathcal B$.  For this aim we denote:
\begin{itemize}
 \item[-]$\tl{e}$ the degree of $\sigma^{(d-1)}(x)$ and   $$\sigma^{(d-1)}(x)=\sigma^{(d-1)}_{\tl{e}}x^{\tl{e}}+\sigma^{(d-1)}_{\tl{e}-1}x^{\tl{e}-1}+\cdots+\sigma^{(d-1)}_1 x+1,$$
 \item[-]$\tl{X}_1,\tl{X}_2,..., \tl{X}_{\tl{e}}$ the inverses of the roots of $\sigma^{(d-1)}(x)$ calculated in \refbm{BM3};
 \item[-]$\tl{E}_i$ the $i$th assumed error value,
 \item[-]$\vct{\tl{e}}$ the assumed error vector.
\end{itemize}
Moreover we define $\tl{S}_i=\vct{\tl{e}}(\alpha^i)$ and $\tl{S}(x)=\tl{S}_1+\tl{S}_2x+\cdots+\tl{S}_{d-1}x^{d-1}$. Finally, let  $\vct{\tl{c}}=\rw-\vct{\tl{e}}$ be again the output vector.

First of all we give two example of possible malfunctions for the \BM decoder.
\begin{ex}\label{ex_bm_mal1}
Let $\alpha$ be a primitive element of $\mathbb{F}_{2^4}$ satisfying $\alpha^4+\alpha+1=0$. Consider  over $\mathbb{F}_{2^4}$ the code, $RS(15,5,\alpha)$ generated by 
$$g(x)=(x-\alpha)(x-\alpha^2)(x-\alpha^3)(x-\alpha^4)$$
The code has distance $5$, so $t=2$.  Suppose that the sent codeword is $\cw=\vct{0}$ and the error vector is
$$\e(x)=\alpha^6+\alpha^3x+\alpha^4x^2+x^{7}$$
Clearly $\rw(x)=\e(x)$, $e=4$  and the syndrome values are:
$$\begin{array}{ll}
    S_1=\rw(\alpha)=\alpha^{3} & S_2=\rw(\alpha^2)=\alpha^5 \\
    S_3=\rw(\alpha^3)=\alpha^7  & S_4=\rw(\alpha^4)=\alpha^{8}\\
  \end{array}$$ 
The iterations \refbm{BM2} are:
\begin{tabbing}
  ($i=0$)$\xrightarrow{\hspace*{1cm}}\,$ \= $\Delta_0=S_1=\alpha^{3}$;\\
                          \>$\sigma^{(1)}(x)=\sigma^{(0)}(x)-\Delta_0x\tau^{(0)}=1+\alpha^{3}x$; \\
                          \>$\Delta_0\neq0 \text{ and }2D(0)<1\Longrightarrow$ \= $D(1)=1$;\\
                          \>\> $\tau^{(1)}(x)=\dfrac{1}{\Delta_0}=\alpha^{12}$;\\
\end{tabbing}
\begin{tabbing}
  ($i=1$)$\xrightarrow{\hspace*{1cm}}\,$ \= $\Delta_1=S_2+S_1\sigma_1^{(1)}=\alpha^5+\alpha^6=\alpha^{9}$;\\
                          \>$\sigma^{(2)}(x)=\sigma^{(1)}(x)-\Delta_1x\tau^{(1)}=1+\alpha^2x$; \\
                          \>$2D(1)\geq2\Longrightarrow$ \= $D(2)=1$;\\
                          \>\> $\tau^{(2)}(x)=\alpha^{12}x $;\\
\end{tabbing}
\begin{tabbing}
  ($i=2$)$\xrightarrow{\hspace*{1cm}}\,$ \= $\Delta_2=S_3+S_2\sigma_1^{(2)}=\alpha^7+\alpha^7=0$;\\
                          \>$\sigma^{(3)}(x)=\sigma^{(2)}(x)=1+\alpha^2x$; \\
                          \>$\Delta_2=0 \Longrightarrow$ \= $D(3)=1$;\\
                          \>\> $\tau^{(3)}(x)=\alpha^{12}x^2 $;\\
\end{tabbing}
\begin{tabbing}
  ($i=3$)$\xrightarrow{\hspace*{1cm}}\,$   \= $\Delta_3=S_4+S_3\sigma_1^{(3)}=\alpha^8+\alpha^9=\alpha^{12}$;\\
                          \>$\sigma^{(4)}(x)=\sigma^{(3)}(x)-\Delta_3x\tau^{(3)}=1+\alpha^2x+\alpha^9x^3$; \\
                          \>$\Delta_3\neq0 \text{ and }2D(3)<4\Longrightarrow$ \= $D(4)=3$;\\
                          \>\> $\tau^{(4)}(x)=\dfrac{\sigma^{(3)}}{\Delta_3}=\alpha^3+\alpha^{5}x$;
\end{tabbing}
Hence  the \BM decoding algorithm \ref{alg:BM} sets $\tl{e}=3$ and since 
$$\sigma^{(4)}(x)=1+\alpha^2x+\alpha^9x^3=
\alpha^9\left(x+\alpha^{5}\right)\left(x^2+\alpha^5x+\alpha\right)$$
\refbm{BM3}  cannot proceed because the polynomial $x^2+\alpha^5x+\alpha$ is irreducible over $\F[x]$. In this case the \BM decoder have to declare a decoder failure. Note  that $\tl{e}=D(4)$ but $\tl{e}>t$.
\end{ex}

\begin{ex} \label{ex_bm_mal2}
Consider the same code of example \ref{ex_bm_mal1}, supposing that $\cw=\vct{0}$ and $$\e(x)=\alpha^3x+\alpha x^2+x^{10}$$
then $\rw(x)=\e(x)$, $e=4$  and the syndrome values are:
$$\begin{array}{ll}
    S_1=\rw(\alpha)=\alpha^{6} & S_2=\rw(\alpha^2)=\alpha^5 \\
    S_3=\rw(\alpha^3)=\alpha^5  & S_4=\rw(\alpha^4)=\alpha^{5}\\
  \end{array}$$ 
The iterations \refbm{BM2}  in this case are:
\begin{tabbing}
  ($i=0$)$\xrightarrow{\hspace*{1cm}}\,$ \= $\Delta_0=S_1=\alpha^{6}$;\\
                          \>$\sigma^{(1)}(x)=\sigma^{(0)}(x)-\Delta_0x\tau^{(0)}=1+\alpha^{6}x$; \\
                          \>$\Delta_0\neq0 \text{ and }2D(0)<1\Longrightarrow$ \= $D(1)=1$;\\
                          \>\> $\tau^{(1)}(x)=\dfrac{1}{\Delta_0}=\alpha^{9}$;\\
\end{tabbing}
\begin{tabbing}
  ($i=1$)$\xrightarrow{\hspace*{1cm}}\,$ \= $\Delta_1=S_2+S_1\sigma_1^{(1)}=\alpha^5+\alpha^{12}=\alpha^{14}$;\\
                          \>$\sigma^{(2)}(x)=\sigma^{(1)}(x)-\Delta_1x\tau^{(1)}=1+\alpha^{14}x$; \\
                          \>$2D(1)\geq2\Longrightarrow$ \= $D(2)=1$;\\
                          \>\> $\tau^{(2)}(x)=\alpha^{9}x $;\\
\end{tabbing}
\begin{tabbing}
  ($i=2$)$\xrightarrow{\hspace*{1cm}}\,$ \= $\Delta_2=S_3+S_2\sigma_1^{(2)}=\alpha^5+\alpha^2=\alpha^8$;\\
                          \>$\sigma^{(3)}(x)=\sigma^{(2)}(x)-\Delta_2x\tau^{(2)}=1+\alpha^{14}x+\alpha^2x^2$; \\
                          \>$\Delta_2\neq0 \text{ and }2D(2)<3\Longrightarrow$ \= $D(3)=2$;\\
                          \>\> $\tau^{(3)}(x)=\dfrac{\sigma^{(2)}}{\Delta_2}=\alpha^7+\alpha^{6}x $;\\
\end{tabbing}
\begin{tabbing}
  ($i=3$)$\xrightarrow{\hspace*{1cm}}\,$ \= $\Delta_3=S_4+S_3\sigma_1^{(3)}+S_2\sigma_2^{(3)}=\alpha^5+\alpha^4+\alpha^7=\alpha^{11}$;\\
                          \>$\sigma^{(4)}(x)=\sigma^{(3)}(x)-\Delta_3x\tau^{(3)}=1+x$; \\
                          \>$2D(3)\geq4\Longrightarrow$ \= $D(4)=2$;\\
                          \>\> $\tau^{(4)}(x)=\alpha^7x+\alpha^{6}x^2 $;
\end{tabbing}
Thus $\tl{e}=1$ and, since $\sigma^{4}(x)=1+x$, \refbm{BM3}  calculates that
$\tl{X}_1=1=\alpha^{0}$. The assumed error position is $j_1=0$.
Finally in \refbm{BM4}  the algorithm calculate  the error-evaluator polynomial $\omega(x)=\alpha^6$ and the assumed error value $\tl{E}_1=\alpha^6$. Thus $\vct{\tl{e}}=\tl{E}_1x^{j_1}=\alpha^6$ and the received word $\rw$ is wrongly decoded in $$\vct{\tl{c}}=\rw-\vct{\tl{e}}=\alpha^6+\alpha^3x+\alpha x^2+x^{10}$$
Evidently a decoder malfunction has happened. Note that $D(4)=2>\tl{e}=1$.
\end{ex}

In both the examples it occurs that $\rw\notin \mathcal B$ but the \BM decoder does not declare the failure and carries on the decoding, bringing about a decoder malfunction. In order to avoid such behavior of the \BM decoder we have to add to its implementations the conditions necessary and sufficient to assure that  $\rw\in\mathcal B$. For this reason we give the following proposition. 

\begin{prop}\label{prop:malfunctionBM}
 With previous notations, we have that $\rw\in\mathcal B$ if and only if
 \begin{enumerate}[\ \ \ (1)]
   \item $\sigma^{(d-1)}(x)$ has exactly  $\tl{e}$ distinct roots, all belonging to the field $\F$ and all different from zero;
   \item $D(d-1)=\tl{e}$ and $\tl{e}\leq t$;
  \end{enumerate}
\end{prop}
\begin{proof}
 If  $\rw\in\mathcal B$, then there exists a unique  $\vct{c}_1\in\ RS(n,d,\alpha)$ such that $\rw \in \overline{B_t(\vct{c}_1)}$. We repeat the construction saw in section \ref{sect:BMalg} with $\vct{c_1}$  in  place of $\cw$. Since $\wt(\rw-\vct{c}_1)\leq t$, then $\vct{c}_1$ is equal to the output vector $\vct{\tl{c}}$ and  the right side conditions are accomplished.  On the other hand, we suppose that conditions (1) and (2) are satisfied and we prove that $\rw\in\mathcal B$. Since $\rw=\vct{\tl{c}}+\vct{\tl{e}}$ and $\tl{e}\leq t$,  we need only to prove that $\vct{\tl{c}}\in RS(n,d,\alpha)$. For condition (1) we may assume that $\sigma^{(d-1)}(x)=\displaystyle\prod_{i=1}^{\tl{e}}\left(1-\tl{X}_ix\right)$ and  that  
 $\tl{S}_i=\vct{\tl{e}}(\alpha^i)=\displaystyle\sum_{j=1}^{\tl{e}}\tl{E}_{j}\tl{X}_j^i$ for any $i=1,2,\dots,\tl{e}$. Moreover we know that $\tl{E}_i=-\dfrac{\omega^{(d-1)}(\tl{X}_i^{-1})}{\left(\sigma^{(d-1)}\right)'(\tl{X}_i^{-1})}$ for any $i=1,2,\dots,\tl{e}$.  Then
 \begin{align*}
  \sigma^{(d-1)}(x)\tl{S}(x)&=
   \prod_{i=1}^{\tl{e}}\left(1-\tl{X}_ix\right)\left(\sum_{l=1}^{d-1}x^{l-1}\sum_{j=1}^{\tl{e}}\tl{E}_{j}\tl{X}_j^l\right)=\\
   &=\prod_{i=1}^{\tl{e}}\left(1-\tl{X}_ix\right)\left(\sum_{j=1}^{\tl{e}}\tl{E}_{j}\sum_{l=1}^{d-1}x^{l-1}\tl{X}_j^l\right)=\\
   &=\sum_{j=1}^{\tl{e}}\left(\tl{E}_{j}\prod_{i=1}^{\tl{e}}\left(1-\tl{X}_ix\right)\sum_{l=1}^{d-1}x^{l-1}\tl{X}_j^l\right)=\\
   &=\sum_{j=1}^{\tl{e}}\left(-\frac{\omega^{(d-1)}(\tl{X}_j^{-1})}{\left(\sigma^{(d-1)}\right)'(\tl{X}_j^{-1})}
   \prod_{i=1}^{\tl{e}}\left(1-\tl{X}_ix\right)\sum_{l=1}^{d-1}x^{l-1}\tl{X}_j^l\right)\\
 \end{align*}
Now we observe that the sum on $l$ is equal to $\displaystyle\tl{X}_j\frac{1-(\tl{X}_jx)^{d-1}}{1-\tl{X}_jx}$ and that
$$\left(\sigma^{(d-1)}\right)'(\tl{X}_j^{-1})=-\tl{X}_j\prod_{i=1\\i\neq j}^{\tl{e}}\left(1-\tl{X}_i\tl{X}_j^{-1}\right).$$
Thus we have
$$\sigma^{(d-1)}(x)\tl{S}(x)=
  \sum_{j=1}^{\tl{e}}\left(\omega^{(d-1)}(\tl{X}_j^{-1}) \prod_{i\neq j}\frac{1-\tl{X}_ix}{1-\tl{X}_i\tl{X}_j^{-1}}
  \left(1-(\tl{X}_jx)^{d-1}\right)\right)=$$
$$=\sum_{j=1}^{\tl{e}}\left(\omega^{(d-1)}(\tl{X}_j^{-1}) \prod_{i\neq j}\frac{1-\tl{X}_ix}{1-\tl{X}_i\tl{X}_j^{-1}}\right)
  +(\tl{X}_jx)^{d-1}\left[\sum_{j=1}^{\tl{e}}\omega^{(d-1)}(\tl{X}_j^{-1}) \prod_{i\neq j}\frac{1-\tl{X}_ix}{1-\tl{X}_i\tl{X}_j^{-1}}\right]$$
  
\noindent The first sum on the right side is the Lagrange interpolation formula for the unique polynomial that has degree less then $\tl{e}$ and takes value $\omega^{(d-1)}(\tl{X}_j^{-1})$ at $\tl{X}_j^{-1}$ for $i=1,2,\dots,\tl{e}$. Since $\Deg{\omega^{(d-1)}(x)}\leq D(d-1)-1=\tl{e}-1$, this unique polynomial must be $\omega^{(d-1)}(x)$ itself. So we obtain that
$$\sigma^{(d-1)}(x)\tl{S}(x)\equiv \omega^{(d-1)}(x)\Mod{x^{d-1}}$$
But for proposition \ref{prop:BM1} we know that $$\sigma^{(d-1)}(x)S(x)\equiv \omega^{(d-1)}(x)\Mod{x^{d-1}}$$
Hence it follows that $S(x)=\tl{S}(x)$, that is 
$$S_i-\tl{S}_i=\rw(\alpha^i)-\vct{\tl{e}}(\alpha^i)=\vct{\tl{c}}(\alpha^i)=0$$
for any $i=1,2,\dots, d-1$. This concludes the proof.
\end{proof}

Thus we can conclude that adding the conditions (1) and (2) of proposition \ref{prop:malfunctionBM} at the end of step \refbm{BM2}, the \BM decoding algorithm \ref{alg:BM} becomes a $t$-bounded distance decoding algorithm.\\

Observe that since the elements $\tl{X}_1,\tl{X}_2,..., \tl{X}_{\tl{e}}$ are the inverses of the roots of $\sigma^{(d-1)}(x)$, by (\ref{eq:bm0})  it follows that 
$$\tl{E}_i=-\frac{\omega^{(d-1)}(\tl{X}_i^{-1})}{\left(\sigma^{(d-1)}\right)'(\tl{X}_i^{-1})}=
-\frac{\left(\tl{X}_i^{-1}\right)^{d-2}}{\left(\sigma^{(d-1)}\right)'(\tl{X}_i^{-1})\cdot\tau^{(d-1)}(\tl{X}_i^{-1})}$$ for any $i=1,2,\dots,\tl{e}$. Hence the proposition \ref{prop:malfunctionBM}  also holds for the implementation of the \BM decoder with the formula (\ref{forney2}) in the place of Forney's formula (\ref{forney}) (see proposition \ref{prop:BM3}).

\section{Comparing \fPGZ.2  and  \BM.2}\label{sect:comparison}
From a computational point of view, when the average number of errors is much less than the error correction capability $t$, then the \fPGZ  decoding algorithm \ref{alg:fastPGZ} (on page \pageref{alg:fastPGZ}) has some significant advantages in comparison with the \BM decoding algorithm \ref{alg:BM} (on page \pageref{alg:BM}). In fact, whatever the number of errors is, the complexity of \BM decoding algorithm \ref{alg:BM} is of order $O(t^2)$, whereas the complexity of \fPGZ  decoding algorithm \ref{alg:fastPGZ} is of order $O(et)$ and so it depends on $e$. 
Nevertheless the program structure in the \fPGZ  decoding algorithm \ref{alg:fastPGZ}, especially when the singularity gap is strictly greater than 1, is more complicated and requires more program code than the \BM decoding algorithm \ref{alg:BM}.\\

Now we want to examine and compare the structure of the two iterative algorithms in order to find some common features and computations. For this purpose we state the following two remarks and theorem \ref{thm:comparison}.
\begin{rem}\label{rem:reciproco}
 We recall the one-to-one correspondence between vectors in $(\F)^i$ and polynomials in $\F[X]$ of degree less than $i$ studied in section \ref{sect:cyclic codes} and given by
 $$\vct{v}=(v_0,v_1, \dots,v_{i-1}) \longleftrightarrow \vct{v}(x)=\sum_{j=0}^{i-1}v_{j} x^j=v_0+v_1x+\cdots+v_{i-1}x^{i-1}$$
 and we note that there exists another possible correspondence. Indeed if $\vct{v}=(v_0,v_1, \dots,v_{i-1})\in (\F)^i$ then we will call  $\overline{\vct{v}}(x)$ the unique polynomial of $\F[x]$ of degree less than $i$ given by:
 $$\overline{\vct{v}}(x)=\sum_{j=0}^{i-1}v_{j} x^{i-1-j}=v_{i-1}+v_{i-2}x\cdots+v_0x^{i-1}$$
 Usually $\overline{\vct{v}}(x)$ is called the \emph{reciprocal polynomial} of $\vct{v}(x)$. Observe that 
 $$\overline{\vct{v}}(x)=x^{i-1}\cdot\vct{v}(x^{-1})$$
 \end{rem}
 
\begin{rem}\label{aiutino3}
 Let $i$ be an index less or equal to $t$ and such that $\det(A_i)\neq 0$. Recall that in section \ref{sect:fastPGZ} we have defined $\vct{w}^{(i)}=\left(w_0^{(i)},w_1^{(i)},\dots,w_{i-1}^{(i)}\right)^T$ as the vector in $(\F)^i$ that satisfies
 \begin{equation}\label{eq:aiutino3}
        A_i\vct{w}^{(i)}=-\begin{mymatrix} S_{i+1}\\ S_{1+2}\\ \vdots \\ S_{2i}\\
                          \end{mymatrix} 
\end{equation}
 Now we introduce the polynomial  
 $$P_{\vct{w}^{(i)}}(x)\overset{def}{=}x\cdot\overline{\vct{w}}^{(i)}(x)+1
 =w_0^{(i)}x^i+w_1^{(i)}x^{i-1}+\cdots+w_{i-1}^{(i)}x+1$$
 and we observe that for any $j\in\{1,2,\dots,t-i\}$, the element $\varepsilon_j$, defined as the sum 
 $$\varepsilon_j=S_{i+j}w_0^{(i)}+S_{i+j+1}w_1^{(i)}+\cdots+S_{2i+j-1}w_{i-1}^{(i)}+S_{2i+j}\in\F$$
 (see definiton \ref{defi:epsilon}) is equal to the coefficient of $x^{2i+j-1}$ in $P_{\vct{w}^{(i)}}(x)S(x)$.
 \end{rem}
 
 We further recall that the vectors $\vct{w}^{(i)}$'s are computed in step \reffpgz{fpgz2} of the \fPGZ  decoding algorithm \ref{alg:fastPGZ}. While the \BM decoding algorithm \ref{alg:BM} in step \refbm{BM2} calculates the polynomials $\sigma^{(i)}(x)$'s as in definition \ref{defi:sigma_omega_i}. Finally we remind that for any $j\in\{0,1,\dots,d-1\}$, $\Delta_j$ is defined as the coefficient of $x^j$ in the polynomial $\sigma^{(j)}(x)S(x)$. 
 
\begin{thm} \label{thm:comparison}
Let $i$ an index less or equal to $t$ and such that $\det(A_i)\neq 0$ and let $\vct{w}^{(i)}$ be the vector in $(\F)^i$ that satisfies  (\ref{eq:aiutino3}). If $r$ is the singularity gap as in definition \ref{defi:epsilon}, i.e.\  
$$r=\min\left\{j\in\{1,2,\dots,t-i\} \;|\; \varepsilon_j\neq 0\right\}\in\N^+$$ 
then we have that:
$$P_{\vct{w}^{(i)}}(x)=\sigma^{(2i)}(x)=\sigma^{(2i+1)}(x)=\cdots=\sigma^{(2i+r-1)}(x)$$
Moreover it follows from this that $\Delta_{2i+r-1}=\varepsilon_r\neq 0$ and that if $r\geq2$, then $\Delta_{2i}=\Delta_{2i+1}=\cdots=\Delta_{2i+r-2}=0$;
\end{thm}
\begin{proof}
From  (\ref{eq:aiutino3}) and $\varepsilon_1=\varepsilon_2=\cdots=\varepsilon_{r-1}=0$ it follows that
 \begin{equation*}
  \begin{mymatrix}  S_1      & \cdots & S_i\\
                   \vdots   &        &\vdots\\
                   S_{i}    & \cdots & S_{2i-1}\\
                   S_{i+1}    & \cdots & S_{2i}\\
                   \vdots   &        &\vdots\\
                   S_{i+r-1}    & \cdots & S_{2i-r-2}\\
  \end{mymatrix}\vct{w}^{(i)}=-\begin{mymatrix} S_{i+1}\\ \vdots \\ S_{2i}\\ S_{2i+1}\\ \vdots\\ S_{2i+r-1}\\
                              \end{mymatrix}
 \end{equation*}
that is
$$ \sum_{l=0}^{i-1} S_{j+l}\,w_l^{(i)} + S_{i+j}=0 \quad\quad\quad\forall \;j=1,2,\dots,i+r-1 $$
Since $S(x)=\displaystyle\sum_{j=1}^{d-1}S_jx^{j-1}$ and 
$P_{\vct{w}^{(i)}}(x)=w_0^{(i)}x^i+w_1^{(i)}x^{i-1}+\cdots+w_{i-1}^{(i)}x+1$,
the sums on the right side in the former equalities are the coefficients of $x^i, x^{i+1},\dots, x^{2i+r-2}$ in the polynomial $P_{\vct{w}^{(i)}}(x)S(x)$. These coefficients are all zero, hence there exists  $f^{(i)}(x)\in\F[X]$ such that $\Deg{f^{(i)}}\leq i-1$ and
$$ P_{\vct{w}^{(i)}}(x)S(x)\equiv f^{(i)}(x) \Mod{x^{2i+j}} \quad \quad \forall\; j=0,1,\dots, r-1$$
So we have obtained that $\left(P_{\vct{w}^{(i)}}(x), f^{(i)}(x)\right)\in\mathcal{M}_{2i+j}$, $\forall\;j=0,1,\dots,r-1$. Now we may use the corollary \ref{cor:BM2} to conclude that, for any $j=0,1,\dots,r-1$
 \begin{equation}\label{sistemacomp}
   \begin{cases}
      P_{\vct{w}^{(i)}}(x)=u(x)\sigma^{(2i+j)}(x)\\
      f^{(i)}(x)=u(x)\omega^{(2i+j)}(x)\\
  \end{cases}
 \end{equation}
where $u(x)=\gcd\left(P_{\vct{w}^{(i)}}(x), f^{(i)}(x)\right)$.\\
It immediately follows  from the first equations of systems (\ref{sistemacomp}) that
$$\sigma^{(2i)}(x)=\sigma^{(2i+1)}(x)=\dots=\sigma^{(2i+r-1)}(x)$$
and therefore, using that $\sigma^{(j+1)}(x)=\sigma^{(j)}(x)-\Delta_jx\tau^{(j)}(x)$
we find that if $r\geq2$, then $\Delta_{2i}=\Delta_{2i+1}=\dots=\Delta_{2i+r-2}=0$.
It also follows from  (\ref{sistemacomp}) with $j=0$ that
$$\Deg{\sigma^{(2i)}(x)}\leq i\quad \text{ and }\quad \Deg{\omega^{(2i)}(x)}\leq i-1.$$
In particular, due to the relation $\sigma^{(2i)}(x)S(x)\equiv\omega^{(2i)}(x)\Mod{x^{2i}}$, the upper bound of $\Deg{\omega^{(2i)}}$ implies that the coefficients of $x^i,x^{i+1},\dots, x^{2i-1}$ in $\sigma^{(2i)}(x)S(x)$ are all zero. Hence, we may write 
$$\sigma^{(2i)}(x)=\sigma^{(2i)}_ix^i+\sigma^{(2i)}_{i-1}x^{i-1}+\cdots+\sigma^{(2i)}_1x+1$$ and
$$\begin{cases}
     S_{i+1}+\sigma^{(2i)}_1S_i+\cdots+\sigma^{(2i)}_iS_1=0\\
     S_{i+2}+\sigma^{(2i)}_1S_{i+1}+\cdots+\sigma^{(2i)}_iS_2=0\\
     \cdots\\
     S_{2i}+\sigma^{(2i)}_1S_{2i-1}+\cdots+\sigma^{(2i)}_iS_i=0\\
\end{cases}$$
that is 
$$A_i\begin{mymatrix}
             \sigma^{(2i)}_i\\ \sigma^{(2i)}_{i-1}\\ \vdots\\ \sigma^{(2i)}_1                      
    \end{mymatrix}=-\begin{mymatrix}
                        S_{i+1}\\ S_{i+2}\\\vdots \\S_{2i}    
                   \end{mymatrix}$$
Since $\det(A_i)\neq0$, from the definition of the vector $\vct{w}^{(i)}$  we deduce that 
$$\vct{w}^{(i)}=\begin{mymatrix}
             \sigma^{(2i)}_i\\ \sigma^{(2i)}_{i-1} \\\vdots\\ \sigma^{(2i)}_1                      
    \end{mymatrix}$$
Thus $P_{\vct{w}^{(i)}}(x)=\sigma^{(2i)}(x)$.

Finally we will show that $\Delta_{2i+r-1}=\varepsilon_r$. As we saw in remark \ref{aiutino3}, $\varepsilon_r$ is the coefficient of $x^{2i+r-1}$ in the polynomial $P_{\vct{w}^{(i)}}(x)S(x)$, while (\ref{defi:Delta}) on page \pageref{defi:Delta} defines $\Delta_{2i+r-1}$ as the coefficient of $x^{2i+r-1}$ in the polynomial $\sigma^{(2i+r-1)}(x)S(x)$. Since $P_{\vct{w}^{(i)}}(x)=\sigma^{(2i+r-1)}(x)$, it follows that $\Delta_{2i+r-1}=\varepsilon_r$, which is different from zero by definition.\\
\end{proof}

Theorem \ref{thm:comparison} shows that every vector $\vct{w}^{(i)}$ computed during the step \reffpgz{fpgz2} corresponds always to the polynomial  $\sigma^{(2i)}(x)$, in the sense that the components of $\vct{w}^{(i)}$  are  the coefficients of $\sigma^{(2i)}(x)$. In other words the intermediate outcomes of the \fPGZ  decoding algorithm are a subset of the intermediate outcomes of the \BM decoder.
The vice versa is not true, indeed, as we can see comparing example \ref{ex_fastpgz} and example \ref{ex_bm},
some of polynomials $\sigma^{(j)}(x)$'s may not appear in the set of the polynomials $P_{\vct{w}^{(j)}}(x)$'s.
In example \ref{ex_fastpgz} we decode the received word $\rw(x)=\alpha^2x^2+\alpha x^8+\alpha^7x^{13}$ by the \fPGZ  decoder and we have as intermediate outcomes the polynomials 
$$P_{\vct{w}^{(1)}}(x)=0x+1=1 \hspace{.5cm} P_{\vct{w}^{(3)}}(x)=\alpha^8x^3+\alpha^9x^2+\alpha^6x+1$$
While, in example \ref{ex_bm}, decoding the same received word with the \BM decoder, we obtain the polynomials:
$$\sigma^{(0)}(x)=1\hspace{.5cm} \sigma^{(1)}(x)=1+\alpha^{12}\hspace{.5cm} \sigma^{(2)}(x)=\sigma^{(3)}(x)=1$$
$$\hspace{-0.2cm}\sigma^{(4)}(x)=1+\alpha^8x^3\hspace{.7cm}\sigma^{(5)}(x)=1+\alpha^6x+\alpha^8x^3$$
$$\sigma^{(6)}(x)=\sigma^{(7)}(x)=\sigma^{(8)}(x)=1+\alpha^6x+\alpha^9x^2+\alpha^8x^3$$
It easy verify that 
$$P_{\vct{w}^{(1)}}(x)=\sigma^{(2)}(x)=\sigma^{(3)}(x)$$
$$P_{\vct{w}^{(3)}}(x)=\sigma^{(6)}(x)=\sigma^{(7)}(x)=\sigma^{(8)}(x)$$
and the polynomial $\sigma^{(1)}(x)$, $\sigma^{(4)}(x)$, $\sigma^{(5)}(x)$ do not corresponds to any polynomial of the form $P_{\vct{w}^{(i)}}(x)$.\\

Moreover theorem \ref{thm:comparison} states that when the singularity gap $r$ of the $i$th iterative step of \reffpgz{fpgz2} is greater than 1 (that is the \fPGZ  decoder ``jumps'' from $i$ to $i+r$) then even the iterative steps \refbm{BM2} ``jump'' from $2i$ to $2i+r$, since 
$\sigma^{(2i)}(x)=\sigma^{(2i+1)}(x)=\cdots=\sigma^{(2i+r-1)}(x)$.

\chapter{Error Value Formulas}
In the former chapters we saw that both in the \fPGZ  decoding algorithm \ref{alg:fastPGZ} and in the \BM decoding algorithm \ref{alg:BM} the error locations are determined by the roots of the error-locator polynomial $\sigma(x)$ through an exhaustive search among the elements of $\F$. Instead, the error values are calculated by different methods in each different decoding algorithm. In the \reffpgz{fpgz4}, the error values $E_1,E_2,\dots,E_e$ are calculated by the \BP algorithm \ref{alg:BP} to solve the following linear system 
\begin{equation*}
 \begin{mymatrix} X_1   & X_2  & \cdots & X_e \\
                 X_1^2   & X_2^2  & \cdots & X_e^2 \\
                 \vdots &\vdots & &\vdots\\
                 X_1^e   & X_2^e  & \cdots & X_e^e \\
  \end{mymatrix}
 \begin{mymatrix} E_1\\ E_2 \\ \vdots \\ E_e\\
   \end{mymatrix}= 
   \begin{mymatrix} S_{1}\\ S_{2} \\ \vdots \\ S_{e}\\
   \end{mymatrix}. 
\end{equation*}  
where $X_i$ is the inverse of $i$th root of $\sigma(x)$ and $S_i$ is the $i$th syndrome (see definitions \ref{defi:sigma_omega} and \ref{defi:syndromes}).
In the \BM decoding algorithm \ref{alg:BM}, the error values are usually determined simultaneously to the error positions using the error-evaluator polynomial $\omega(x)$ and Forney's formula (\ref{forney}) (see remark \ref{rem:pipelinedBM}). But we have  also shown in proposition \ref{prop:BM3} that the computation of the error-evaluator polynomial in the \BM decoding algorithm  can be avoided. Indeed, we proved that for any $i\in\{1,2,\dots,e\}$:
$$E_i=-\frac{\left(X_i^{-1}\right)^{d-2}}{\sigma'(X_i^{-1})\tau^{(d-1)}(X_i^{-1})}$$
where the polynomial $\tau^{(d-1)}(x)$ is one of the auxiliary polynomials calculate in the \BM decoding algorithm.\\

The aim of this chapter is to complete the comparison between the \fPGZ  and the \BM decoders showing how to calculate the error values  in both  the decoding algorithms by using the linear algebra of the syndrome matrix and the byproducts of the computations for the error-locator polynomial.

\section{Horiguchi's formula}\label{sect:horiguchi}
In \cite{horiguchi} T. Horiguchi presents a new error-evaluation method for computing error values in decoding Reed-Solomon codes through  the \BM decoding algorithm using simply tools of linear algebra. In this section we presented this method in order to apply it both to the \fPGZ  and the \BM decoding algorithms. With this aim, we introduce the following tools:       
\begin{defi} \label{defi:B_and_v}
For any $i\in\{1,2,\dots, t\}$, let $B_i$ be the $i\times (i+1)$ matrix defined by
 $$ B_i=\begin{mymatrix}
                         S_1          & S_2             & \cdots & S_{i+1} \\
                         S_2          & S_3             & \cdots & S_{i+2}\\
                         \vdots       & \vdots          &        & \vdots \\
                         S_{i}    & S_{i+1}        & \cdots & S_{2i}\\
         \end{mymatrix}$$ 
 Moreover for any $j=0,1,\dots,i$ we call $B_i^{(j)}$ the square $i\times i$ matrix obtained from $B_i$ removing the $(j+1)$th column and  we define 
 $$k_{j}^{(i)}=(-1)^{i+j}\det(B_i^{(j)})$$
 Finally we call $\vct{k}^{(i)}$ the vector in $(\F)^{i+1}$ given by
        $$ \vct{k}^{(i)}=\begin{mymatrix}
                          k^{(i)}_0\\k^{(i)}_1\\ \vdots\\ k^{(i)}_{i}
                          \end{mymatrix}$$
\end{defi}

\begin{rem}
 Note that for any $i\in\{2,3,\dots,t\}$ we have that
\begin{equation}\label{rel_det}
k_{i}^{(i)}=(-1)^{2i}\det(B_i^{(i)})=\det(A_i)= \displaystyle \sum_{j=0}^{i-1} S_{i+j}\,k_j^{(i-1)}
\end{equation}
where $A_i$ is the $i\times i$ leading principal minor of the syndrome matrix (see definition \ref{def:syndromematrix}). Moreover for any $i\in\{1,2, \dots,t\}$ and $j\in\{1,2, \dots,i\}$, expanding the determinant of the following singular matrix along the last row 
$$\begin{mymatrix}
                         S_1          & S_2             & \cdots & S_{i+1} \\
                         S_2          & S_3             & \cdots & S_{i+2}\\
                         \vdots       & \vdots          &        & \vdots \\
                         S_{i}    & S_{i+1}        & \cdots & S_{2i}\\
                         S_{j}    & S_{j+1}        & \cdots & S_{j+i}\\
         \end{mymatrix}$$
we obtain that $\displaystyle\sum_{l=0}^{i} S_{j+l}k_l^{(i)} =0$. Hence $B_i\vct{k}^{(i)}=\vct{0}$ for any $\forall\,i\in\{1,2, \dots,t\}$.
\end{rem}

We recall that $\cw\in RS(n,d,\alpha)$, $\e$ and $\rw=\cw+\e$ respectively denote the transmitted codeword, the vector error and the received word. Whereas $t$ is the error correction capability of the code $RS(n,d,\alpha)$. From now on, throughout this  chapter, we suppose that $e=\wt(\e)>1$.
\begin{prop}\label{prop:kernel}
If $e=\wt(\e)$, then 
$$\dim(\ker(B_e))=\dim(\ker(B_{e-1}))=1$$  
In particular $\ker(B_e)=\emph{\text{span}}\{\vct{k}^{(e)}\}$ and $\ker(B_{e-1})=\emph{\text{span}}\{\vct{k}^{(e-1)}\}$.
\end{prop}
\begin{proof}
We saw in proposition \ref{prop:syndromematrix} that the square matrix $A_e$ is non-singular, so the matrix $B_e$ has rank equal to $e$ and it immediately follows that the dimension of $\ker(B_e)$ is equal to 1. Moreover as $k_{e}^{(e)}=\det(A_e)\neq0$ and $\vct{k}^{(e)}\in \ker(B_e)$,  it is clear that $\ker(B_e)=\text{span}\{\vct{k}^{(e)}\}$.
Since $$\displaystyle k_{e}^{(e)}= \sum_{j=0}^{e-1} S_{e+j}\,k_j^{(e-1)}$$ from $k_{e}^{(e)}\neq 0$ it follows that there is an index $l\in \{1,2,\dots,e\}$ such that $k_l^{(e-1)} \neq 0$. Therefore, $\vct{k}^{(e-1)}\neq \vct{0}$ and there is a $(e-1)\times (e-1)$ non-singular minor in the matrix $B_{e-1}$. Hence we can conclude that $\rk(B_{e-1})=e-1$, i.e.\ $\ker(B_{e-1})$ has dimension equal to 1, and that  $\ker(B_{e-1})=\text{span}\{\vct{k}^{(e-1)}\}$.\\
\end{proof}

We recall the notation used for polynomials of $\F[x]$ (see remark \ref{rem:reciproco}) :
\begin{rem}\label{rem:reciproco2}
 If $\vct{v}=(v_0, v_1, \dots,v_{i-1})\in (\F)^i$ then we will call:
 \begin{align*}
 &\vct{v}(x)\overset{def}{=}\sum_{j=0}^{i-1}v_{j} x^j=v_0+v_1x+\cdots+v_{i-1}x^{i-1}\\
 &\overline{\vct{v}}(x)\overset{def}{=}\sum_{j=0}^{i-1}v_{j} x^{i-1-j}=v_{i-1}+v_{i-2}x+\cdots+v_0x^{i-1}\\
 &P_{\vct{v}}(x)\overset{def}{=}1+v_{i-1}x+\dots+v_{0}x^1 
 \end{align*}
 \end{rem}

\begin{cor}\label{cor:kersigma}
Let $\overline{\vct{k}}^{(e)}(x)$ be the polynomial obtained from vector $\vct{k}^{(e)}$ as in remark \ref{rem:reciproco2}, that is 
$\overline{\vct{k}}^{(e)}(x)=k_0^{(e)}x^e+\cdots+k_{e-1}^{(e)}x+k_e^{(e)}$. 
If $e\leq t$, then we have that
$$ \sigma(x)=\displaystyle \frac{1}{k_{e}^{(e)}}\,\overline{\vct{k}}^{(e)}(x)$$
\end{cor}
\begin{proof}
If $\sigma(x)=\sigma_ex^e + \cdots+\sigma_1x+1$, then we call   
$\vct{\sigma}=\left(\sigma_e, \dots, \sigma_1, 1\right)^T\in(\F)^{e+1}  $
Since $e\leq t$, from the linear system (\ref{sistema1Matrice}) on page \pageref{sistema1Matrice} it follows that $\vct{\sigma} \in \ker(B_e)$. Thus, by proposition \ref{prop:kernel}, there exists  $\lambda \in \F$ such that $\vct{\sigma}=\lambda \vct{k}^{(e)}$. Using the last component of the former equality we find  that $\lambda=\frac{1}{v^{(e)}_{e}}$. Thus
$$ \vct{\sigma}= \frac{1}{k_{e}^{(e)}}\vct{k}^{(e)}$$
and the  proof is completed.
\end{proof}

We recall that if $e=\wt(\e)$ is the number of the error that have occurred in the positions $p_1,p_2,\dots,p_e$, then the elements $X_1,X_2,\dots,X_e$ are defined as $X_i=\alpha^{p_i}$ and are the inverses of the roots of the error-locator polynomial $\sigma(x)$ (see definition \ref{defi:sigma_omega}).

\begin{thm}[Horiguchi's formula]\label{thm:horiguchi}
 Consider the vectors $\vct{k}^{(e-1)}$ and $\vct{k}^{(e)}$ as in definition \ref{defi:B_and_v} and let $\overline{\vct{k}}^{(e-1)}(x)$ be the polynomial obtained as in remark \ref{rem:reciproco2}.
 For any $i\in\{1,2,\dots,e\}$, the following equation holds: 
 $$E_i=-\frac{k_{e}^{(e)} \cdot \left(X_i^{-1}\right)^{2(e-1)}}{\sigma'(X_i^{-1})\overline{\vct{k}}^{(e-1)}(X_i^{-1})}$$
where $\sigma'(x)$ is the formal derivative of the error-locator polynomial $\sigma(x)$.
\end{thm}
\begin{proof} 
 In order to prove the theorem, we introduce the following elements:\\
 for any $i, j \in\{1,2,\dots,e\}$ let
 \begin{align*}
    &\widehat{E}^{(i)}_j=E_jX_j(1-X_jX_i^{-1})\in \F,\\
    &\widehat{X}^{(i)}_j=X_jX_i^{-1} \in \F.
 \end{align*}
 Moreover for any $l\geq1$ we consider the sum given by
 $$\widehat{S}^{(i)}_l=\sum_{\substack{j=1\\j\neq i}}^{e}\widehat{E}^{(i)}_j\left(\widehat{X}^{(i)}_j\right)^{l-1}.$$
 Finally let $\widehat{A}_{e-1}^{(i)}$ be  the $(e-1)\times(e-1)$ matrix defined by
 $$\widehat{A}_{e-1}^{(i)}=\begin{mymatrix}
                  \widehat{S}^{(i)}_1 & \widehat{S}^{(i)}_2 & \cdots &\widehat{S}^{(i)}_{e-1}\\ 
                  \widehat{S}^{(i)}_2 & \widehat{S}^{(i)}_3 & \cdots &\widehat{S}^{(i)}_{e}\\
                  \vdots            & \vdots            &        &\vdots\\
                  \widehat{S}^{(i)}_{e-1} & \widehat{S}^{(i)}_e & \cdots &\widehat{S}^{(i)}_{2e-3}\\ 
               \end{mymatrix}$$
 It is easy to verify that $\widehat{A}_{e-1}^{(i)}=(\widehat{V}^{(i)})^T\widehat{D}^{(i)}\widehat{V}^{(i)}$, where $\widehat{V}^{(i)}$ and $\widehat{D}^{(i)}$ are the square matrices defined by:
 $$\widehat{V}^{(i)}=\begin{mymatrix}
                  1 & \widehat{X}^{(i)}_1 & \cdots &\left(\widehat{X}^{(i)}_1\right)^{e-2}\\ 
                  \vdots            & \vdots            &        &\vdots\\
                  1 & \widehat{X}^{(i)}_{i-1} & \cdots &\left(\widehat{X}^{(i)}_{i-1}\right)^{e-2}\\
                  1 & \widehat{X}^{(i)}_{i+1} & \cdots &\left(\widehat{X}^{(i)}_{i+1}\right)^{e-2}\\
                  \vdots            & \vdots            &        &\vdots\\
                  1 & \widehat{X}^{(i)}_e & \cdots &\left(\widehat{X}^{(i)}_e\right)^{e-2}\\ 
               \end{mymatrix}$$
  $$\widehat{D}^{(i)}=\begin{mymatrix}
                 \widehat{E}^{(i)}_1\widehat{X}^{(i)}_1   & 0      &\cdots  &\cdots   &\cdots    & 0 \\
                  0                                   &\ddots  &        &         &          & \vdots \\
                 \vdots                               &        &\widehat{E}^{(i)}_{i-1}\widehat{X}^{(i)}_{i-1}   &  &    \\
                                                      &        &    &\widehat{E}^{(i)}_{i-1}\widehat{X}^{(i)}_{i-1}  &  &\vdots \\
                 \vdots                                &        &        &         & \ddots         & 0 \\
                 0                                 & \cdots &\cdots  &\cdots        &   0 & \widehat{E}^{(i)}_{e}\widehat{X}^{(i)}_{e}\\
         \end{mymatrix}$$  
 
 \noindent Hence, by Binet's formula, we have that
 \begin{align*}
      \det(\widehat{A}_{e-1}^{(i)})&=\prod_{\substack{j=1\\j\neq i}}^{e}\widehat{E}^{(i)}_j\widehat{X}^{(i)}_j
                     \prod_{\substack{l>j\\l,j\neq i}}(\widehat{X}^{(i)}_l-\widehat{X}^{(i)}_j)^2=\\
                 &=\left(X_i^{-1}\right)^{(e-1)(e-2)}\prod_{\substack{j=1\\j\neq i}}^{e}E_jX_j(1-X_jX_i^{-1})\prod_{\substack{l>j\\l,j\neq i}}(X_l-X_j)^2
 \end{align*}
 On the other hand, we observe that 
 \begin{align*}
          \widehat{S}^{(i)}_l& =\sum_{\substack{j=1\\j\neq i}}^{e}\widehat{E}^{(i)}_j\left(\widehat{X}^{(i)}_j\right)^{l-1}
                     \overset{\displaystyle\overset{\displaystyle\widehat{E}^{(i)}_i=0}{\big\downarrow}}{=}
                   \sum_{j=1}^{e}\widehat{E}^{(i)}_j\left(\widehat{X}^{(i)}_j\right)^{l-1}=\\
                  &=\sum_{j=1}^{e}E_jX_j(1-X_jX_i^{-1})\left(X_jX_i^{-1}\right)^{l-1}=\\
                  &=\sum_{j=1}^{e}E_jX_j^l\left(X_i^{-1}\right)^{l-1}-\sum_{j=1}^{e}E_jX_j^{l+1}\left(X_i^{-1}\right)^{l}=\\ 
                  &=S_l\left(X_i^{-1}\right)^{l-1}-S_{l+1}\left(X_i^{-1}\right)^{l}
 \end{align*} 
 Thus $\widehat{A}_{e-1}^{(i)}$ is equal to  
 $$\begin{mymatrix}
     S_1-S_2X_i^{-1} &   \cdots & S_{e-1}\left(X_i^{-1}\right)^{e-2}-S_{e}\left(X_i^{-1}\right)^{e-1}\\
     S_2X_i^{-1}-S_3\left(X_i^{-1}\right)^2 & \cdots & S_{e}\left(X_i^{-1}\right)^{e-1}-S_{e+1}\left(X_i^{-1}\right)^{e}\\ 
     \vdots                        &        &\vdots\\
     S_{e-1}\left(X_i^{-1}\right)^{e-2}-S_{e}\left(X_i^{-1}\right)^{e-1} &\cdots & S_{2e-3}\left(X_i^{-1}\right)^{2e-4}-S_{2e-2}\left(X_i^{-1}\right)^{2e-3}\\ 
               \end{mymatrix} $$
  and, by linearity, we have that 
 \begin{align*} 
     \det(\widehat{A}_{e-1}^{(i)})&=
          \det\begin{mymatrix}
               S_1 &  S_2X_i^{-1} &\cdots & S_e\left(X_i^{-1}\right)^{e-1}\\
               S_2X_i^{-1} &  S_3\left(X_i^{-1}\right)^{2} &\cdots & S_{e+1}\left(X_i^{-1}\right)^{e}\\
               \vdots & \vdots & & \vdots\\
               S_{e-1} \left(X_i^{-1}\right)^{e-2} &  S_e\left(X_i^{-1}\right)^{e-1} &\cdots & S_{2e-2}\left(X_i^{-1}\right)^{2e-3}\\
               1 & 1 & \cdots &1\\
              \end{mymatrix}=\\ \\
          &= \left(X_i^{-1}\right)^{\frac{(e-2)(e-1)}{2}}\det
               \begin{mymatrix}
               S_1 &  S_2X_i^{-1} &\cdots & S_e\left(X_i^{-1}\right)^{e-1}\\
               S_2 &  S_3X_i^{-1} &\cdots & S_{e+1}\left(X_i^{-1}\right)^{e-1}\\
               \vdots & \vdots & & \vdots\\
               S_{e-1}&  S_eX_i^{-1} &\cdots & S_{2e-2}\left(X_i^{-1}\right)^{e-1}\\
               1 & 1 & \cdots &1\\
              \end{mymatrix}=\\ \\
        &=\left(X_i^{-1}\right)^{\frac{(e-2)(e-1)}{2}}\sum_{j=0}^{e-1}\left(X_i^{-1}\right)^{\frac{(e-1)e}{2}-j}k_j^{(e-1)}=\\
         &=\left(X_i^{-1}\right)^{\frac{(e-2)(e-1)}{2}+\frac{(e-1)e}{2}-e+1}\sum_{j=0}^{e-1}\left(X_i^{-1}\right)^{e-1-j}k_j^{(e-1)}=\\
          &=\left(X_i^{-1}\right)^{(e-2)(e-1)}\overline{\vct{k}}^{(e-1)}(X_i^{-1})
 \end{align*}
 Now, comparing the two formulas just obtained for $\det(A_i)$, we deduce that   
 $$\overline{\vct{k}}^{(e-1)}(X_i^{-1})=\prod_{\substack{j=1\\j\neq i}}^{e}E_jX_j(1-X_jX_i^{-1})\prod_{\substack{l>j\\l,j\neq i}}(X_l-X_j)^2$$
 Finally, using the last equation we can calculate that
 \begin{align*}
  \frac{k_{e}^{(e)}}{\overline{\vct{k}}^{(e-1)}(X_i^{-1})}=
  \frac{\det(A_e)}{\overline{\vct{k}}^{(e-1)}(X_i^{-1})}
      &=\frac{\displaystyle\prod_{j=1}^{e}E_jX_j\prod_{l>j}(X_l-X_j)^2}{\displaystyle\prod_{\substack{j=1\\j\neq i}}^{e}E_jX_j(1-X_jX_i^{-1})\prod_{\substack{l>j\\l,j\neq i}}(X_l-X_j)^2}=\\
      &=\frac{E_iX_i}{\displaystyle\prod_{\substack{j=1\\j\neq i}}^{e}
      (1-X_jX_i^{-1})}\prod_{\substack{j=1\\j\neq i}}^{e}(X_i-X_j)^2=\\
      &=E_iX_i \; X_i^{2(e-1)}\prod_{\substack{j=1\\j\neq i}}^{e}(1-X_jX_i^{-1})=\\
      &=-E_iX_i^{2(e-1)}\sigma'(X_i^{-1})
    \end{align*}  This concludes the proof.\\
\end{proof}
\vspace{2cm}
\begin{cor}[generalized Horiguchi's formula]\label{cor:horiguchi}\hspace{2cm}\\
 If $\vct{u}=(u_0,u_1,\dots,u_{e-1})\in \ker(B_{e-1})\setminus\{\vct{0}\}$, then for any $i\in\{1,2,\dots,e\}$ it holds that
 $$E_i=-\frac{\displaystyle\left(X_i^{-1}\right)^{2(e-1)}\sum_{j=0}^{e-1} S_{e+j}\,u_j}{\sigma'(X_i^{-1})\overline{\vct{u}}(X_i^{-1})}$$
 where $\overline{\vct{u}}(x)=u_0x^{e-1}+u_1 x^{e-2}+\cdots+u_{e-1}$ as in remark \ref{rem:reciproco2}.
\end{cor}
\begin{proof}
 By proposition \ref{prop:kernel}, we know that there exists $\lambda\in \F^*$ such that $$\vct{u}=\lambda\vct{k}^{(e-1)}$$
 Hence we have that
 \begin{align*}
  \frac{\displaystyle\left(X_i^{-1}\right)^{2(e-1)}\sum_{j=0}^{e-1} S_{e+j}\,u_j}{\sigma'(X_i^{-1})\overline{\vct{u}}(X_i^{-1})}&=
   \frac{\displaystyle\left(X_i^{-1}\right)^{2(e-1)}\sum_{j=0}^{e-1} S_{e+j}\,\lambda v^{(e-1)}_j}{\sigma'(X_i^{-1})\lambda\overline{\vct{k}}^{(e-1)}(X_i^{-1})}\quad\;=\\
   &=\frac{\displaystyle\left(X_i^{-1}\right)^{2(e-1)}\sum_{j=0}^{e-1} S_{e+j}\,v^{(e-1)}_j}{\sigma'(X_i^{-1})\overline{\vct{k}}^{(e-1)}(X_i^{-1})}\underset{\displaystyle\underset{\text{equation }(\ref{rel_det})}{\displaystyle\uparrow}}{=}\\
   &=\frac{ \left(X_i^{-1}\right)^{2(e-1)}k_{e}^{(e)}}{\sigma'(X_i^{-1})\overline{\vct{k}}^{(e-1)}(X_i^{-1})}=-E_i
 \end{align*}  Thus the proof is completed.
\end{proof}

Throughout the next two sections, the formula stated in corollary \ref{cor:horiguchi} will be the key to find new error-evaluation formulas for the \fPGZ  and the \BM decoding  algorithms.

\section{Application to the \fPGZ  Decoder}
\label{sect:horiguchiPGZ}
In this section we will show how use corollary \ref{cor:horiguchi} of Horiguchi's formula in order to find a new method of computing error values suited to  the \fPGZ  decoder. We recall the \fPGZ  decoding algorithm \ref{alg:fastPGZ} (summarized on page \pageref{alg:fastPGZ}): in step \reffpgz{fpgz2} the vectors  $\vct{w}^{(i)}\in(\F)^i$ satisfying
 \begin{equation*}
 A_i\vct{w}^{(i)}=-\begin{mymatrix} S_{i+1}\\ S_{i+2}\\ \vdots \\ S_{2i}\\
   \end{mymatrix}
 \end{equation*}
are computed for any index $i\leq t$ such that  $\det(A_i)\neq 0$.
In particular we will assume to have computed 
$$\theta=\max\left\{j<e \;|\;\det(A_j)\neq0\right\}$$
and we will show how use the vector $\vct{w}^{(\theta)}$, which is a byproduct of the computation of $\sigma(x)$, and Horiguchi's formula (theorem \ref{thm:horiguchi}) in order to calculate the error values.

\begin{prop}\label{prop:HoriPGZ}
 Let $\vct{w}^{(\theta)}=\left(w_0^{(\theta)},w_1^{(\theta)},\dots,w_{\theta-1}^{(\theta)}\right)^T\in(\F)^\theta$ be the vector satisfying 
 \begin{equation}\label{aiutino4}
 A_\theta\vct{w}^{(\theta)}=-\begin{mymatrix} S_{\theta+1}\\ S_{\theta+2} \\\vdots \\ S_{2\theta}\\
   \end{mymatrix}
 \end{equation} and
 \begin{align*}
&P_{\vct{w}^{(\theta)}}(x)
=w_0^{(\theta)}x^\theta+w_1^{(\theta)}x^{\theta-1}+\cdots+w_{\theta-1}^{(\theta)}x+1\in\F[x]\\[.3cm]
&\varepsilon_{e-\theta}=S_{e}w_0^{(\theta)}+S_{e+1}w_1^{(\theta)}+\cdots+S_{e+\theta-1}w_{\theta-1}^{(\theta)}+S_{e+\theta}\in\F
\end{align*}
(see remark \ref{rem:reciproco2} and definition \ref{defi:epsilon}).
 Then for any $i\in\{1,2,\dots, e\}$ we have that
 \begin{equation}\label{rel:horiPGZ}
  E_i=-\frac{\varepsilon_{e-\theta}\cdot \displaystyle\left( X_i^{-1}\right)^{e+\theta-1}}{\displaystyle\sigma'(X_i^{-1})P_{\vct{w}^{(\theta)}}( X_i^{-1})}
 \end{equation}
 \end{prop}
\begin{proof}
 Recalling definition \ref{defi:B_and_v}, system (\ref{aiutino4}) implies that
\begin{equation*}
 B_\theta\begin{mymatrix}\vct{w}^{(\theta)}\\ 1\\
   \end{mymatrix}=\vct{0}
\end{equation*}
Further the definition of $\theta$ implies that the coefficients 
$$\varepsilon_j=S_{\theta+j}w_0^{(\theta)}+\cdots+S_{2\theta+j-1}w_{\theta-1}^{(\theta)}+S_{2\theta+j}$$
are equal to zero for any $j\in\{1,2,\dots,e-\theta-1\}$.
It follows that
$$\begin{mymatrix}
     & B_{\theta}\\
    S_{\theta+1} & \cdots & S_{2\theta+1}\\
    \vdots &             &\vdots\\
    S_{e-1} & \cdots & S_{e+\theta-1}\\
   \end{mymatrix}\begin{mymatrix}
                     \vct{w}^{(\theta)} \\ 1
                 \end{mymatrix}=\vct{0}$$
Hence we have that $\begin{mymatrix}
           \vct{w}^{(\theta)} \\ 1 \\ \vct{0}^{e-\theta-1}
       \end{mymatrix}\in\ker(B_{e-1})\setminus\{\vct{0}\}$ and by corollary \ref{cor:horiguchi} we can conclude that:
\begin{align*}
 E_i&=-\frac{\displaystyle\left[S_{e+\theta}+\sum_{j=0}^{\theta-1} S_{e+j}w_j^{(\theta)}\right]\left( X_i^{-1}\right)^{2(e-1)}}{\displaystyle\sigma'(X_i^{-1})\left( X_i^{-1}\right)^{e-\theta-1}
 P_{\vct{w}^{(\theta)}}( X_i^{-1})}=\\ \\
    &=-\frac{\varepsilon_{e-\theta}\cdot\displaystyle\left( X_i^{-1}\right)^{e+\theta-1}}{\displaystyle\sigma'(X_i^{-1})P_{\vct{w}^{(\theta)}}( X_i^{-1})}.
\end{align*}
\end{proof}


Proposition \ref{prop:HoriPGZ} gives an alternative way to compute the error values $E_i$'s in the \fPGZ decoder without the need of separately solving the linear system (\ref{syst:GZ}) in step \reffpgz{fpgz4}. Indeed  (\ref{rel:horiPGZ}) relates each $E_i$,  independently of the others error values, directly  with the syndromes, the roots of $\sigma(x)$ and  the vector $\vct{w}^{(\theta)}$. Note that the latter  and the coefficient $\varepsilon_{e-\theta}$ are byproducts of the computation of $\sigma(x)$, thus no extra computations are needed to know them. Step \fPGZ.4 can be replaced with\\
\begin{enumerate}[\bfseries \fPGZ.4b]
 \item  \begin{tabbing}
           \texttt{for } \= $i=1,2,\dots,e$ \texttt{ do }\\[0,2cm]
           \>$E_i:=-\dfrac{\varepsilon_{e-\theta}\cdot \displaystyle\left( X_i^{-1}\right)^{e+\theta-1}}{\displaystyle\sigma'(X_i^{-1})P_{\vct{w}^{(\theta)}}( X_i^{-1})}$; \\
           \texttt{endfor}\\
         \end{tabbing}  
\end{enumerate}

\noindent As already seen in remark \ref{rem:pipelinedBM} for the \BM decoder, formula (\ref{rel:horiPGZ})  has the advantages of allowing to the \fPGZ decoder the calculation of error values together with those of positions. Indeed the polynomials $\sigma'(x)$ and $ P_{\vct{w}}^{(\theta)}(x)$ can be evaluated in $\alpha^i$ simultaneously to the error-locator polynomial during Chien's search of \reffpgz{fpgz3} and, when $\sigma(\alpha^i)=0$, the multiplication and the division necessary to calculate the occurring error value can be executed by a multiplier and a divider while the polynomial evaluations go on in the element $\alpha^{i+1}, \alpha^{i+2}, \dots$

With this procedure step \fPGZ.4b, after the polynomial evaluations, requires $e$ division and $e$ multiplications to compute the error values $E_i$'s.

\section{Comparison with the \BM Decoder}
Now consider the \BM decoding algorithm \ref{alg:BM} as described on page \pageref{alg:BM}. In particular we focus our attention on step 2 of algorithm \ref{alg:BM}. In that step the algorithm finds the polynomials $\sigma^{(i)}(x)$ and $\omega^{(i)}(x)$ such that $$\sigma^{(i)}(x)S(x)\equiv\omega^{(i)}(x)\Mod{x^i}$$
for any $i=1,2,\dots,d-1$ (see definition \ref{defi:sigma_omega} and proposition \ref{prop:BM0}).

In this section, using again Horiguchi's formula, we will show how the polynomials $\sigma^{(i)}(x)$'s can be used to find the error values in the place of the error-evaluator polynomial. For this purpose we will give the following definition, for which we recall that
$$\Delta_i\displaystyle=\sum_{j=0}^{i}S_{i+1-j}\,\sigma_j^{(i)}$$
is the coefficient of $x^i$ in the polynomial $\sigma^{(i)}(x)S(x)$ and that the function $D:\N\rightarrow\N$ is defined in definition \ref{defi:D(i)} on page \pageref{defi:D(i)}. 

\begin{defi}\label{defi:C}
Let $C:\mathbb N\rightarrow \mathbb Z$ the integer function defined by:
 \begin{align*} &C(0)=-1;\\
                &C(i+1)=\begin{cases}
                            C(i)   &\text{ if } \Delta_i=0 \text{ or } 2D(i)\geq i+1;\\
                            i     &\text{ otherwise;}
           \end{cases}
 \end{align*}
\end{defi}

We observe that $C(i)$ represents the iterative step where the most recent change of the function $D(i)$ occurred prior to step $i$. In other words, $$C(i)=j \Longleftrightarrow D(i)=D(i-1)=\cdots=D(j+1)>D(j)$$

\begin{lem}\label{lem:delta}
 Let $c=C(d-1)$. If $e\leq t$, then we have that
 $$c-D(c)=e-1$$
 and moreover $1+D(c) \leq e$.
\end{lem}
\begin{proof}
 By definition \ref{defi:C}, since $C(d-1)=c$, we have that
 $$D(d-1)=D(d-2)=\cdots=D(c+1)=c+ 1- D(c)>D(c)$$
 By corollary \ref{cor:BM2}, we know that $D(d-1)=e$. Hence $e>D(c)$ and $e=c+ 1- D(c)$.\\
\end{proof}
\newpage
\begin{prop}\label{prop:HoriBM}
 If $c=C(d-1)$ and $e\leq t$, then for any $i\in\{1,2,\dots e\}$ the following equation holds:
 \begin{equation}\label{rel:horiBM}
E_i=-\frac{\displaystyle\left( X_i^{-1}\right)^{c}\Delta_c}{\displaystyle\sigma'(X_i^{-1})\sigma^{(c)}( X_i^{-1})}  
 \end{equation}
 where $\Delta_c$ is the coefficient of $x^c$ in the polynomial $\sigma^{(c)}(x)S(x)$ as defined in section \ref{sect:BMalg}.
\end{prop}
\begin{proof}
The \BM decoding algorithm \ref{alg:BM} calculates the polynomials
$$\sigma^{(c)}(x)=\sigma_{D(c)}^{(c)}x^{D(c)}+\cdots+\sigma_{1}^{(c)}x+1$$
$$\omega^{(c)}(x)=\omega_{D(c)-1}^{(c)}x^{D(c)-1}+\cdots+\omega_{1}^{(c)}x+1$$
such that
$$\sigma^{(c)}(x)S(x)\equiv \omega^{(c)}(x)\Mod{x^{c}}$$
It follows that the coefficients of $x^{D(c)}, \dots, x^{c-1}$ in the polynomial $\sigma^{(c)}(x)S(x)$ are equal to zero. Hence if we call 
$\vct{\sigma}^{(c)}=\left(\sigma_{D(c)}^{(c)},\dots,\sigma_1^{(c)}, 1\right)^T\in(\F)^{D(c)+1}$, then we have
$$\begin{mymatrix}
   S_1 &\cdots & S_{D(c)+1}\\
   S_2 &\cdots & S_{D(c)+2}\\
   \vdots & &\vdots\\
   S_{c-D(c)} &\cdots & S_{c}\\
  \end{mymatrix}\vct{\sigma}^{(c)}=\vct{0}^{(c-D(c))}$$
By lemma \ref{lem:delta}, we obtain that
$$\begin{mymatrix}
   S_1 &\cdots & S_e\\
   S_2 &\cdots & S_{e+1}\\
   \vdots & &\vdots\\
   S_{e-1} &\cdots & S_{2e-2}\\
  \end{mymatrix}
  \begin{mymatrix}
   \vct{\sigma}^{(c)}\\ \vct{0}^{(e-D(c)-1)}
  \end{mymatrix}=B_{e-1}\begin{mymatrix}
   \vct{\sigma}^{(c)}\\ \vct{0}^{(e-D(c)-1)}
  \end{mymatrix}=\vct{0}^{(e-1)}$$ 
Thus $$\begin{mymatrix}
   \vct{\sigma}^{(c)}\\ \vct{0}^{(e-D(c)-1)}
  \end{mymatrix}\in \ker(B_{e-1})\setminus\{\vct{0}^{(e)}\}$$ 
Now we use corollary \ref{cor:horiguchi} and we have that
\begin{align*}
E_i&=-\frac{\displaystyle\left( X_i^{-1}\right)^{2(e-1)}
     \left[S_{e+D(c)}+\sum_{j=1}^{D(c)} S_{e+D(c)-j}\sigma_j^{(c)}\right]}
     {\displaystyle\sigma'(X_i^{-1})\left( X_i^{-1}\right)^{e-1-D(c)}\sigma^{(c)}( X_i^{-1})}\underset{\displaystyle\underset{e+D(c)=c+1}{\displaystyle\big\uparrow}}{=}\\ \\
   &=-\frac{\displaystyle\left( X_i^{-1}\right)^{c}\Delta_c}{\displaystyle\sigma'(X_i^{-1})\sigma^{(c)}( X_i^{-1})}
\end{align*}
This concludes the proof.
\end{proof}

Since the polynomial $\sigma^{(c)}(x)$ and the coefficient $\Delta_c$ are obtained as byproducts of the computation for $\sigma(x)$, the formula stated in \ref{prop:HoriBM} permits a more efficient decoding algorithm than Forney's formula (\ref{forney}) saving the cost of computing the error-evaluator polynomial $\omega(x)$.

\begin{rem}
We have already seen in proposition \ref{prop:BM3} that the computation of the error-evaluator polynomial $\omega(x)$ in  step \refbm{BM4} can be avoided. Indeed, we proved that for any $i\in\{1,2,\dots,e\}$
\begin{equation*}
 E_i=-\frac{\left(X_i^{-1}\right)^{d-2}}{\sigma'(X_i^{-1})\tau^{(d-1)}(X_i^{-1})}\tag{\ref{forney2}}
\end{equation*}

where the polynomial $\tau^{(d-1)}(x)$ is the last of the auxiliary polynomials calculate in step \refbm{BM2}. It is easy to show that the latter and  (\ref{rel:horiBM}) are the same formula. Indeed
\begin{align*}
 C(d-1)=c\, &\,\Rightarrow D(d-1)=D(d-2)=\cdots=D(c+1)>D(c)\\
                  &\,\Rightarrow \tau^{(d-1)}(x)=x\tau^{(d-2)}(x)=\cdots=x^{d-c-2}\tau^{(c+1)}(x)=\\
                  &\hspace*{2,3cm}=x^{d-c-2}\frac{\sigma^{(c)}(x)}{\Delta_c}
\end{align*}
Hence
$$E_i =-\frac{\left(X_i^{-1}\right)^{d-2}}{\sigma'(X_i^{-1})\tau^{(d-1)}(X_i^{-1})}=
-\frac{\displaystyle\left( X_i^{-1}\right)^{c}\Delta_c}{\displaystyle\sigma'(X_i^{-1})\sigma^{(c)}( X_i^{-1})}$$
\end{rem}

As seen in section \ref{sect:BMalg}, the polynomials in (\ref{rel:horiBM}) and in (\ref{forney2}) are evaluated in the field element using Chien's search circuit (figure \ref{fig:chien}).

\noindent We observe that
$$e\leq t\leq \frac{d-1}{2} \quad \Longrightarrow \quad  e\leq d-1-e$$
and moreover 
$$\Deg{\sigma^{(c)}(x)}\leq D(c)\leq e-1$$
$$\Deg{\tau^{(d-1)}(x)}\leq d-1-D(d-1)=d-1-e$$
Thus $\Deg{\sigma^{(c)}(x)}<\Deg{\tau^{(d-1)}(x)}$ and so (\ref{rel:horiBM}) can be a better choice than  (\ref{forney2}) because the evaluation of $\tau^{(d-1)}(x)$ in the field elements needs more circuit components than the evaluation of $\sigma^{(c)}(x)$.

\chapter{Parallel Implementation}
The continuing improvements of microelectronics technology leads to the availability of integrated circuits (microchips) with high-speed parallel architectures. Indeed, as at the state of the art one microchip can contain a large number of circuits and it is possible that independent tasks are accomplished simultaneously by different circuits of the same microchip. A very simple example is a microchip with $m$ adders:  with the computational time complexity of only one addition, this microchip  can add two vectors of $m$ components because the $i$th adder of the microchip adds the $i$th components of the two vectors simultaneously for $i=1,2,\dots,m$. The same idea is used to implement the multiplication of a vector of $m$ components by a scalar element with $m$ multipliers and with a time complexity equal to only one multiplication. Another example regards the sum of $m$ elements: observing that
\begin{align*}
\sum_{i=1}^m a_i&=\pt{\sum_{i=1}^{\Pint{m/2}} a_i}+ \pt{\sum_{i=\Pint{m/2}+1}^m a_i} 
\end{align*}
and repeating this splitting, the sum of $m$ elements can be computed with $\Pint{\log_2m}$ additions using  at most $m$ adders forming a tree of depth $\Pint{\log_2m}$ (see figure \ref{fig:adderTree}).

We can note immediately that the complexity of parallel algorithms is estimated in terms  of the time (as the classical sequential algorithms) and of the space (number of circuit elements) that they take.

We also observe that the improvement brought by the parallel implementation also depends on the inner parallelism of the instructions that have to be implemented: adding two vectors (or two polynomials) and multiplying a vector (or a polynomial) by a scalar element are operations that have a high level of inner parallelism. Instead the same is not true for the sum of $m$ elements. Indeed with $m$ adders the complexity of this operation decrease only from $m$ to $\Pint{\log_2m}$ additions.
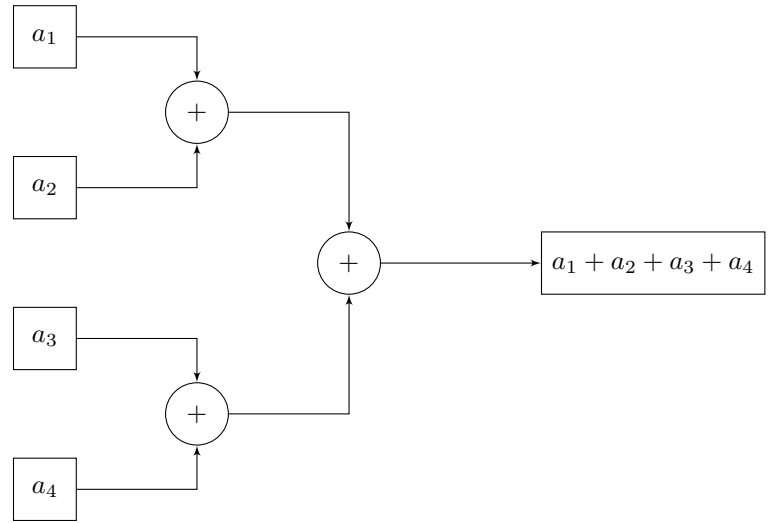
\begin{figure*}[h!]
\begin{center}
\begin{tikzpicture}[auto, node distance=1cm,>=latex', font=\footnotesize]
    \node [block](a1){$a_1$};
    \node [coordinate, below of=a1](p1){};
    \node [block, below of=p1](a2){$a_2$};
    \node [coordinate, below of=a2](p2){};
    \node [block, below of=p2](a3){$a_3$};
    \node [coordinate, below of=a3](p3){};
    \node [block, below of=p3](a4){$a_4$};
    \node [sum, right of=p1, node distance=2cm](s1){$+$};
    \node [sum, right of=p3, node distance=2cm](s2){$+$};
    \node [sum, right of=p2, node distance=4cm](s3){$+$};
    \node [block, right of=s3, node distance=4cm](output){$a_1+a_2+a_3+a_4$};
    
      \draw [->] (a1) -|  (s1);
      \draw [->] (a2) -|  (s1);
      \draw [->] (a3) -|  (s2);
      \draw [->] (a4) -|  (s2); 
      \draw [->] (s1) -| (s3);
      \draw [->] (s2) -| (s3);
      \draw [->] (s3) -- (output);
\end{tikzpicture}
\caption{adder tree for summing four elements}\label{fig:adderTree}
\end{center}
\end{figure*}

In more complicated cases, it is necessary to have a fundamental understanding of the algorithm structure to realize the full potential of parallel computing. Indeed often it may happen that the most efficient algorithms for classic sequential implementation are not necessarily the most efficient for parallel architectures.  For this, the task of this chapter is revising the main steps of the \PGZ and the \BM decoding algorithms studied in chapters 2 and 3  in order to understand when a parallel implementation of these steps can bring some significant advantages decreasing the computational cost of the decoding. Moreover, when it is necessary,  we will modify some steps of the decoding algorithms mentioned above in order to achieve the maximum benefits from their parallelization.\\

\section{Parallel Implementation of \PGZ Decoder}
As seen in section \ref{sect:PGZ},  step \refpgz{pgz2} of the \PGZ decoding algorithm \ref{alg:PGZ} calculates first  $e$ the number of errors that occur, computing for $i=t, t-1, \dots$ the determinants of the matrices
$$A_i=\begin{mymatrix}
                        S_1          & S_2             & \cdots & S_{i} \\
                        S_2          & S_3             & \cdots & S_{i+1}\\
                        \vdots       & \vdots          &        & \vdots \\
                        S_{i}    & S_{i+1}        & \cdots & S_{2i-1}\\
        \end{mymatrix}$$ 
and setting $e:=\max\{ i\leq t\;\;|\; \det(A_i)\neq0\}$, and after calculates the error-locator polynomial $\sigma(x)$ solving the linear system
$$A_e\begin{mymatrix} \sigma_e\\ \sigma_{e-1} \\ \vdots \\ \sigma_1\\
   \end{mymatrix}= - 
   \begin{mymatrix} S_{e+1}\\ S_{e+2} \\ \vdots \\ S_{2e}\\
   \end{mymatrix} $$
   
In this section we will give a parallel implementation of this procedure, exploiting the Laplace expansion for determinants for the calculation of the determinants of the matrices $A_i$'s, which are nothing but the leading  principal minors of the syndrome matrix 
$$A=\begin{mymatrix}
                        S_1          & S_2             & \cdots & S_{t+1} \\
                        S_2          & S_3             & \cdots & S_{t+2}\\
                        \vdots       & \vdots          &        & \vdots \\
                        S_{d-1-t}    & S_{d-t}        & \cdots & S_{d-1}\\
        \end{mymatrix}$$
and exploiting  Horiguchi's formula (theorem \ref{thm:horiguchi}) for finding the error-locator polynomial using the minors already computed during the calculation of $e$.
More generally, we will use the properties of the minors of the syndrome matrix, which have been proved in section \ref{sect:horiguchi}, to also express the error values in terms of minors already computed, modifying so  the implementation of the step \refpgz{pgz4}.

For this, we introduce the following tools:
\begin{defi} \label{defi:D}
Let $i$ be an index less or equal to $t$. We consider the set $\mathcal{N}_i$ defined by
$$\mathcal{N}_i=\left\{\vct{j}=(j_1,j_2,\dots,j_i)\in\N^i \;|\; 1\leq j_1<j_2<\cdots<j_i\leq t\right\}$$
For any $\vct{j}=(j_1,j_2,\dots,j_i)\in\mathcal N_i$ we define $\myd^{(i)}_{\vct{j}}$ as the determinant of the $i\times i$ minor of $A$ given by the first $i$ rows and the columns $j_1,j_2, \dots, j_i$. That is
$$\myd^{(i)}_{\vct{j}}=\det\begin{mymatrix}
                         S_{j_1} & S_{j_2} &\cdots & S_{j_i}\\
                         S_{j_1+1} & S_{j_2+1} &\cdots & S_{j_i+1}\\
                         \vdots &\vdots & &\vdots\\
                         S_{j_1+i-1} & S_{j_2+i-1} &\cdots & S_{j_i+i-1}\\
                        \end{mymatrix}$$
\end{defi}

We now observe that using the Laplace expansion for determinants we have that:
\begin{equation}\label{eq:D}
\myd^{(i)}_{(j_1,j_2,\dots,j_i)}=\sum_{l=1}^i S_{j_l+i-1}(-1)^{i+l}\myd^{(i-1)}_{(j_1,\dots,j_{l-1},j_{l+1},\dots,j_i)}
\end{equation}
for any $(j_1,j_2,\dots,j_i)\in\mathcal N_i$ and for any $i\in\{1,2,\dots,t\}$.
\begin{rem}
We recall that the elements $k^{(i)}_j$'s are defined in section \ref{sect:horiguchi} as
$$k^{(i)}_j=(-1)^{i+j}\begin{mymatrix}
                         S_1          & S_2             & \cdots & S_{i+1} \\
                         S_2          & S_3             & \cdots & S_{i+2}\\
                         \vdots       & \vdots          &        & \vdots \\
                         S_{i}    & S_{i+1}        & \cdots & S_{2i}\\
         \end{mymatrix}$$
for any $i\geq1$, $j\in\{0,1,\dots,i\}$ and we note that they are a subset of the elements $\myd^{i}_{\vct{j}}$'s. More precisely for any $i\geq1$, $j\in\{0,1,\dots,i\}$ we have that:
\begin{equation}\label{eq:k_and_d}
k^{(i)}_j=(-1)^{i+j}\myd^{(i)}_{\pt{1,\dots,j,j+2,\dots,i+1}}
\end{equation}
\end{rem}

The next proposition shows that we can express the number of errors that occur, the error-locator polynomial and the error values in terms of the minors $\myd^{(e)}_{\vct{j}}$'s.
\begin{prop}\label{prop:pPGZ1}
 If $e\leq t$, then
 \begin{enumerate}[\ \ \ (1)]
  \item $e=\min\left\{i\in\{1,2,\dots,t\}\,|\,  \myd^{(i+1)}_{\vct{j}}=0 \;\;\forall\, \vct{j}\in\mathcal N_{i+1}\right\}$;
  \item $\sigma(x)= \displaystyle\frac{1}{\myd^{(e)}_{(1,2,\dots,e)}}\left(\sum_{j=0}^{e}(-1)^{e+j}\myd^{(e)}_{(1,    \dots,j,j+2,\dots,e+1)}\,x^{e-j}\right)$;
  \item For any $i\in\{1,2,\dots,e\}$ we have that
   $$E_i=-\displaystyle\frac{\myd^{(e)}_{(1,2,\dots,e)} \left(X_i^{-1}\right)^{2(e-1)}}{\sigma'(X_i^{-1})\widehat{\omega}(X_i^{-1})}$$
  where $\displaystyle\widehat{\omega}(x)=\sum_{j=0}^{e-1} (-1)^{e-1+j}\myd^{(e-1)}_{(1,\dots,j,j+2,\dots,e)}\,x^{e-1-j}$
 \end{enumerate} 
\end{prop}
\begin{proof}
 To prove (1) we define $$m=\min\left\{i\in\{1,2,\dots,t\}\,|\,  \myd^{(i+1)}_{\vct{j}}=0 \;\forall \vct{j}\in\mathcal N_{i+1}\right\}$$ and we will show that $e=m$.
 Recalling proposition \ref{prop:syndromematrix}, we have that
 $$e\leq t\Rightarrow e=\rk(A)\Rightarrow (\myd^{(e+1)}_{\vct{j}}=0\; \;\forall \vct{j}\in \mathcal N_{e+1})\Rightarrow m\leq e$$ 
 Further we note that by (\ref{eq:D}), if $m<e$, then
 $\myd^{(e)}_{(1,2,\dots,e)}=\det(A_e)=0$ and this is absurd (see proposition \ref{prop:syndromematrix}). Thus $e=m$ and (1) is proved.  Using (\ref{eq:k_and_d}), (2) follows immediately by  corollary \ref{cor:kersigma} and  (3) follows from  Horiguchi's formula (theorem \ref{thm:horiguchi}).
\end{proof}

\begin{rem}
 In order to save divisions, we may consider the polynomial
 $$\myhat\sigma(x)\overset{def}{=}\myd^{(e)}_{(1,2,\dots,e)}\cdot\sigma(x)$$
 It is clear that $\myhat\sigma(x)$ has the same roots as the error-locator polynomial $\sigma(x)$ and that it holds
 $$E_i=-\displaystyle\frac{\pt{\myd^{(e)}_{(1,2,\dots,e)}}^2 \left(X_i^{-1}\right)^{2(e-1)}}{\myhat\sigma'(X_i^{-1})\widehat{\omega}(X_i^{-1})}$$
 for any $i\in\{1,2,\dots,e\}$ 
\end{rem}

Finally we deduce from (\ref{eq:D}) that if we already know the minor $\myd^{(i-1)}_{\vct{j}}$ for every $\vct{j}\in\mathcal{N}_{i-1}$, then any minor  $\myd^{(i)}_{(j_1,\dots,j_i)}$ can be calculated with only $i$ independent multiplications and one sum of $i$ elements. Furthermore we observe that the computation of one $\myd^{(i)}_{(j_1,\dots,j_i)}$ is independent of the computation of the others minors $i\times i$. Hence we may think of an iterative parallel algorithm in which the $i$th step computes at the same time all the determinants $\myd^{(i)}_{\vct{j}}$  using (\ref{eq:D}).  We write in details this algorithm, which we will call \emph{parallel Peterson-Gorenstein-Zierler} (\pPGZ) \emph{decoding algorithm}: 

\begin{alg}[\textbf{\pPGZ decoding algorithm} for $\RS$] \hspace{4cm}
\label{alg:pPGZ}
\begin{description}
\item[\texttt{Input}:]the received word $\rw(x)$;
\item[\texttt{Output}:]the codeword $\cw(x)$;
\end{description}
\texttt{Begin}
\begin{enumerate}[\bfseries \pPGZ.1]
        \item \begin{tabbing}
              \textit{(syndrome computation)}\\
               \texttt{for }\= $i=1,2,\dots, d-1$ \texttt{ do}\\
               \> $S_i:=\rw(\alpha^i)$;\\
               \texttt{endfor}               
              \end{tabbing}
        \item \textit{(error-locator polynomial computation)}\\[-1cm]
              \begin{enumerate}[a)]
               \item \begin{tabbing}
                       {\small\textit{(inizialization)}}\\
                       $t:=\pint{\frac{d-1}{2}}$;\\
                       $i:=1$;\\
                       \texttt{whi}\=\texttt{le }   $(i\leq t)$ \texttt{ repeat }\\
                       \>$\myd^{(1)}_i:=S_i$;\\
                       \>$i:=i+1$;\\
                       \texttt{endwhile}
                     \end{tabbing}
               \item \begin{tabbing} 
                     {\small\textit{(iterative parallel procedure)}}\\
         $i:=2$;\\
         \texttt{whi}\=\texttt{le} $(i\leq t)$ \texttt{ do in parallel }\\[0,3cm] 
         \>$\myd^{(i)}_{(j_1,\dots,j_i)}:=\displaystyle\sum_{l=1}^i S_{j_l+i-1}(-1)^{i+l}\myd^{(i-1)}_{(j_1,\dots,j_{l-1},j_{l+1},\dots,j_i)}\;$;\\[0,3cm] 
         \>\texttt{if }\= (all $\myd^{(i)}_{(j_1,\dots,j_i)}$'s are zero)  \texttt{ then }\\[0,4cm]
         \>\>\texttt{if } \=$\displaystyle\left(\myd^{(i-1)}_{(1,\dots,i-1)}\neq0\right)$ \texttt{ then }\\[0.3cm]
         \>\>\>$e:=i-1$;\\
         \>\>\>\texttt{go to step c};\\[0.2cm]
         \>\>\texttt{else declare a malfunction};\\[0.2cm]
         \>\texttt{else} $i:=i+1$;\\[0.2cm]
         \texttt{endwhile}\\
         $e:=t$;
         \end{tabbing}
        \item  $\myhat\sigma(x):= \displaystyle\sum_{j=0}^{e}(-1)^{e+j}\myd^{(e)}_{(1,\dots,j,j+2,             			\dots,e+1)}\,x^{e-j};$\\[0,2cm]
               $\displaystyle\widehat{\omega}(x):=\sum_{j=0}^{e-1} (-1)^{e-1+j}\myd^{(e-1)}_{(1,\dots,j,j+2,\dots,e)}\,x^{e-1-j}$;\\         
         \end{enumerate}
                  
       \item \textit{(finding errors positions)}\\
             \texttt{calculate  the error positions $p_1,p_2,\dots,p_e$ and the elements $X_1^{-1}, X_2^{-1},\dots,X_e^{-1}$ using Chien's search;}\\
             
       \item \begin{tabbing}
             \textit{(finding error values)}\\
             \texttt{for } \= $i=1,2,\dots,e$ \texttt{ do in parallel}\\[0,4cm] 
             \>$E_i:=-\displaystyle\frac{\pt{\myd^{(e)}_{(1,2,\dots,e)}}^2 \left(X_i^{-1}\right)^{2(e-1)}}{\myhat\sigma'(X_i^{-1})\widehat{\omega}(X_i^{-1})}$\\
             \texttt{endfor}\\
             \end{tabbing}
                         
\end{enumerate}
\texttt{Return } $\cw(x):=\displaystyle \rw(x)-\sum_{i=1}^e E_{i}x^{p_i};$\\
\texttt{End}\\
\end{alg}

The correctness of the \pPGZ algorithm is assured by proposition \ref{prop:pPGZ1}. As regards the computational cost, we observe that since
$$\#\{\myd^{(i)}_{\vct{j}}\;\;|\;\; \vct{j}\in\mathcal{N}_{i}\}=\#\,\mathcal{N}_{i}=
\binom{t}{i}$$
in order that the $i$th step  of \pPGZ.2 can calculate the minors $\myd^{(i)}_{\vct{j}}$'s simultaneously for every $\vct{j}\in\mathcal{N}_{i}$, at most $\displaystyle\binom{t}{i}$ circuit (with shared memory) are necessary. Each of these is formed from $i$ multipliers and $i$ adders and  computes one of the minor $\myd^{(i)}_{\vct{j}}$ with:
\begin{itemize}
 \item 1 multiplication  (the multipliers compute simultaneously the products of the type $S_{j_l+i-1}(-1)^{i+l}\cdot\myd^{(i-1)}_{(j_1,\dots,j_{l-1},j_{l+1},\dots,j_i)}$);
 \item $\Pint{\log_2i}$ additions (the products computed are added in an adder tree of maximum depth $\Pint{\log_2t}$);
\end{itemize}
As $i$ goes from 1 to $e$,  the computational cost of \pPGZ.2 is upper bounded by $e$ multiplications and $e\Pint{\log_2e}$ additions using at most
$$\displaystyle\max\left\{\binom{t}{i}\,|\, 1\leq i\leq e\right\}=\binom{t}{\pint{t/2}}$$
circuits. That is 
$$ t\binom{t}{\pint{t/2}}  \text{ multipliers}\quad \quad  t\binom{t}{\pint{t/2}}  \text{ adders}$$
In order that this number is not too big, we have to suppose to be correcting Reed-Solomon codes with a small correction capability. For example if $t\leq8$ we need only some hundreds of multipliers and adders that nowadays can be held in only one microchip (see table \ref{tab:binomiale}). 

\begin{table}[h!]
 \begin{center}
 {
\renewcommand{\arraystretch}{2.1}
 \begin{tabularx}{300pt}{|c|X|X|X|X|X|X|X|}
 \hline
  $t$ & 3 & 4 & 5 & 6 & 7 & 8 & 9\\
  \hline
  $\displaystyle t\binom{t}{\pint{t/2}}$ & 9 & 24 & 50 & 120 & 245 & 560 & 1134\\[0,2cm]
 \hline
 \end{tabularx}}
 \caption{number of circuit elements required by \pPGZ.2}\label{tab:binomiale}
\end{center}
\end{table}


As already seen for step \fPGZ.4b (see section \ref{sect:horiguchiPGZ}), for the implementation of \pPGZ.3 and \pPGZ.4  we suppose that there are 4 circuits for Chien's search that compute all the polynomial evaluations needed to the error positions and error values computation with a negligible computational time complexity compared to the cost of other operations involved. 
After the polynomial evaluations, first the multiplications $\myhat{\sigma}'(X_i^{-1})\cdot\widehat{\omega}(X_i^{-1})$ are executed at the same time by  $e$ multipliers and then the divisions needed to calculate $E_i$'s are simultaneously handled by $e$ dividers. 
The complexity of the main step of the \pPGZ decoding algorithm are summarized in the following table:\\
\begin{table}[h]
 \begin{center}
 {\renewcommand{\arraystretch}{2.3}
 \begin{tabular}{|c|c|c|}
 \hline
    & \emph{time complexity} & \emph{space complexity}\\
  \hline
  \textbf{\pPGZ.2} & \begin{tabular}{cl}
                   $e$ &{\small multiplications}\\[-.3cm]
                 $e\Pint{\log_2e} $ &{\small additions}
             \end{tabular} 
          & \begin{tabular}{ll}
                   $t\binom{t}{\pint{t/2}}$  &{\small multipliers}\\[-.3cm]
                   $t\binom{t}{\pint{t/2}}$  &{\small adders}
             \end{tabular}  
             \\[0,2cm]
 \hline
  \textbf{\pPGZ.4} & \begin{tabular}{ll}
                   $1$ &{\small division}\\[-.3cm]
                   $1$ &{\small multiplication}
             \end{tabular} 
           & \begin{tabular}{cl}
                   $t$  &{\small dividers}\\[-.3cm]
                   $2t$  &{\small multipliers}
             \end{tabular}\\[0,2cm]
 \hline
 \end{tabular}}
 \caption{parallel complexity of \pPGZ.2 and \pPGZ.4 }\label{tab:pPGZ}
\end{center}
\end{table}
 
%

\section{Parallel Implementation of the  \BM decoder}
 We recall the implementation given in section \ref{sect:BMalg} of the \BM decoding algorithm \ref{alg:BM}, in which the error-locator polynomial is calculated computing recursively the polynomials $\sigma^{(1)}(x), \sigma^{(2)}(x),\dots,\sigma^{(d-1)}(x)=\sigma(x)$ (see definition \ref{defi:sigma_omega_i} and corollary \ref{cor:BM2bis}) by the instructions of the step \refbm{BM2} recalled below.
 \begin{description}
  \item[\textit{BM}.2]\textit{(error-locator computation)}\\[-1cm]
            \begin{tabbing} 
             $\sigma^{(0)}:=1$;\\ 
              $\tau^{(0)}:=1$;\\
              $D(0):=0$;  \\
                 \texttt{for} \= $i=0,1,\dots,d-2$ \texttt{ do}\\
                 \>$\Delta_i:\displaystyle=\sum_{j=0}^{D(i)}S_{i+1-j}\,\sigma_j^{(i)}$;\\
                 \>$\sigma^{(i+1)}(x):=\sigma^{(i)}(x)-\Delta_ix\tau^{(i)}(x)$;\hspace{3,5cm}(\textasteriskcentered)\\[0,2cm]
                 \>\texttt{if} \= ($\Delta_i=0$ or $2D(i)\geq i+1$)   \texttt{then}\\ 
                 \>\>$D(i+1):=D(i)$;\\
                 \>\>$\tau^{(i+1)}(x):=x\tau^{(i)}(x)$;\hspace{4,05cm}(\textborn)\\
                 \>\texttt{else}  \= \\
                 \>\>$D(i+1):=i+1-D(i)$;\\
                 \>\>$\tau^{(i+1)}(x):=\frac{\sigma^{(i)}(x)}{\Delta_i}$;\hspace{4cm}(\textborn)\\
                 \texttt{endfor}\\
                 $e:=D(d-1)$;\\
                 $\sigma(x):=\sigma^{(d-1)}(x)$;
               \end{tabbing}
  \end{description}
 
We can note that the polynomial updates (\textasteriskcentered) and (\textborn)  do not constitute obstacles to parallel computation. Indeed as seen we can add two polynomials or multiply a polynomial by a field element with the cost of one multiplication or one addition (with a linear number of multipliers and adders). Thus the instructions (\textasteriskcentered) and (\textborn) can be implemented at same time, after the discrepancy $\Delta_i$ has been computed, with the overall cost of one inversion, one multiplication an one addition.
It is the computation of discrepancy the $\Delta_i$
that constitutes the parallel computing bottleneck. This is because $\Delta_i$ is computed via a sum of $i$ elements (which requires one multiplication but $\Pint{\log_2i}$ additions implemented in a circuit with $i$ multipliers and $i$ adders) and moreover because  it must be computed before the polynomial update (\textasteriskcentered). 

In order to eliminate this bottleneck and reach a greater parallelism, the \BM decoding algorithm \ref{alg:BM} is modified following an idea of Dilip V. Sarwate and Naresh R. Shanbhag in \cite{sarwateP}. 

For this we introduce the following polynomials:
\begin{defi} \label{defi:DeltaTheta}
Let $\sigma^{(i)}(x)$ and $\tau^{(i)}$ the polynomials given in definitions \ref{defi:sigma_omega_i} and \ref{defi:tau_gamma}. If $S(x)=S_1+S_2x+\cdots+s_{d-1}x^{d-2}$ is the syndrome polynomial (definition \ref{defi:syndromes}), then for any $i\in\{0,1,\dots,d-1\}$ we define the polynomials 
\begin{align*}
\Delta^{(i)}(x)&=\sigma^{(i)}(x)S(x)\\
\Theta^{(i)}(x)&=\tau^{(i)}(x)S(x)
\end{align*}
Moreover $\Delta_j^{(i)}$ and $\Theta_j^{(i)}$ will denote respectively the coefficient of the term $x^j$  in $\Delta^{(i)}(x)$ and in $\Theta^{(i)}(x)$.\\
\end{defi}

We recall that the $i$th discrepancy $\Delta_i$ is defined as the coefficient of the term $x^i$ in the polynomial $\sigma^{(i)}(x)S(x)$, hence it is obvious that
$$\Delta_i^{(i)}=\Delta_i$$
for every $i=0,1,\dots,d-2$. 

The following proposition states the  recursive properties of the polynomial families  $\left\{\Delta^{(i)}(x)\right\}_{i=0}^{d-1}$ and $\left\{\Theta^{(i)}(x)\right\}_{i=0}^{d-1}$;
\begin{prop}\label{prop:DeltaTheta}
With the notations of definition \ref{defi:DeltaTheta}, it holds that:
\begin{enumerate}[\ \ \ (1)]
 \item $\Delta^{(0)}(x)=S(x)$ and $\,\Theta^{(0)}(x)=S(x)$;
 \item For any $i\in\{0,1,\dots,d-2\}$
       \begin{equation}\label{rel:Delta}
        \Delta^{(i+1)}(x)=\Delta^{(i)}(x)-\Delta_i^{(i)}x\Theta^{(i)}(x)
       \end{equation}

 \item If $D:\N\rightarrow\N$ is the function stated in definition \ref{defi:D(i)} as $D(0)=0$ and
$$D(i+1)= \begin{cases} D(i)  \quad &\text{ if } \;\Delta_i^{(i)}=0 \text{ or } 2D(i)\geq i+1\\
                      i+1-D(i) \quad &\text{ if } \;\Delta_i^{(i)}\neq0 \text{ and } 2D(i)< i+1
        \end{cases}$$
 then
 \begin{equation}\label{rel:Theta}
  \Theta^{(i+1)}(x)=\begin{cases} x\Theta^{(i)}(x)\quad&\text{ if }\;\Delta_i^{(i)}=0 \text{ or } 2D(i)\geq i+1\\
                 \frac{\Delta^{(i)}(x)}{\Delta_i^{(i)}}\quad&\text{ if }\;\Delta_i^{(i)}\neq0\text{ and }2D(i)<i+1
        \end{cases}
 \end{equation}
 for any $i\in\{0,1,\dots,d-2\}$
 \end{enumerate}
\end{prop}
\begin{proof}
 \begin{enumerate}
  \item Trivial using that $\sigma^{(0)}(x)=\tau^{(0)}(x)=1$ by definition;
  \item From the definition \ref{defi:sigma_omega_i}, it follows that
  \begin{align*}
    \Delta^{(i+1)}(x)&=\sigma^{(i+1)}(x)S(x)=\left(\sigma^{(i)}(x)-\Delta_ix\tau^{(i)}(x)\right)S(x)=\\
                     &=\sigma^{(i)}(x)S(x)-\Delta_i^{(i)}x\tau^{(i)}(x)S(x)=\\
                     &=\Delta^{(i)}(x)-\Delta_i^{(i)}x\Theta^{(i)}(x)
  \end{align*}
  \item From definition \ref{defi:tau_gamma} it follows that if $\Delta_i^{(i)}=0$ or $ 2D(i)\geq i+1$ then
  $$\Theta^{(i+1)}(x)=\tau^{(i+1)}(x)S(x)=x\tau^{(i)}(x)S(x)=x\Theta^{(i)}(x)$$
  Whereas if $\Delta_i^{(i)}\neq0$ and $2D(i)<i+1$ then
  $$\Theta^{(i+1)}(x)=\tau^{(i+1)}(x)S(x)= \frac{\sigma^{(i)}(x)}{\Delta_i^{(i)}}S(x)= \frac{\Delta^{(i)}(x)}{\Delta_i^{(i)}}$$
  \end{enumerate}
\end{proof}

Proposition \ref{prop:DeltaTheta} allows us to calculate the discrepancies $\Delta_1,\Delta_2,\dots,\Delta_{d-2}$ avoiding the sum of several elements as in (\textasteriskcentered). Indeed it is sufficient to implement the update of the polynomial $\Delta^{(i)}(x)$ as described in the previous proposition in order to compute the $(i+1)$th discrepancy at the same time that the polynomial $\sigma^{(i)}(x)$  is been computed.

Moreover we note that for any $j<i$, the coefficients $\Delta_j^{(i)}$ and $\Theta_j^{(i)}$  cannot affect the value of any later discrepancies $\Delta_{i+k}$. Consequently, we need not store the coefficients $\Delta_j^{(i)}$ and $\Theta_j^{(i)}$ for $j<i$ and we can save memory space and circuit elements, defining
\begin{align}
 \widehat\Delta^{(i)}(x)&\overset{def}{=}\sum_{j=0}^{d-2}\Delta_{i+j}^{(i)}\,x^j\label{Deltacap}\\
 \widehat\Theta^{(i)}(x)&\overset{def}{=}\sum_{j=0}^{d-2}\Theta_{i+j}^{(i)}\,x^j\label{Thetacap}
\end{align}
From this it holds that:
\begin{prop}\label{prop:DeltaTheta2}
 If $\widehat\Delta^{(i)}_j$ is the coefficient of $x^j$ in the polynomial $\widehat\Delta^{(i)}(x)$, then
 \begin{enumerate}[\ \ \ (1)]
  \item $\widehat\Delta^{(0)}=S(x)$ and $\;\widehat\Theta^{(0)}(x)=S(x)$;
  \item For any $i\in\{0,1,\dots,d-2\}$, $$\displaystyle\widehat\Delta^{(i+1)}(x) =\pt{\sum_{j=1}^{d-2}\widehat\Delta^{(i)}_jx^{j-1}}-\widehat\Delta^{(i)}_0\widehat\Theta^{(i)}(x)$$
  \item For any $i\in\{0,1,\dots,d-2\}$, 
  $$\widehat\Theta^{(i+1)}(x)=
  \begin{cases} \widehat\Theta^{(i)}(x)\quad&\text{ if }\;\widehat\Delta_0^{(i)}=0 \text{ or } 2D(i)\geq i+1\\
  \dfrac{1}{\widehat\Delta_0^{(i)}}\pt{\displaystyle\sum_{j=1}^{d-2}\widehat\Delta^{(i)}_jx^{j-1}}\quad&\text{ if }\;\widehat\Delta_0^{(i)}\neq0\text{ and } 2D(i)<i+1
        \end{cases}$$
 \end{enumerate}
 \end{prop}
\begin{proof}
 \begin{enumerate}
  \item By (\ref{Deltacap}) and (\ref{Thetacap}) it follows that $\widehat\Delta^{(0)}(x)=\Delta^{(0)}(x)=S(x)$ and $\widehat\Theta^{(0)}(x)=\widehat\Theta^{(0)}(x)=S(x)$.
  \item By (\ref{Deltacap}) and (\ref{rel:Delta}) it follows that 
        \begin{align*}
        \displaystyle\widehat\Delta^{(i+1)}(x)&=\sum_{j=0}^{d-2}\Delta_{i+1+j}^{(i+1)}\,x^j =\\
        &=\sum_{j=0}^{d-2}\pt{\Delta_{i+1+j}^{(i)}-\Delta_{i}^{(i)}\Theta_{i+j}^{(i)}}\,x^j \overset{\overset{\Delta^{(i)}_{i+d-1}=0}{\downarrow}}{=}\\
        &=\sum_{j=1}^{d-2}\Delta_{i+j}^{(i)}\,x^{j-1}-\Delta_{i}^{(i)}\sum_{j=0}^{d-2}\Theta_{i+j}^{(i)}\,x^j=\\
        &=\pt{\sum_{j=1}^{d-2}\widehat\Delta^{(i)}_jx^{j-1}}-\widehat\Delta^{(i)}_0\widehat\Theta^{(i)}(x)
        \end{align*}
        where we used that $\Delta^{(i)}_{i+d-1}=0$ because $$\Deg{\Delta^{(i)}(x)}=\Deg{\sigma^{(i)}(x)}+\Deg{S(x)}\leq i+d-2$$
  \item (3) follows by (\ref{Thetacap}) and (\ref{rel:Theta}) similarly to what already done for (2).
 \end{enumerate}
\end{proof}

Before showing in detail the parallel implementation of the \BM decoding algorithm due to the reformulated discrepancy computation as seen in proposition \ref{prop:DeltaTheta2}, we present in the following proposition the relation between the error values and the coefficients of the polynomial $\widehat\Delta^{(d-1)}(x)$. We will use it to reformulate the step \refbm{BM4}.

\begin{prop}\label{prop:DeltaTheta3}
 Let $\widehat \Delta^{(d-1)}_j$'s be the  coefficients of the polynomial $\widehat \Delta^{(d-1)}(x)$ and let $e=\wt(\e)$ the number of errors which occurred. If $e\leq t$  then for any $i\in\{1,2,\dots,e\}$ it holds that
 \begin{equation} \label{fine}
 E_i=\dfrac{\widehat\Delta^{(d-1)}(X_i^{-1})\cdot\pt{X_i^{-1}}^{d-1}}{\sigma'(X_i^{-1})} 
 \end{equation}
\end{prop}
\begin{proof}
 By corollary \ref{cor:BM2bis} we know that if $e\leq t$ then $\sigma^{(d-1)}(x)=\sigma(x)$. Thus we have that:
 $$\Delta^{(d-1)}(x)=\sigma^{(d-1)}(x)S(x)=\sigma(x)S(x)$$
 and by the key equation (\ref{keyequation}) it follows that 
 $$\Delta^{(d-1)}(x)\equiv \omega(x) \Mod{x^{d-1}}$$
 Now we observe that $\Deg{\Delta^{(d-1)}(x)}=\Deg{\sigma(x)}+\Deg{S(x)}=e+d-2$, while $\Deg{\omega(x)}\leq e-1$. Thus we conclude that $\Delta^{(d-1)}_{e+d-2+j}=0$ for any $j\leq1$ and that the previous congruence implies:
 \begin{align*}
  \Delta^{(d-1)}(x)  &=\omega(x)+\pt{\Delta^{(d-1)}_{d-1}+\Delta^{(d-1)}_{d}x+\cdots+\Delta^{(d-1)}_{e+d-2}x^{e-1}}x^{d-1}=\\[0,2cm]
  &=\omega(x)+\widehat\Delta^{(d-1)}(x)\,x^{d-1}
 \end{align*}
 Hence for any $i\in\{1,2,\dots,e\}$ we have that
 $$0=\Delta^{(d-1)}(X_i^{-1})=\omega(X_i^{-1})+\widehat\Delta^{(d-1)}(X_i^{-1})\cdot\pt{X_i^{-1}}^{d-1}$$
 and substituting $\omega(X_i^{-1})=-\widehat\Delta^{(d-1)}(X_i^{-1})\cdot\pt{X_i^{-1}}^{d-1}$ in  Forney's formula (\ref{forney}) we can conclude the proof.\\
\end{proof}
\vspace*{1cm}

We state now the \emph{parallel Berlekamp-Massey} (\pBM) \emph{decoding algorithm:}

\begin{alg}[\textbf{\pBM decoding algorithm} for $\RS$]\hspace{2cm}
\label{alg:pBM}
\begin{description}
\item[\texttt{Input}:] the received word $\rw(x)$;
 \item[\texttt{Output}:] the  codeword $\cw(x)$;
\end{description}
\texttt{Begin}
\begin{enumerate}[\bfseries \pBM.1]
        \item \label{pBM1}
              \textit{(syndrome computation)}\\[-1cm]
              \begin{tabbing}
               \texttt{for }\= $i=1,2,\dots, d-1$ \texttt{ do in parallel}\\
               \>$S_i:=\rw(\alpha^i)$;\\
               \texttt{endfor}               
              \end{tabbing}
        \item \label{pBM2}
              \textit{(error-locator polynomial computation)}\\
              $\sigma^{(0)}:=1$;\\ 
              $\tau^{(0)}:=1$;\\
              $\widehat\Delta^{(0)}:=S(x)$;\\
              $\widehat\Theta^{(0)}:=S(x)$;\\
              $D(0):=0$;                   
               \begin{tabbing} 
                 \texttt{for} \= $i=0,1,\dots,d-2$ \texttt{ do in parallel}\\[0,2cm]
                 \>$\sigma^{(i+1)}(x):=\sigma^{(i)}(x)-\widehat\Delta_0^{(i)}x\tau^{(i)}(x)$;\\[0,2cm]
                 \>$\displaystyle\widehat\Delta^{(i+1)}(x) :=\pt{\sum_{j=1}^{d-2}\widehat\Delta^{(i)}_jx^{j-1}}-\widehat\Delta^{(i)}_0\widehat\Theta^{(i)}(x)$;\\[0,2cm]                                    
                 \>\texttt{if} \= ($\widehat\Delta_0^{(i)}=0$ or $2D(i)\geq i+1$)   \texttt{then}\\[0,2cm] 
                 \>\>$D(i+1):=D(i)$;\\
                 \>\>$\tau^{(i+1)}(x):=x\tau^{(i)}(x)$;\\
                 \>\>$\widehat\Theta^{(i+1)}(x):=\widehat\Theta^{(i)}(x)$;\\                                                                                                                                                                                                                                                               
                 \>\texttt{else}  \= \\
                 \>\>$D(i+1):=i+1-D(i)$;\\[0,2cm]
                 \>\>$\displaystyle\tau^{(i+1)}(x):=\frac{\sigma^{(i)}(x)}{\widehat\Delta_0^{(i)}}$;\\[0,2cm]
                 \>\>$\displaystyle\widehat\Theta^{(i+1)}(x):=\dfrac{1}{\widehat\Delta_0^{(i)}}\pt{\displaystyle\sum_{j=1}^{d-2}\widehat\Delta^{(i)}_jx^{j-1}}$;\\                                                                                                                                                                                                                                                        
                 \texttt{endfor}\\[0,2cm]
                 $e:=D(d-1)$;\\
                 $\sigma(x):=\sigma^{(d-1)}(x)$;\\
               \end{tabbing}
       \item \label{pBM3}
	     \textit{(finding error positions)}\\
             \texttt{calculate  the error positions $p_1,p_2,\dots,p_e$ and the elements $X_1^{-1}, X_2^{-1},\dots,X_e^{-1}$ using Chien's search;}\\[0,5cm]
       \item \label{pBM4}
             \textit{(finding the error values)}\\[-1,2cm]
             \begin{tabbing}
             \texttt{for }\= $i=1,2,\dots,e$ \texttt{ do in parallel}\\[0,2cm]
             \>$E_i:=\dfrac{\widehat\Delta^{(d-1)}(X_i^{-1})\cdot\pt{X_i^{-1}}^{d-1}}{\sigma'(X_i^{-1})}$;\\
             \texttt{endfor}
             \end{tabbing}
\end{enumerate}
\texttt{Return } $\cw(x):=\displaystyle \rw(x)-\sum_{i=1}^e E_{i}x^{p_i};$\\
\texttt{End}
\end{alg}

The correctness of the \pBM   decoding algorithm follows by propositions \ref{prop:DeltaTheta2} and \ref{prop:DeltaTheta3}.
The instructions of step \pBM.2 can be implemented in a systolic architecture  composed of an array of $d+t$ circuit of the same type, which store and update the polynomial coefficients, and of a control unit, which computes the function $D$ and determines whether the discrepancy is zero computing its inverse if necessary (see figure \ref{fig:pBM2}).

\begin{figure*}[h!]
\begin{center}
\begin{tikzpicture}[auto, node distance=2cm,>=latex']
    \node [coordinate](pall0){};
    \node [block, right of=pall0, node distance=98mm](control){\footnotesize CONTROL};
    \node [block, below of=pall0, text width=1cm](s1){$c_0$};
    \node [block, right of=s1,text width=1cm](s2){$c_1$};
    \node [coordinate, right of=s2](spazio1){};
    \node [block, right of=spazio1,text width=1cm](s3){$c_{d-3}$};
    \node [block, right of=s3,text width=1cm](s4){$c_{d-2}$};
    \node [pall, below of=s1,node distance=3cm] (pall1){};
    \node [pall, below of=s2, node distance=3cm] (pall2){};
    \node [pall, below of=s3, node distance=3cm] (pall3){};
    \node [pall, below of=s4, node distance=3cm] (pall4){};
    \node [pall, right of=pall4, node distance=6cm](pall5){};
    \node [block, below of=pall1, text width=1cm] (p1){$c_{d-1}$};
    \node [block, right of=p1,text width=1cm](p2){$c_{d}$};
    \node [coordinate, right of=p2](spazio2){};
    \node [block, right of=spazio2,text width=1cm](p3){$c_{d+t-2}$};
    \node [block, right of=p3,text width=1cm](p4){$c_{d+t-1}$};
    \node [pall, below of=p1,node distance=3cm] (pall6){};
    \node [pall, below of=p2, node distance=3cm] (pall7){};
    \node [pall, below of=p3, node distance=3cm] (pall8){};
    \node [pall, below of=p4, node distance=3cm] (pall9){};
    \node [pall, right of=pall9, node distance=6cm](pall10){};
     \draw [->] (s2) --  (s1);
     \draw [->] (s4) --  (s3);
     \draw [->] (p2) --  (p1);
     \draw [] (control) --(pall10);
     \draw []  (pall10) -- (pall6);
     \draw []  (pall1) -- (pall5);
     \draw [->] (p4) -- (p3); 
     \draw [->, dashed](p3) -- (p2);
     \draw [->, dashed](s3)-- (s2);
     \draw [->] (pall1) -- (s1);
     \draw [->] (pall2) -- (s2);
     \draw [->] (pall3) -- (s3);
     \draw [->] (pall4) -- (s4);
     \draw [->] (pall6) -- (p1);
     \draw [->] (pall7) -- (p2);
     \draw [->] (pall8) -- (p3);
     \draw [->] (pall9) -- (p4);
     \draw [] (s1) -- node{{\scriptsize $\widehat\Delta_0^{(i)}$}}(pall0);
     \draw [->] (pall0) -- (control);
\end{tikzpicture}
\caption{architectures for \pBM.2}\label{fig:pBM2}
\end{center}
\end{figure*}
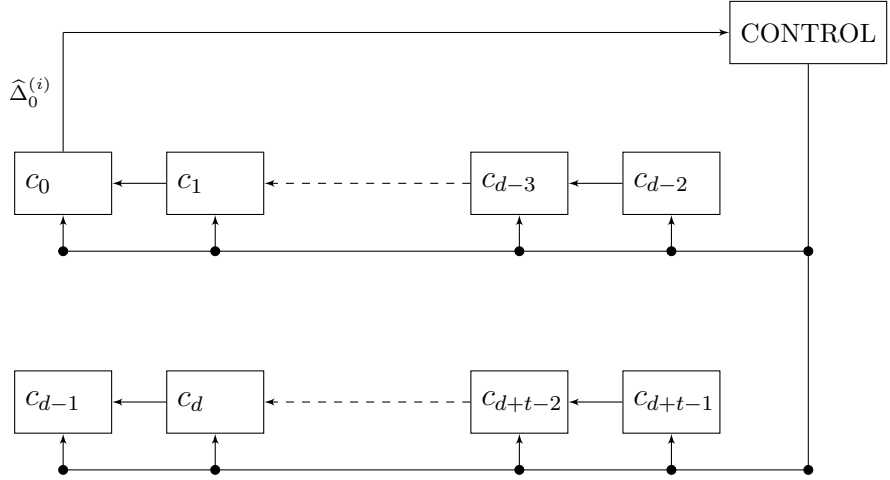

More precisely, each circuit is composed of two storage devices, two multipliers and an adder. After $i$ steps the circuit $c_j$ for $j=0,1,\dots,d-2$ contains the coefficients $\widehat\Delta^{(i)}_j$ and $\widehat\Theta^{(i)}_j$, while for $j=d-1,d, \dots, d+t-1$ stores the coefficients $\sigma^{(i)}_{d+t-1-j}$ and $\tau^{(i)}_{d+t-1-j}$. During the $(i+1)$th step, the coefficients contained in each circuit shift left and the circuit $c_j$ carries out the following instructions to update the polynomial coefficient which it contains:
\begin{align*}
&\widehat\Delta^{(i+1)}_j=\widehat\Delta^{(i)}_{j+1}-\widehat\Delta^{(i)}_0\widehat\Theta^{(i)}_j \\[0,2cm]
&\widehat\Theta^{(i+1)}_j=
  \begin{cases} \widehat\Theta^{(i)}_j\quad&\text{ if }\;\widehat\Delta_0^{(i)}=0 \text{ or } 2D(i)\geq i+1\\
  \dfrac{\widehat\Delta^{(i)}_{j+1}}{\widehat\Delta_0^{(i)}}\quad&\text{ if }\;\widehat\Delta_0^{(i)}\neq0\text{ and } 2D(i)<i+1
       \end{cases}\\
&\text{for } j=0,1,\dots,d-2       
\end{align*}
or        
\begin{align*}
&\sigma^{(i+1)}_j=\sigma^{(i)}_{j}-\widehat\Delta^{(i)}_0\tau^{(i)}_{j-1} \\[0,2cm]
&\tau^{(i+1)}_j=
  \begin{cases} \tau^{(i)}_{j-1}\quad&\text{ if }\;\widehat\Delta_0^{(i)}=0 \text{ or } 2D(i)\geq i+1\\
  \dfrac{\sigma^{(i)}_{j}}{\widehat\Delta_0^{(i)}}\quad&\text{ if }\;\widehat\Delta_0^{(i)}\neq0\text{ and } 2D(i)<i+1
        \end{cases}\\
&\text{for } j=d-1,d, \dots, d+t-1        
\end{align*}
Step \pBM.4  requires 2 circuits for Chien's search in order to evaluate the polynomials $x^{d-1}\widehat{\Delta}^{(d-1)}(x)$ and $\sigma'(x)$ in the field elements. As usual these polynomial evaluations are executed at the same time of the evaluation
of $\sigma(x)$ done in step \pBM.3 with a negligible cost with respect to the others instructions. 
Following the evaluations, the error values $E_i$'s are computed executing simultaneously the divisions needed with $e$ dividers. Thus, recalling that $t=\pint{\frac{d-1}{2}}$, the complexity of step \pBM.2 and \pBM.4 can be summarized as in table \vref{tab:pBM}. Note that  we consider the upper bound of the quantities involved, which is reached when $d$ is even and $d-1=2t+1$. When $d$ is odd and $d-1=2t$, the number of arithmetic operations needed decreases to  $2t$, while $6t+2$ multipliers and $3t+1$ adders are enough.

\begin{table}[h]
 \begin{center}
 {\renewcommand{\arraystretch}{2.1}
 \begin{tabular}{|c|c|c|}
 \hline
    & \emph{time complexity} & \emph{space complexity}\\
  \hline
  \textbf{\pBM.2} & \begin{tabular}{ll}
                   $2t+1$ &{\small inversions}\\[-.3cm]
                   $2t+1$ &{\small multiplications}\\[-.3cm]
                   $2t+1$ &{\small additions}
             \end{tabular} 
          & \begin{tabular}{cl}
                   1  &{\small inversion circuit}\\[-.3cm]
                   $6t+4$  &{\small multipliers}\\[-.3cm]
                   $3t+2$  &{\small adders}
             \end{tabular}   \\[0,2cm]
 \hline
  \textbf{\pBM.4} & \begin{tabular}{ll}
                   $1$ &{\small division}
             \end{tabular} 
           & \begin{tabular}{ll}
                   $t$  &{\small dividers}
             \end{tabular}\\[0,2cm]
 \hline
 \end{tabular}}
 \caption{parallel complexity of \pBM.2 and \pBM.4 }\label{tab:pBM}
\end{center}
\end{table}

\begin{rem}\label{rem:inversionless}
As seen at the end of section \ref{sect:BMalg}, the \BM decoding algorithm \ref{alg:BM} can be modified to avoid the inversions computed in step \BM.2. If in  this section we consider the polynomials $\myhat\sigma^{(i)}(x)$ and $\myhat\tau^{(i)}(x)$ (see definition \ref{defi:inversionless}) instead of the polynomials $\sigma^{(i)}(x)$ and $\tau^{(i)}(x)$, then in step \pBM.2 
we will have the following instructions 
\begin{enumerate}[\bfseries \pBM.2b]
       \item 
              $\myhat\sigma^{(0)}:=1$;\\ 
              $\myhat\tau^{(0)}:=1$;\\
              $\myhat\Delta^{(0)}:=S(x)$;\\
              $\myhat\Theta^{(0)}:=S(x)$;\\
              $D(0):=0$;\\ $\beta(0):=1$;                   
               \begin{tabbing} 
                 \texttt{for} \= $i=0,1,\dots,d-2$ \texttt{ do in parallel}\\[0,2cm]
                 \>$\myhat\sigma^{(i+1)}(x):=
                 \beta(i)\myhat\sigma^{(i)}(x)-\myhat\Delta_0^{(i)}x\myhat\tau^{(i)}(x)$;\\[0,2cm]
                 \>$\displaystyle\myhat\Delta^{(i+1)}(x): =\beta(i)\pt{\sum_{j=1}^{d-2}\myhat\Delta^{(i)}_jx^{j-1}}-\myhat\Delta^{(i)}_0\myhat\Theta^{(i)}(x)$;\\[0,2cm]                                    
                 \>\texttt{if} \= ($\myhat\Delta_0^{(i)}=0$ or $2D(i)\geq i+1$)   \texttt{then}\\[0,2cm] 
                 \>\>$D(i+1):=D(i)$;\\
                 \>\>$\beta(i+1):=\beta(i)$;\\
                 \>\>$\myhat\tau^{(i+1)}(x):=x\myhat\tau^{(i)}(x)$;\\
                 \>\>$\myhat\Theta^{(i+1)}(x):=\myhat\Theta^{(i)}(x)$;\\                                                                                                                                                                                                                                                               
                 \>\texttt{else}  \= \\
                 \>\>$D(i+1):=i+1-D(i)$;\\
                 \>\>$\beta(i+1):=\myhat\Delta_0^{(i)}$\\
                 \>\>$\myhat\tau^{(i+1)}(x):=\myhat\sigma^{(i)}(x)$;\\[0.2cm]
                 \>\>$\myhat\Theta^{(i+1)}(x):=\displaystyle\sum_{j=1}^{d-2}\myhat\Delta^{(i)}_jx^{j-1}$;\\                                                                                                                                                                                                                                                        
                 \texttt{endfor}\\
                 $e:=D(d-1)$;\\
                 $\myhat\sigma(x):=\myhat\sigma^{(d-1)}(x)$;
               \end{tabbing}
 \end{enumerate}

In this way each discrepancy inversion  is replaced by multiplications (of the coefficients of the polynomials involved) executable at the same time the other multiplications already present (see \cite{sarwateP} for more details). So the polynomials $\myhat{\sigma}^{(d-1)}(x)=b\cdot \sigma^{(d-1)}$ and $\myhat{\Delta}^{(d-1)}(x)=b\cdot\widehat{\Delta}^{(d-i)}(x)$, where
$$b=\prod_{j=0}^{d-2}\beta(j)$$
can be computed by the circuit $c_j$ with a  complexity upper bounded by:\\
\begin{table}[h!]
 \begin{center}
 {\renewcommand{\arraystretch}{2.1}
 \begin{tabular}{|c|c|c|}
 \hline
   & \emph{time complexity} & \emph{space complexity}\\
  \hline
   \textbf{\pBM.2b} &\begin{tabular}{ll}
                   $2t+1$ &{\small multiplications}\\[-.3cm]
                   $2t+1$ &{\small additions}
             \end{tabular} 
          & \begin{tabular}{cl}
                   $6t+4$  &{\small multipliers}\\[-.3cm]
                   $3t+2$  &{\small adders}
             \end{tabular}   \\[0,2cm]
 \hline
 \end{tabular}}
 \caption{parallel complexity of inversionless \pBM.2}\label{tab:INpBM}
\end{center}
\end{table}
\end{rem}
Finally we note that the error value formula (\ref{fine}) used in \pBM.4 does not change, indeed it holds that
$$E_i=\dfrac{\myhat\Delta^{(d-1)}(X_i^{-1})\cdot\pt{X_i^{-1}}^{d-1}}{\myhat\sigma'(X_i^{-1})} $$

\backmatter

\chapter{Conclusions}

In this thesis we have studied and proposed several decoding algorithms for Reed-Solomon codes, dwelling on their computational time complexity. In particular, we have proved that, when the decoding is expressed in terms of linear systems and tools of linear algebra,  a detailed study of the matrices involved (the syndrome matrix $A$, the matrices $A_i$ and $B_j$) leads to efficient procedures for computing both the error-locator polynomial and the error values. This permits to see the linear algebra approach in a new light, since it is now competitive with other decoding strategies. 
Finally the linear algebra techniques allow to reach the goal of a parallel implementation of the \PGZ decoding strategy.
We note  that the computational time cost is lower for the \pPGZ  decoding algorithm \ref{alg:pPGZ} than for the \pBM   decoding algorithm \ref{alg:pBM} (proposed in \cite{sarwateP}),
but the second one allows a parallel implementation that employs a smaller number of circuit elements arranged in a simpler systolic architecture (see table \vref{tab:comparision}). Thus the \pPGZ  decoding algorithm \ref{alg:pPGZ} can represent the better choice for Reed-Solomon codes in some special cases, as when the error correction capability is small and the transmission channel is quite good,  while the \pBM   decoding algorithm \ref{alg:pBM} allows a linear time decoding with a linear number of hardware elements also in more general cases.


As intermediate results we have proved the following:
\begin{itemize}
%
 \item from a theoretical point of view, we have used the \fPGZ decoding algorithm \ref{alg:fastPGZ} to uncover the existing relationship  between  the leading principal minors $A_i$ of the syndrome matrix and the discrepancies $\Delta_j$ computed by the \BM decoding algorithm \ref{alg:BM}. Indeed in theorem \ref{thm:comparison} we have proved that if $\det(A_i)\neq0$, then
 $$r=\min\{j\,|\,\det(A_{i+j})=0\}  \Rightarrow  \Delta_{2i}=\cdots=\Delta_{2i+r-2}=0 
 \text{ and }\Delta_{2i+r-1}\neq0 $$
 We have achieved this result by comparing the intermediate outcomes of the \fPGZ decoding algorithm \ref{alg:fastPGZ} with the ones of the \BM decoding algorithm \ref{alg:BM};
   
 \item we have found the necessary and sufficient conditions that must be added to the \fPGZ and the \BM decoding algorithms in order to make them $t$-bounded distance decoding algorithms;
 
 \item for what concerns the error value computation, we have proved a new formula for the \BM decoder (used in \BM.4b), which  needs fewer arithmetic operations than Forney's formula. Moreover we have noted that step \BM.4 (or \BM.4b)  allows an advantageous pipelined implementation with step \BM.3. In step \fPGZ.4 this is not possible because the error values are calculated by solving a linear system in which the coefficient matrix is formed by some powers of the $\sigma(x)$ roots.  We have eliminated this disadvantage of the \fPGZ decoding algorithm by proving in proposition \ref{prop:HoriPGZ} that
 $$E_i=-\dfrac{\varepsilon_{e-\theta}\cdot \displaystyle\left( X_i^{-1}\right)^{e+\theta-1}}{\displaystyle\sigma'(X_i^{-1})P_{\vct{w}^{(\theta)}}( X_i^{-1})}$$
 This formula allows to compute each $E_i$ using the vector $\vct{w}^{(\theta)}$ and coefficient $\varepsilon_{e-\theta}$, which are byproducts of the computation of $\sigma(x)$ in \fPGZ.2. Moreover it can be executed by a pipelined implementation with step \fPGZ.3;
 
 \item as regards the parallel implementation, we have proved that the proposed \pPGZ  decoding algorithm \ref{alg:pPGZ} has an $O(e)$ multiplicative time complexity with $O(t\cdot\binom{t}{\pint{t/2}})$ circuit elements. Moreover we have studied the \pBM   decoding algorithm \ref{alg:pBM}, which has an $O(t)$ multiplicative time complexity with $O(t)$ circuit elements. For the second one our main contribution concerns the formalization of the proof of the algorithm correctness.\\
 \begin{table}[h!]
 \begin{center}
 {\renewcommand{\arraystretch}{2.1}
 \begin{tabular}{|c|c|c|}
 \hline
   & \emph{time complexity} & \emph{space complexity} \\
  \hline
   \textbf{\pPGZ} &  \begin{tabular}{cl}
                 $1$ &{\small divisions}\\[-.3cm]
                   $e+1$ &{\small multiplications}\\[-.3cm]
                 $e\Pint{\log_2e} $ &{\small additions}
             \end{tabular} 
         & \begin{tabular}{cl}
                  $t$ &{\small dividers}\\[-.3cm]
                   $t\binom{t}{\pint{t/2}}+2t$  &{\small multipliers}\\[-.3cm]
                   $t\binom{t}{\pint{t/2}}$  &{\small adders}
             \end{tabular} \\[0,2cm]
 \hline
 \textbf{\pBM} & \begin{tabular}{cl}
                 $1$ &{\small divisions}\\[-.3cm]
                  $2t+1$ &{\small multiplications}\\[-.3cm]
                   $2t+1 $ &{\small additions}
             \end{tabular} 
      & \begin{tabular}{cl}
                  $t$ &{\small dividers}\\[-.3cm]
                  $6t+4$ &{\small multipliers}\\[-.3cm]
                   $3t+2$  &{\small adders}
             \end{tabular} \\[-1.3cm]
            {\footnotesize (inversionless)} & &\\[0,2cm]
 \hline
 \end{tabular}}
 \caption{complexity comparison}\label{tab:comparision}
\end{center}
\end{table}

\end{itemize}

\addcontentsline{toc}{chapter}{\bibname}
\nocite{*}
\bibliographystyle{amsalpha}
\bibliography{bibliografia}

\end{document}